\documentclass[a4paper, 11pt, reqno]{amsart}

\usepackage[textheight=23 cm, textwidth=15.5cm, headheight=50pt, headsep=10pt, footskip=32pt, centering]{geometry}
\usepackage{dsfont}
\usepackage{cases}
\usepackage{comment}
\usepackage{float}
\usepackage{nicefrac,caption}
\usepackage{algorithm,algorithmic}
\usepackage[colorlinks=true,linkcolor=black,citecolor=black]{hyperref}
\usepackage[absolute]{textpos}
\usepackage{placeins}
\usepackage{soul}
\usepackage{yhmath}
\usepackage{amssymb,mathrsfs,xcolor,bbold,comment,enumerate}
\numberwithin{equation}{section}
\usepackage[onehalfspacing]{setspace}
\usepackage{graphicx}
\usepackage{subcaption}

\theoremstyle{plain}
\newtheorem{theorem}{Theorem}[section]

\newtheorem*{corollary*}{Corollary}
\newtheorem*{theorem*}{Theorem}
\newtheorem{proposition}[theorem]{Proposition}

\newtheorem{lemma}[theorem]{Lemma}

\theoremstyle{definition}
\newtheorem{assumption}[theorem]{Assumption}
\newtheorem{remark}[theorem]{Remark}

\newtheorem{example}[theorem]{Example}
\newtheorem{definition}[theorem]{Definition}

\newcommand{\mf}{\mathfrak}

\DeclareMathOperator{\interior}{\textnormal{int}}

\newcommand{\R}{\mathbb{R}}

\newcommand{\E}{\mathbb{E}}

\newcommand{\N}{\mathbb{N}}
\newcommand{\dom}{\textnormal{dom}}

\newcommand{\eps}{\varepsilon}

\newcommand{\CC}{\mathcal C}

\newcommand{\CS}{{\mathcal S}}

\newcommand{\ph}{\varphi}

\newcommand*\diff{\mathop{}\!\mathrm{d}}

\DeclareMathOperator{\symbolOperatorExpectedShortfall}{ES}

\newcommand{\symbolExpectedShortfall}[1]{\symbolOperatorExpectedShortfall_{#1}}

\newcommand{\expectedShortfall}[2]{\symbolExpectedShortfall{#1}(#2)}

\title{Multi-asset return risk measures}

\author[C.\ Laudag\'e]{Christian Laudagé\,$^{1,\ast}$}
\author[F.-B.\ Liebrich]{Felix-Benedikt Liebrich\,$^{2,\dagger}$}
\author[J.\ Sass]{Jörn Sass\,$^{1,\ddagger}$}
\thanks{$^1$ Department of Mathematics, RPTU Kai\-sers\-lautern-Landau, Gottlieb-Daimler-Straße 47, 67663 Kaiserslautern, Germany}
\thanks{$^2$ Amsterdam School of Economics, University of Amsterdam, Roetersstraat 11, 1018 WB Amsterdam, Netherlands}

\thanks{$^{\ast}$  \textit{E-mail address:} \texttt{christian.laudage@rptu.de} (corresponding author)}
\thanks{$^\dagger$ \textit{E-mail address:} \texttt{f.b.liebrich@uva.nl}}
\thanks{$^\ddagger$ \textit{E-mail address:} \texttt{joern.sass@rptu.de}}

\date{\today}

\begin{document}

\setstretch{1.15}

\begin{abstract}
    We revisit the recently introduced concept of return risk measures (RRMs) and extend it by incorporating risk management via multiple so-called eligible assets.
    The resulting new class of risk measures, termed multi-asset return risk measures (MARRMs), introduces a novel economic model for multiplicative risk sharing.
    We analyze  properties of these risk measures.
    In particular, we prove that a positively homogeneous MARRM is quasi-convex if and only if it is convex.
    Furthermore, we provide conditions to avoid inconsistent risk evaluations. 
    Then, we point out the connection between MARRMs and the well-known concept of multi-asset risk measures (MARMs).
    This is used to obtain various dual representations of MARRMs. 
    Moreover, we conduct a series of case studies, in which we use typical continuous-time financial markets and different notions of acceptability of losses to compare RRMs, MARMs, and MARRMs and draw conclusions about the cost of risk mitigation.

    \medskip
	
    \item[\hskip\labelsep\scshape Keywords:] Multi-asset risk measure, return risk measure, risk sharing, dual representations, Average-Return-at-Risk
    
    \medskip
    
    \item[\hskip\labelsep\scshape JEL Classification:] D81, G10, G32
\end{abstract}

\maketitle

\section{Introduction}

Classical \textit{monetary} risk measures map uncertain monetary amounts to a single key figure that expresses the riskiness of said financial position. 
Recently, an axiomatic foundation for risk measures has been suggested which focuses on {\em relative} positions instead of absolute monetary amounts, so-called {\em return risk measures} (RRMs). 
The goal of this paper is to enrich this framework with additional flexibility for risk managers,
allowing them to use multiple eligible assets to mitigate unacceptable risk concentration. Further, we elaborate on some mathematical difficulties and interpretations behind return risk measures. 

RRMs have been introduced in Bellini et al.~\cite{Return} and studied further in Laeven \& Rosazza Gianin~\cite{Laeven1}. In the latter paper, the authors focus on return risk measures that are quasi-logconvex, a property that is not taken into account in~\cite{Return}. Among other contributions, dual representations and several concrete examples for such RRMs are developed. Under the stronger condition of log-convexity, called geometrical convexity (or GG-convexity for short), Aygün et al.~\cite{Aygun} also derive a dual representation.  

RRMs are substantially different from the well-studied class of monetary risk measures, see, e.g., Artzner et al.~\cite{artzner_coherent_1999} and Chapter 4 in the monograph by F\"ollmer \& Schied \cite{follmer_stochastic_2016}. 
As pointed out in~\cite{Return}, a RRM is obtained by applying a monetary risk measure to log-returns. 
Hence, the log-transformation establishes a direct connection between RRMs and monetary risk measures.

There is a notable gap between return risk measures and  classical Markowitz theory introduced in~\cite{markowitz_1952}. 
There, the standard deviation is used to measure the risk of relative values. As an alternative, it is possible to use a monetary risk measure instead of the standard deviation. This has been originally suggested by Jaschke \& K\"uchler~\cite{jaschke_coherent_2001} and analyzed rigorously in the works of Herdegen \& Khan~\cite{herdegen_2022,herdegen_2024}.
In their theory, monetary risk measures are applied to returns, i.e.,~fractions of payoffs. 
In contrast, RRMs arise by applying monetary risk measures to log-returns. 
This approach is standard in the field of time series analysis, see Candia \& Herrera~\cite{candia_2024}, McNeil \& Frey \cite{mcneil_2000}, and Sun \& Li \cite{sun_2021}. 
The reason is  that logarithmic transformations of financial time series data often lead to a more stationary behaviour for which classical time series models, like, e.g.,~ARMA-GARCH, are applicable. So, while the classical theory of monetary risk measures is based on absolute monetary amounts, RRMs describe what happens if log-returns are plugged into monetary risk measures.
To the best of our knowledge, RRMs are the first attempt for an axiomatic analysis of such an application. 

We contribute to this debate by extending the framework of RRMs in the sense that we allow for multiple eligible assets. The philosophy leading to this generalisation is the same as in the case of the so-called {\em multi-asset risk measures} (MARMs), an extension of classical monetary risk measures. This family of risk measures is implicitly introduced in F\"ollmer \& Schied~\cite{follmer_stochastic_2016} and explicitly introduced in Frittelli \& Scandolo~\cite{frittelli_2006,scandolo_2004}. 
Our analysis is motivated by the thorough analysis of MARMs in Farkas et al.~\cite{farkas_2015}. 
Here, the authors provide necessary and sufficient conditions for finiteness and continuity of MARMs and develop dual representation results. 

Let us briefly sketch the idea behind a RRM. 
To do so, assume that an agent faces an uncertain future loss $X$ modelled by a positive random variable.
Here and throughout the paper, we follow actuarial conventions and model losses of such uncertain prospects with positive values. 
Now, the agent considers the log-return of the position with respect to a fictitious initial wealth $x>0$, i.e.,~she considers the random variable $\log\left(\frac{X}{x}\right)$, with $\log$ referring to the natural logarithm. 
Acceptability of such a log-return is defined via a monetary risk measure $\rho$. 
Assuming zero interest, this means that the value $\rho\left(\log\left(\frac{X}{x}\right)\right)$ should be less than or equal to zero. 
From a risk management perspective, it is logical to seek the least amount of initial wealth $x$ for which this is satisfied.
This yields the RRM at $X$,
\begin{align}\label{eq:introduction_rrm}
    \inf\left\{x>0\,;\, \rho\left(\log\left(\tfrac{X}{x}\right)\right)\le 0\right\}.
\end{align}

By a closer inspection of~\eqref{eq:introduction_rrm}, an alternative interpretation emerges:
Assume that there is a riskless bank account with zero interest rate on the market. Then, the criterion in~\eqref{eq:introduction_rrm} compares the log-returns of the loss $X$ and the gain by investing $x$ at the beginning of the time period in the bank account. 
Having more ``secure'' -- eligible -- assets available to generate wealth, which is sufficient to bear the losses represented by $X$, logically leads to the quest for a {\em multi-asset version} of this RRM. 
In other words, we aim to give an agent the ability not only to compare the performance of their actual position $X$ with a riskless bank account, but also to allow them to compare its performance with alternative investment opportunities. For instance, if $X$ represents a short position in a single stock, then it could be of interest to see how this stock performs relative to a specific stock index.
This can be achieved by a slight modification of Equation~\eqref{eq:introduction_rrm}, namely comparing $X$ and a payoff of a portfolio of multiple eligible assets which are traded on the market. 
Such payoffs are themselves modelled by random variables stemming from a set $\mathcal{S}$. 
Since these represent wealth, we also model them as positive: $Z>0$ holds for all $Z\in\mathcal S$. 
Under the law of one price, each wealth $Z$ is generated by trading in a market which comes at the initial cost $\pi(Z)$. 
Assembling all these aspects, we obtain the following generalisation of the return risk of $X$ in~\eqref{eq:introduction_rrm}:
\begin{align*}
    \inf\left\{\pi(Z)\,;\, Z\in\mathcal{S},\rho\left(\log\left(\tfrac{X}{Z}\right)\right)\le 0\right\}.
\end{align*}
The interpretation of this quantity is immediate: it is the infimal price of a trading strategy resulting in payoff $Z\in\mathcal S$ relative to which the agent's loss has an acceptable log-return. Acceptability is measured by the monetary risk measure $\rho$. 
Even more generally, one can omit $\rho$ altogether and substitute the notion of acceptability by $\frac X Z$ belonging to a suitable set $\mathcal{B}$ of positive random variables. 
This is the analogue to the so-called acceptance set appearing in the definition of MARMs. 
We obtain
\begin{align}\label{eq:introduction_marrm}
    \inf\left\{\pi(Z)\,;\, Z\in\mathcal{S},\tfrac{X}{Z}\in\mathcal{B}\right\}.
\end{align}

In a nutshell, this is how we suggest to follow the philosophy of MARMs to generalize RRMs. We call our new approach \textit{multi-asset return risk measure} (MARRM).
In this manuscript, we focus on analyzing properties of maps of the form~\eqref{eq:introduction_marrm}. 
In Section~\ref{sec:marrm}, we formally introduce our new concept and point out its connection to risk sharing in Section~\ref{sec:riskSharing}. The latter is analogous to similar representations of MARMs via inf-convolutions (see, e.g., Baes et al.~\cite{baes_2020}), but replaces additive aggregation by multiplicative aggregation: 
An optimal solution $Z^{\ast}$ of Equation~(\ref{eq:introduction_marrm}) means that the MARRM decomposes the risk into a product of this wealth-generating market portfolio $Z^{\ast}$ and a factor $\frac{X}{Z^{\ast}}$, which represents a relative loss that should meet an acceptability constraint. 
In Section~\ref{sec:riskSharing}, we interpret this risk-sharing approach inherent to MARRMs comprehensively, framing it as a multiplicative inf-convolution.

Then, in Section~\ref{sec:algebraicProperties} we uncover the algebraic properties of MARRMs. In particular, in Theorem~\ref{thm:quasiConvexity} we show that a positively homogeneous MARRM is quasi-convex if and only if it is convex. 
The proof utilizes the representation of MARRMs via RRMs with an adjusted acceptance set, see Lemma~\ref{lem:reductionLemma}.
In addition, in Examples~\ref{exam:missingConvexity} and~\ref{exam:nonConvexityRelativeAcceptanceSet}, we illustrate situations in which we obtain a convex MARRM or not.    
In Section~\ref{sec:finitness}, we study finiteness of the new risk measure, see Lemma~\ref{lem:finiteness}, and state conditions on the set $\mathcal{S}$ and the map $\pi$, i.e.,~the financial market, as well as on the acceptance set $\mathcal{B}$ to prevent the MARRM to attain the value zero, see Theorem~\ref{thm:noArbitrage_MARRM}. 
This is important to prevent a sort of acceptability arbitrage opportunities where actual losses are evaluated as riskless. 
Our condition on the financial market is weaker than the well-known ``no free lunch with vanishing risk'' condition. In Section~\ref{sec:marrm_logReturns}, we point out the connection to MARMs (Proposition~\ref{prop:MARRMasMARM}) and how their properties transfer to MARRMs (Lemma~\ref{lem:connectionMARRMandMARM}).
In Section~\ref{sec:dualRepresentation}, we develop dual representations for this setup. Here, Proposition~\ref{prop:dualLogconvexMARRRM} is the classical counterpart to dual representations of convex MARMs. Since nowadays the focus is also on properties weaker than convexity, Theorems~\ref{thm:dualLogStarShapedMARRRM} and~\ref{thm:dualQuasiLogconvexMARRRM} state dual representations for star-shaped, repsectively quasi-convex, MARRMs. To the best of our knowledge, these results are new and in particular, they are also of independent interest in the theory of MARMs.
Finally, in Section~\ref{sec:marrm_continuous_time}, we illustrate the new risk measure class for a continuous-time financial market model and different concrete choices of the set $\mathcal{B}$ from Equation~(\ref{eq:introduction_marrm}) based on losses from a US private automobile insurance. For the case study, we use sets $\mathcal{B}$ based on RRMs derived from the classical risk measures Value-at-Risk and Expected Shortfall. Payoffs in $\mathcal{S}$ are modeled by trading in a Black-Scholes market with constant portfolio processes. Then, we compare MARRMs, RRMs and MARMs in this concrete setup. 
To the best of our knowledge, comparing return and classical risk measures based on different numbers of eligible assets in a case study is a novel contribution.
We find a significant relative difference greater than 6\% between MARRMs and RRMs in all considered cases. Furthermore, while MARRM and MARM agree in the Value-at-Risk case, they differ in the Expected Shortfall case and display a relative difference above $5\%$.

In Section A.3 of the online appendix, we also test MARRMs for time series data of a stock index. Here, we also reevaluate MARRMs over time and find that they suggest larger investments in the stocks in times of crisis.
\newline

\textit{Throughout the whole manuscript we use the following standard notation:} The natural logarithm of a real number $x>0$ is denoted by $\log(x)$.
For two sets $\mathcal{A}$ and $\mathcal{B}$, we write $\mathcal{A}\subset \mathcal{B}$ for $\mathcal{A}$ being a subset of $\mathcal{B}$, where $\mathcal{A}=\mathcal{B}$ is possible. 

The effective domain of a function $f\colon\mathcal{X}\rightarrow [-\infty,\infty]$ is $\dom(f):=\{X\in \mathcal{X}\,;\, f(X)<\infty\}$. The map $f$ is called proper if it does
not attain $-\infty$ and $\dom(f)\neq\emptyset$. If $\mathcal{X}$ is a topological space, then $f$ is lower semicontinuous if for each $c\in\mathbb{R}$ the set $\{X\in\mathcal{X}\,;\, f(X)\leq c\}$ is closed. 

Further, underlying is a probability space $(\Omega,\mathcal{F},P)$. 
We denote by $L^{0}(\Omega,\mathcal{F},P)$ -- or $L^{0}$ in short -- the linear space of all equivalence classes of real-valued $\mathcal{F}$-measurable functions with respect to the $P$-almost sure ($P$-a.s.~or a.s.)~order, and the linear space of all equivalence classes of $P$-a.s.\ bounded random variables by $L^{\infty}(\Omega,\mathcal{F},P)$, or $L^{\infty}$ for short.
The linear space of equivalence classes of $p$-integrable random variables with $p\in[1,\infty)$ is denoted by $L^p(\Omega,\mathcal{F},P)$, or $L^p$ for short. We equip every $L^p$-space with $p\in\{0\}\cup[1,\infty]$ with the almost sure ordering. Further, the positive cone and the set of strictly positive elements are given by $L^{p}_{+}=\{X\in L^p\,;\,X\geq 0\text{ a.s.}\}$ and $L^{p}_{++}=\{X\in L^p\,;\,X>0\text{ a.s.}\}$, respectively. The left quantile function of a random variable $X$ is denoted by $q_{X}^{-}$. 

For two subsets $\mathcal{A}$ and $\mathcal{B}$ of the space $L^0$,  Minkowski sum, product and quotient are defined by 
\begin{gather*}\mathcal{A}+\mathcal{B}:=\{X+Y\,;\,(X,Y)\in\mathcal{A}\times\mathcal{B}\},\quad \mathcal{A}\cdot\mathcal{B}:=\{XY\,;\,(X,Y)\in\mathcal{A}\times\mathcal{B}\},\\
\tfrac{\mathcal{A}}{\mathcal{B}}:=\{XY^{-1}\,;\,(X,Y)\in\mathcal{A}\times\mathcal{B}\}.\end{gather*}
In the definition of $\frac{\mathcal A}{\mathcal B}$, we tacitly assume that no random variable in $\mathcal B$ attains value $0$ with positive probability.

\section{A multi-asset extension of return risk measures}\label{sec:marrm}

The goal of the manuscript is to generalize return risk measures, as in~\eqref{eq:introduction_rrm}, by using multiple eligible assets. This new concept allows agents to trade in order to generate wealth relative to which their loss-bearing capacity is sufficient to carry the (potentially non-tradable) loss $X$. 

\subsection{Definition of multi-asset return risk measures}

In Equation~\eqref{eq:introduction_marrm}, we only sketched the idea of a multi-asset return risk measure; here, we make it precise.
As a first step, we introduce all ingredients in our setting.
For clarity, we always mention their role in the special setting of return risk measures in \cite{Return}.

\textbf{Domain:} Risk measures are functionals
    defined on a particular domain to which we shall refer as \textbf{\textit{model set}}. 
    It is given by a nonempty subset $\mathcal{C}\subset L^0_{++}$ whose elements we interpret as possible losses an agent faces at a pre-specified future time point. 

    In the example of RRMs from~\cite{Return}, the model set  is $L^{\infty}_{++}$. 
    However, there is no reason to restrict our attention to bounded losses, and we shall work with an arbitrary set $\mathcal{C}$. 
    
\textbf{Financial market:} Second, we have to describe the payoff that we can create by trading in a financial market. 
    All payoffs generated in this manner form the so-called \textbf{\textit{security set}}, which is also a nonempty subset $\mathcal{S}\subset L^0_{++}$. 
    This inclusion imposes the constraint that the generated payoff must be positive, thereby also restricting which trading strategies are deemed admissible.
    An element of $\mathcal{S}$ is called a \textbf{\em security payoff}.
    An important convention to keep in mind is that the random variables in the model set $\mathcal{C}$ are losses, while the security payoffs in $\CS$ are interpreted as wealth. 
    The price of such a security payoff is given by a map $\pi\colon\mathcal{S}\rightarrow(0,\infty)$, a so-called \textbf{\textit{pricing map}}. The fact that $\pi>0$ reflects the absence of arbitrage opportunities in the financial market given by $\mathcal{S}$. 
    
    In the specific example of  RRMs, the security set is  based on a single eligible asset, namely a risk-free bank account. In this case, it holds $\mathcal{S}=(0,\infty)$ and $\pi(x)=x$ for all $x\in(0,\infty)$.
    Here, however, we aim to allow for other eligible assets like stocks. 
    Therefore, the set $\mathcal{S}$ is introduced as a generic subset of $L^{0}_{++}$. In particular, unbounded payoffs are possible. For example, $\mathcal{S}$ could represent buy-and-hold strategies in a classical Black-Scholes model, involving investments in a riskless bank account on the one hand and a single stock with lognormally distributed payoff $Y$ on the other.
    In this case, $\mathcal{S}=\{a+bY\,;\,a,b\geq 0\text{ with }(a,b)\neq(0,0)\}.$

\textbf{Acceptability criteria:} 
    Next, we formally introduce the set $\mathcal{B}$ from~\eqref{eq:introduction_marrm}, which describes if a loss relative to a wealth is acceptable. Here, it is crucial that $\mathcal{B}$ is usually specified without recourse to any financial market, e.g.,~by a regulatory authority. 
    Furthermore, it is already apparent from \eqref{eq:introduction_marrm} that we need to check if a value of the form $\frac{X}{Z}$ with $X\in\mathcal{C}$ and $Z\in\mathcal{S}$ lies in $\mathcal{B}$. 
    However, it is well possible that $\frac{X}{Z}\notin \mathcal{C}$.
    Therefore, we define $\mathcal{B}$ as a subset of the so-called \textbf{\textit{set of relative losses}}, which is an independent nonempty subset $\mathcal{K}\subset L^0_{++}$ and which can differ from $\mathcal{C}$. Then, the \textbf{\textit{relative acceptance set}} is a nonempty proper subset $\mathcal{B}$ of $\mathcal{K}$ which is monotone, i.e.,
$$\forall\,X\in\mathcal{B}\,\,\forall\,Y\in\mathcal{K}:\quad Y\leq X\quad\implies\quad Y\in\mathcal{B}.$$

We emphasize that we are interested in values of the form $\frac{X}{Z}$ with $X\in\mathcal{C}$ and $Z\in\mathcal{S}$. Hence, we divide a loss by a payoff, which means that smaller values of $\frac{X}{Z}$ are more desirable for an agent, because then the loss $X$ is less dangerous in comparison to the payoff $Z$. This explains the inequality in the definition of a monotone set.

For this setting, we now define the new concept of a multi-asset version of a return risk measure.

\vspace{-9pt}

\begin{definition}\label{defi:marrm}
A \textbf{\textit{return risk measurement regime}} is a triplet $\mf R=(\mathcal B,\mathcal S,\pi)$ consisting of a relative acceptance set $\mathcal B$, a security set $\mathcal S$ and a pricing map $\pi$.

Given a return risk measurement regime $\mf R$, the associated \textbf{\textit{multi-asset return risk measure (MARRM)}} $\eta_{\mf R}\colon \mathcal{C}\rightarrow[0,\infty]$ is given by
    \begin{align*}
        \eta_{\mf R}(X):=\inf\left\{\pi(Z)\,;\,Z\in\mathcal{S},\tfrac{X}{Z}\in\mathcal{B}\right\}.
    \end{align*}
\end{definition}

The original introduction of RRMs in~\cite{Return} can be embedded in our general framework by setting  $\mathcal{C}=\mathcal{K} = L^\infty_{++}$ and $\mathcal{S} = (0,\infty)$. 
In contrast, if we use an arbitrary security set $\mathcal{S}\subset L^{\infty}_{++}$, it is not ensured {\em a priori} that $\frac{X}{Z}\in L^{\infty}_{++}$ unless $Z$ is bounded away from $0$. Hence, a choice of $\mathcal{K}=L^{\infty}_{++}$ is not possible any longer without additional assumptions. This motivates us to introduce a notion of compatibility between the model set $\mathcal{C}$, the set of relative losses $\mathcal{K}$ and the security set $\mathcal S$. 
This will be crucial to establish basic properties like monotonicity of the MARRM as an inheritance of the monotonicity of the relative acceptance set $\mathcal B$
(see Lemma~\ref{lem:algebraicProperties} below).
The following new  assumption is in place throughout the manuscript:

\begin{assumption}[Compatibility of $\CC,\mathcal S$, and $\mathcal K$]\label{assum:comp}
    It holds that 
    $\tfrac{\mathcal C}{\mathcal S}\subset\mathcal K.$
\end{assumption}

Assumption~\ref{assum:comp} emphasizes that the set of relative losses $\mathcal{K}$ encompasses fractions $\frac{X}{Z}$ with $X\in\mathcal{C}$ and $Z\in\mathcal{S}$, thus expressing the size of loss $X$ in terms of wealth $Z$. Up to scaling, the elements in $\mathcal K$ are  random percentages of wealth that are lost, depending on the future state of the economy. In particular, they are dimensionless.

\begin{remark}\label{rem:compatibilityCondition}
\begin{enumerate}[(a)]
\item Any of the following conditions ensures compatibility of $\mathcal C,\mathcal S$, and $\mathcal K$:
    \begin{enumerate}[(i)]
    \item $\mathcal{K} = \tfrac{\mathcal C}{\mathcal S}$, i.e.,~we can choose $\mathcal{C}$ and $\mathcal{S}$, and this choice gives us $\mathcal{K}$.
    \item $\mathcal{K}=\mathcal{C}$, the multiplication on $\mathcal{C}$ is a binary operation -- i.e., for all $X,Y\in\mathcal{C}$ it holds that $XY\in\mathcal{C}$ --, and $\left\{\frac{1}{Z}\,;\;Z\in\mathcal{S}\right\}\subset \mathcal{C}$. This condition is satisfied  in~\cite{Return}, see the more general observation in (b). Furthermore, it is satisfied for $\mathcal{K}=\mathcal{C} = L_{++}^{0}$.
    \end{enumerate}
    \item Regarding the situation in~\cite{Return}, compatibility is satisfied for $\mathcal{K}=\mathcal{C}=L^{\infty}_{++}$ and $$\mathcal{S}\subset \text{int}(L^{\infty}_{+}) = \{X\in L^{\infty}\,;\,\exists\,\eps>0: X\geq \eps\text{ a.s.}\}.$$ 
    In this case we have for all $X\in\mathcal{C}$ and $Z\in\mathcal{S}$ that $\frac{X}{Z}\in L^{\infty}_{++} = \mathcal{K}$. 
    \end{enumerate}
\end{remark}

\subsection{Connection to risk sharing}\label{sec:riskSharing}

Mathematically, risk sharing among a collective of agents is often described by an additive inf-convolution of the involved individual risk measures; see, among many other references, \cite{BarKaroui,FilSvi,Liebrich}.
We uncover here that MARRMs exhibit a similar structure, but the additive inf-convolution needs to be replaced by its multiplicative counterpart which has been studied recently in \cite{Aygun}. More precisely,
\begin{align*}
    \eta_{\mf R}(X)=\Big(\Pi\otimes \delta^\odot_{\mathcal{B}\cap \frac{\mathcal{C}}{\mathcal{S}}}\Big)(X),
\end{align*} 
where for two functions $f,h:\mathcal{C}\rightarrow(0,\infty]$, the multiplicative inf-convolution is defined by 
\begin{align*}
    (f\otimes h)(X):=\inf\{f(Y)h(V)\,;\,Y,V\in\mathcal{C}, YV = X\};
\end{align*}
for a set $\mathcal{A}\subset\mathcal{K}$ the multiplicative indicator function $\delta^\odot_{\mathcal{A}}$ is defined by $\delta^\odot_{\mathcal{A}}(X)=1$ if $X\in\mathcal A$ and $\delta^\odot_{\mathcal{A}}(X)=\infty$ otherwise; 
and where the function $\Pi$ extends the pricing map $\pi$ to all of $L^0_{++}$ by setting $\Pi(X)=\infty$ if $X\notin\CS$.
This representation shows that a MARRM is obtained by splitting a loss $X$ into a wealth part $Z\in \mathcal S$ generated by a market portfolio and a leverage factor $\frac XZ$ which needs to satisfy a -- potentially externally prescribed -- constraint formulated by the set $\mathcal B$. The aggregated risk of this procedure is $\pi(Z)$ for a $Z$ guaranteeing the acceptability of the relative loss $\frac{X}{Z}$. The trade-off is between creating enough wealth $Z$ to make the relative loss acceptable and minimizing the hedging costs $\pi(Z)$. Overall, the MARRM offers a new economic model for risk sharing, as the loss is shared multiplicatively instead of additively.
 
We illustrate multiplicative risk sharing with a concrete insurance example at the end of Section~\ref{sec:study_comparison_marm},  comparing US private automobile insurance losses with portfolios from a Black-Scholes market. The risk sharing has two components: an ``optimal payoff'' (following the terminology of \cite{baes_2020})  from Black-Scholes trading, and a leverage factor that compares the insurance loss to the market portfolio. 
These components are analyzed separately in Figures~\ref{fig:var_arar_portfolio_processes} and \ref{fig:risk_sharing_marrm_marm}. We also discuss the difference between additive and multiplicative risk sharing by drawing a comparison to classical MARMs, as introduced in Section~\ref{sec:marrm_logReturns} below.

\subsection{Algebraic properties and the reduction lemma}\label{sec:algebraicProperties}

In this section, we analyze under which conditions a MARRM satisfies classical algebraic properties for multi-asset risk measures, namely monotonicity, subadditivity, quasi-convexity, positive homogeneity. 
In particular, as for classical monetary risk measures or multi-asset risk measures more generally, 
it turns out that quasi-convexity is equivalent to convexity; see Theorem~\ref{thm:quasiConvexity}. Further, we state some preliminary results that we need later.

For the first result, we recall that a set $\mathcal{A}\subset L^{0}_{++}$ is a cone if $\lambda \mathcal{A}\subset \mathcal{A}$ holds for each $\lambda\in(0,\infty)$. 
It is star-shaped if the inclusion $\lambda\mathcal{A}\subset\mathcal{A}$ holds for all $\lambda\in(0,1)$. Note that we do not require that zero is an element of the set in our definitions of a cone and a star-shaped set, because the sets we are working with are typically subsets of $L^0_{++}$. Furthermore, given $\mathcal A\subset L^0$, a map  $f\colon\mathcal{A}\rightarrow(-\infty,\infty]$ is positively homogeneous, if $\mathcal{A}$ is a cone and for each $\lambda\in(0,\infty)$ and $X\in\mathcal{A}$ we have that $f(\lambda X)=\lambda f(X)$. The map $f$ is star-shaped, if $\mathcal{A}$ is star-shaped and for all $\lambda\in(0,1)$ and $X\in\mathcal{A}$ it holds that $f(\lambda X)\leq \lambda f(X)$.
	
\begin{lemma}\label{lem:algebraicProperties}
    Let $\mf R=(\mathcal B,\mathcal S,\pi)$ be a return risk measurement regime and consider the associated MARRM. 
    \begin{enumerate}[(i)]
        \item $\eta_{\mf R}$ is monotone, i.e., for all $X,Y\in\mathcal C$ with $X\le Y$,
        $\eta_{\mf R}(X)\le \eta_{\mf R}(Y).$ 
        \item If $\mathcal{C}$, $\mathcal{S}$ are cones, and $\pi$ is positively homogeneous, then $\eta_{\mf R}$ is positively homogeneous.
        \item If $\mathcal{C}$, $\mathcal{S}$ and $\pi$ are star-shaped, then $\eta_{\mf R}$ is star-shaped. 
    \end{enumerate}
\end{lemma}	

The proof of Lemma~\ref{lem:algebraicProperties} is straightforward and shall be omitted. 

In the following lemma, we use the notation $\mathcal{S}_m:=\{Z\in\mathcal{S}\,;\,\pi(Z)=m\}$ for all $m>0$. We call it ``reduction lemma'' as it is the counterpart of the result in~\cite[Lemma 3]{farkas_2015}.

\begin{lemma}\label{lem:reductionLemma}
    If $\mathcal{S}$ is a cone and $\pi$ is positively homogeneous, then
    \begin{align*}
        \eta_{\mf R}(X) &= \inf\left\{\lambda >0\,;\,\tfrac{X}{\lambda}\in\mathcal{B}\cdot \mathcal{S}_1\right\}.
    \end{align*}
    Further, the set $\mathcal{B}\cdot \mathcal{S}_1$ is monotone in the following sense:
    $$\forall\,X\in\mathcal{B}\cdot \mathcal{S}_1\,\,\forall\,Y\in \mathcal{C}:\quad Y\leq X\quad\implies\quad Y\in\mathcal{B}\cdot \mathcal{S}_1.$$
\end{lemma}

\begin{proof}
    We only prove the second statement. Let $B\in\mathcal{B},Z\in\mathcal{S}_1,Y\in \mathcal{C}$ with $Y\leq B Z$. Then, $\frac{Y}{Z}\in\mathcal{B}$, which implies that $Y\in\mathcal{B}\cdot \mathcal{S}_1$. 
\end{proof}

For classical risk measures, quasi-convexity describes that the risk measure does not punish, but rather rewards diversification; see, e.g.,~\cite{cerreiaVioglio_2011,karoui_2009}. Moreover, cash-additivity of monetary risk measures guarantees that quasi-convexity is equivalent to convexity. For MARRMs, we have to substitute cash-additivity by positive homogeneity to obtain such an equivalence. This is the implication ``$(ii)\Rightarrow(iii)$'' in the next result, which relies on the following fact: A quasi-convex and positively homogeneous map with range $[0,\infty]$ is subadditive. This statement fails for classical risk measures, because their range is the extended real line. In addition, the next result shows that convexity for MARRMs is characterized by the set $\mathcal{B}\cdot \mathcal{S}_1$ from the reduction lemma. 

\begin{theorem}\label{thm:quasiConvexity}
    Let $\mf R=(\mathcal B,\mathcal S,\pi)$ be a return risk measurement regime such that $\mathcal{S}$ is a cone and $\pi$ is positively homogeneous. Further, let $\mathcal{C}$ be a convex cone. Then, the following statements are equivalent:
    \begin{enumerate}[(i)]
        \item The set $(\mathcal{B}\cdot \mathcal{S}_1)\cap \mathcal{C}$ is convex.
        \item $\eta_{\mf R}$ is quasi-convex.
        \item $\eta_{\mf R}$ is subadditive.
        \item $\eta_{\mf R}$ is convex.
    \end{enumerate}
\end{theorem}

\begin{proof}
    Applying Lemma~\ref{lem:reductionLemma} and using that $\mathcal{C}$ is a cone, it holds for all $X\in \mathcal{C}$ that
        \[\eta_{\mf R}(X) = \inf\left\{\lambda >0\,;\,\tfrac{X}{\lambda}\in(\mathcal{B}\cdot \mathcal{S}_1)\cap \mathcal{C}\right\}.\]
        ``$(i)\Rightarrow(ii)$'': Let $X,Y\in \mathcal{C}$ and choose an arbitrary $\alpha\in(0,1)$. Without loss of generality, we assume that $\eta_{\mf R}(X)\leq \eta_{\mf R}(Y)$. By the monotonicity of $\mathcal{B}\cdot \mathcal{S}_1$ and $\mathcal{C}$ being a cone, we obtain for all $\lambda_X\in(\eta_{\mf R}(X),\infty)$ and for all $\lambda_Y\in(\eta_{\mf R}(Y),\infty)$ that 
        $\tfrac{X}{\lambda_X}\in(\mathcal{B}\cdot \mathcal{S}_1)\cap \mathcal{C}$ and $\tfrac{Y}{\lambda_Y}\in(\mathcal{B}\cdot \mathcal{S}_1)\cap\mathcal{C}$.
        So, by the convexity of $(\mathcal{B}\cdot \mathcal{S}_1)\cap\mathcal{C}$, we get for all $\lambda \in(\eta_{\mf R}(Y),\infty)$ that
        \begin{align*}
            \alpha \tfrac{X}{\lambda} + (1-\alpha) \tfrac{Y}{\lambda} \in (\mathcal{B}\cdot \mathcal{S}_1)\cap \mathcal{C}.
        \end{align*}
        This shows that for $\lambda \in(\eta_{\mf R}(Y),\infty)$, the condition $\frac{Y}{\lambda}\in(\mathcal{B}\cdot \mathcal{S}_1)\cap \mathcal{C}$ implies that $\alpha \frac{X}{\lambda} + (1-\alpha) \frac{Y}{\lambda} \in (\mathcal{B}\cdot \mathcal{S}_1)\cap \mathcal{C}$. Hence, $\eta_{\mf R}(\alpha X + (1-\alpha) Y)\leq \eta_{\mf R}(Y)$.
         
        ``$(ii)\Rightarrow(iii)$'': Pick a pair $X,Y\in\mathcal C$ such that $\eta_{\mf R}(X)<\infty$ and $\eta_{\mf R}(Y)<\infty$. If this condition is not satisfied, the defining inequality for subadditivity is automatically satisfied. 
        Let $\eps>0$ be arbitrary and set $\lambda_X(\eps):=\eta_{\mf R}(X)+\eps$, $\lambda_Y(\eps):=\eta_{\mf R}(Y)+\eps$, and $\alpha:=\lambda_X(\eps)/(\lambda_X(\eps)+\lambda_Y(\eps))$.
        By (ii), we know that $\eta_{\mf R}$ is quasi-convex, which gives us that
        \begin{align*}
            \eta_{\mf R}\left(\tfrac{X+Y}{\lambda_X(\eps)+\lambda_Y(\eps)}\right)&=\eta_{\mf R}\left(\alpha \tfrac{X}{\lambda_X(\eps)} + (1-\alpha)\tfrac{Y}{\lambda_Y(\eps)}\right)\leq \max\left\{\eta_{\mf R}\left(\tfrac{X}{\lambda_X(\eps)}\right), \eta_{\mf R}\left(\tfrac{Y}{\lambda_Y(\eps)}\right)\right\}.
        \end{align*} 
        Further, by Lemma~\ref{lem:algebraicProperties} (ii), $\eta_{\mf R}$ is positively homogeneous, which gives us that 
        \begin{align*}
            \eta_{\mf R}(X+Y)\leq\max\left\{\tfrac{\lambda_X(\eps)+\lambda_Y(\eps)}{\lambda_X(\eps)}\eta_{\mf R}(X), \tfrac{\lambda_X(\eps)+\lambda_Y(\eps)}{\lambda_Y(\eps)}\eta_{\mf R}(Y)\right\}.
        \end{align*} 
        Letting $\eps\downarrow 0$, we get the desired inequality. 

        ``$(iii)\Rightarrow(iv)$'': $\eta_{\mf R}$ being positively homogeneous and subadditive implies that $\eta_{\mf R}$ is convex.

        ``$(iv)\Rightarrow(i)$'': We perform a proof by contraposition. To do so, note that $(\mathcal{B}\cdot \mathcal{S}_1)\cap \mathcal{C}$ not being convex and $\mathcal{C}$ being convex implies that there exist $X,Y\in(\mathcal{B}\cdot \mathcal{S}_1)\cap \mathcal{C}$ and $\alpha\in(0,1)$ such that
        $\alpha X + (1-\alpha) Y \in \mathcal{C}\backslash (\mathcal{B}\cdot \mathcal{S}_1)$.
        Finally, by the monotonicity of $\mathcal{B}\cdot \mathcal{S}_1$ we obtain 
        $$\alpha\eta_{\mf R}(X)+(1-\alpha)\eta_{\mf R}(Y) \leq 1 < \eta_{\mf R}(\alpha X + (1-\alpha) Y).$$
\end{proof}

In the following example, we elaborate further on the condition of $(\mathcal{B}\cdot \mathcal{S}_1)\cap \mathcal{C}$ being convex. 
We illustrate that, even if $\mathcal B$, $\mathcal S_1$, and $\mathcal C$ are convex, this intersection may or may not be convex. 

\begin{example}\label{exam:missingConvexity}
    Let $\Omega = \{\omega_1,\omega_2\}$, i.e., we can identify $L^{0}$ with the two-dimensional plane  $\mathbb{R}^{2}$. 
    Further, we set $\mathcal{C}=(0,\infty)\times(0,\infty)$ and assume that 
    $$\mathcal{S} = \{\lambda (1,1)^{\intercal}+\mu(0.1,2)^{\intercal}\,;\,\lambda,\mu\in(0,\infty)\}.$$ 
    Suppose that the pricing map $\pi$ is linear and given by $\pi((1,1)^{\intercal})=\pi((0.1,2)^{\intercal})=1$. Hence,
    $$\mathcal{S}_1= \{\lambda (1,1)^{\intercal}+(1-\lambda)(0.1,2)^{\intercal}\,;\,\lambda\in[0,1]\}.$$
    Figure~\ref{fig:missingConvexity} contains the ``worst-case relative acceptance set'' 
    $\mathcal{B}_{\text{left}}=\{x\in\mathbb{R}^2\,;\,x_1\leq 1, x_2\leq 1\}$ on the left-hand and $\mathcal{B}_{\text{right}}=\{x\in\mathbb{R}^2\,;\,x_1\leq 1,\,x_2\leq 4-3x_1\}$ on the right-hand side. 
    We can see that $(\mathcal{B}_{\text{left}}\cdot \mathcal{S}_1)\cap \mathcal{C}=\mathcal{B}\cdot \mathcal{S}_1$ is convex, while $\mathcal{B}_{\text{right}}\cdot \mathcal{S}_1$ fails to be convex.
    \begin{figure}
        \begin{subfigure}[b]{0.42\textwidth}
            \includegraphics[width=\linewidth]{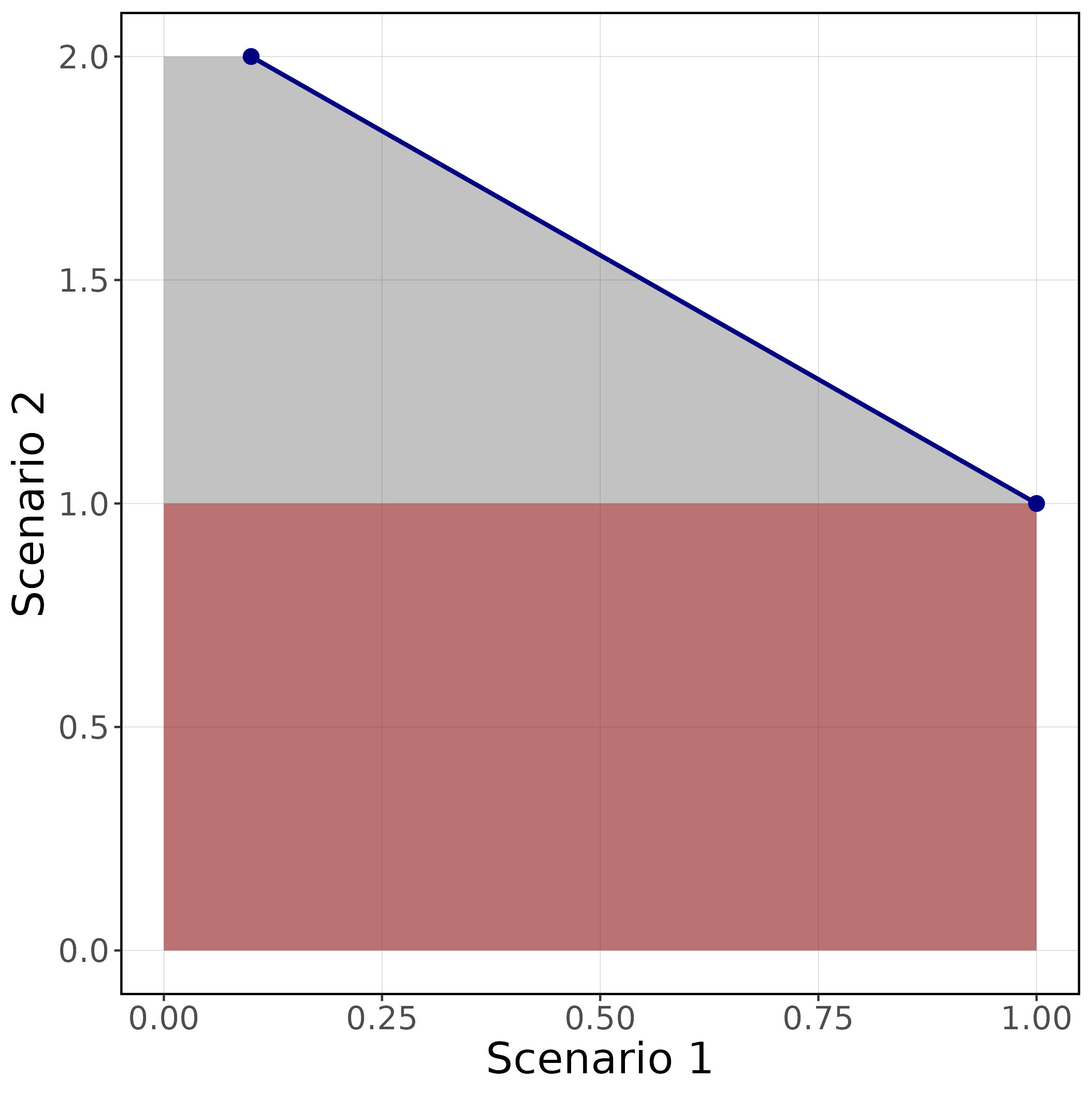}
        \end{subfigure}
        \hspace{1.1cm}
        \begin{subfigure}[b]{0.42\textwidth}
            \includegraphics[width=\linewidth]{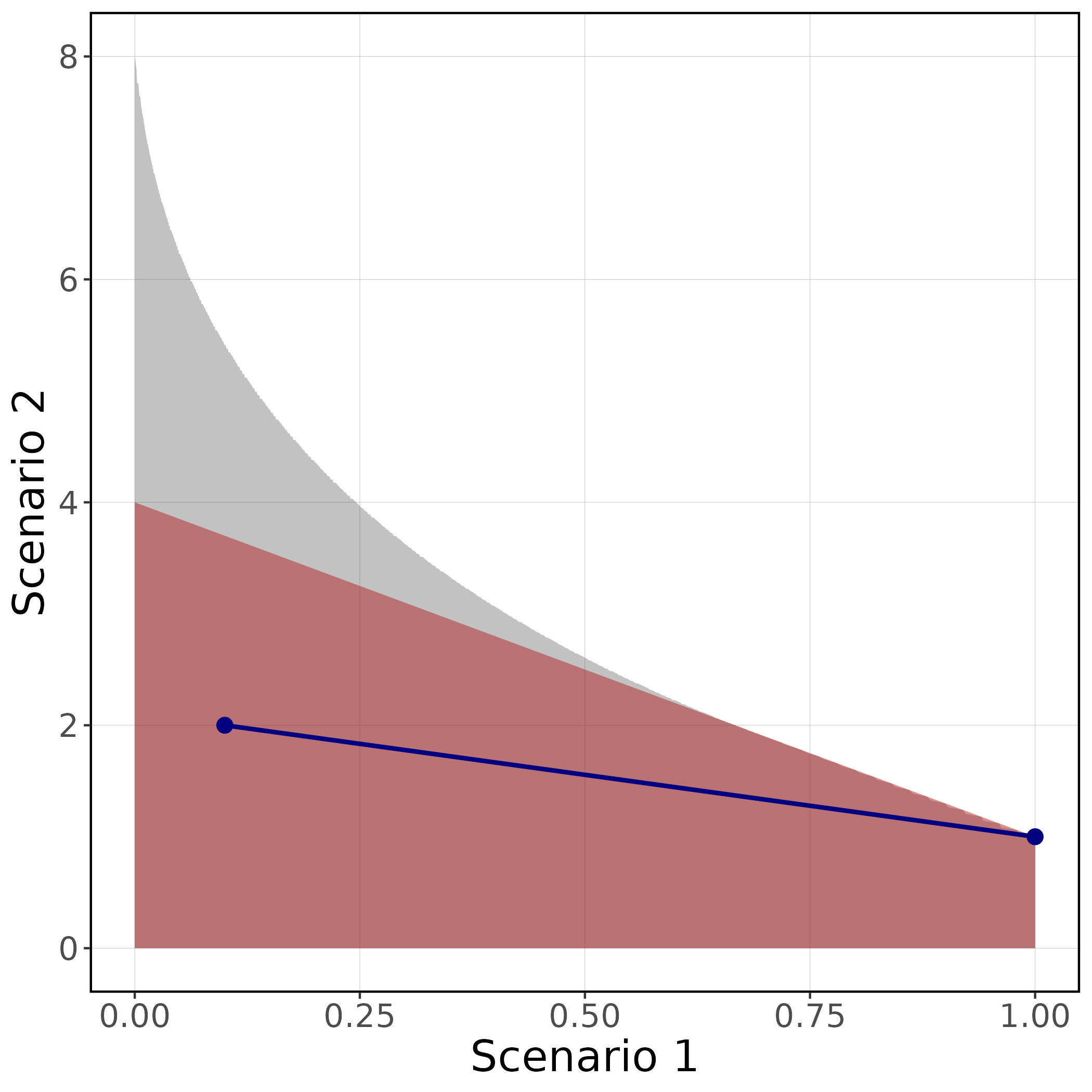}
        \end{subfigure}
        \captionsetup{font=footnotesize}
        \caption{\footnotesize In both figures, the red surface represents the relative acceptance set $\mathcal{B}$, while the blue line is the set $\mathcal{S}_1$ of security payoffs with price $1$.
        The gray and red surface together is the set $\mathcal{B}\cdot\mathcal{S}_1$. \textit{LHS:} The set $\mathcal{B}\cdot\mathcal{S}_1$ is convex. \textit{RHS:}  The set $\mathcal{B}\cdot\mathcal{S}_1$ is not convex. }
        \label{fig:missingConvexity}
    \end{figure}
\end{example}

\begin{remark}\label{rem:discussionConvexityRelativeAcceptanceSet}
    Note, $\mathcal{B}\cdot \mathcal{S}_1$ being convex is a stronger condition than $(\mathcal{B}\cdot \mathcal{S}_1)\cap \mathcal{C}$ being convex. To obtain a convex set $\mathcal{B}\cdot \mathcal{S}_1$, it is intuitive (but not necessary) to choose a relative acceptance set $\mathcal{B}$ which is convex.
    However, a typical form of such a relative  acceptance set is given by 
    $$\mathcal{B} = \{X\in\mathcal{C}\,;\,\rho(\log(X))\leq 0\}$$ with a monetary risk measure $\rho$. 
    In this situation, convexity of $\mathcal{B}$ is equivalent to convexity of $\rho\circ\log$.
    Example~\ref{exam:nonConvexityRelativeAcceptanceSet} below shows that this condition can be violated, even if $\rho $ is convex. Such a discussion is missing for~\cite[Lemma 1 (b)]{Return}. 
\end{remark}

\begin{example}\label{exam:nonConvexityRelativeAcceptanceSet}
    Suppose that the underlying probability space is atomless. 
    Select an event $A\in\mathcal F$ and a probability level $p$ with $0<P(A)<p<\frac 1 2$, and set $x:=1-\frac{P(A)}{1-p}$.
    Moreover, for a constant $a>0$ we consider $X=e^{a} 1_A + 1_{A^c}$ and $Y=1_{A}+e^{a} 1_{A^c}$. 

    The Expected Shortfall at level $p$ for $Z\in L^1$ is defined by ${\rm ES}_{p}(Z) = \frac{1}{1-p}\int_{p}^{1}q_{Z}^{-}(u){\rm d}u$. With this definition, we obtain 
    \[
        \lambda \expectedShortfall{p}{\log(X)} + (1-\lambda) \expectedShortfall{p}{\log(Y)}= a \left(1-\lambda x\right)=:f_1(\lambda, x, a).
    \]

    Note that $\lambda X + (1-\lambda) Y = (\lambda e^{a} + (1-\lambda))1_A + ((1-\lambda) e^{a} + \lambda) 1_{A^c}$. Then, for $\lambda > \frac{1}{2}$, we obtain
    \begin{align*}
        \expectedShortfall{p}{\log(\lambda X + (1-\lambda) Y} = (1-x)\log(\lambda e^a + (1-\lambda)) + x \log((1-\lambda) e^{a} + \lambda) =: f_2(\lambda, x, a).
    \end{align*}
    Then, $f_2(0.75,0.5,100) > 99.1 > 62.5 = f_1(0.75,0.5,100)$  shows that $\symbolExpectedShortfall{p} \circ \log$ is not convex.
\end{example}

Now, we state a sufficient condition to obtain a property that we call ``log-star-shapedness'' of the MARRM. By Remark~\ref{sec:discussion_assumptions_logstarshaped} below, this property holds, if the relative acceptance set is defined via a star-shaped monetary risk measure. Star-shaped monetary risk measures have recently been studied in detail in~\cite{castagnoli_2022,Laeven2, Zullino}. 

For the next result, given a constant $\alpha>0$ and a set $\mathcal{A}\subset{L^{0}}$, we use the notation $\mathcal{A}^{\alpha}$ for the set of powers $\{X^{\alpha}\,;\,X\in\mathcal{A}\}$. We call a set $\mathcal{A}\subset L^0_{++}$ log-star-shaped, if for all $\alpha\in(0,1)$ it holds that $\mathcal{A}^{\alpha}\subset\mathcal{A}$. For such a set $\mathcal{A}$, we call a map $f:\mathcal{A}\mapsto[0,\infty]$ log-star-shaped, if for all $\alpha\in(0,1)$ and $X\in\mathcal{A}$ it holds that $f(X^{\alpha})\leq f(X)^{\alpha}$.

\begin{proposition}\label{prop:log-star-shaped}
    Let $\mf R=(\mathcal B,\mathcal S,\pi)$ be a return risk measurement regime. Assume that $\mathcal{C},\mathcal{B}$ and $\mathcal{S}$, as well as, $\pi$ are log-star-shaped. Then, $\eta_{\mf R}$ is log-star-shaped.
\end{proposition}

\begin{proof}
    Note that for a subset $\mathcal{A}\subset L^0_{++}$ and for all $\alpha\in(0,1)$ it holds that $\mathcal{A}^{\alpha}\subset\mathcal{A}$ if and only if $\mathcal{A}\subset\mathcal{A}^{\frac{1}{\alpha}}$. Then, for an arbitrary $\alpha\in(0,1)$ we obtain
    \begin{align*}
        \eta_{\mf R}(X^{\alpha}) &= \inf\left\{\pi(Z)\,;\,Z\in\mathcal{S},\tfrac{X^{\alpha}}{Z}\in\mathcal{B}\right\}= \inf\left\{\pi\left(W^{\alpha}\right)\,;\,W\in\mathcal{S}^{\frac{1}{\alpha}},\,\tfrac{X^\alpha}{W^\alpha}\in\mathcal{B}\right\}\\
        &\leq \inf\left\{\pi(W^\alpha)\,;\,W\in\mathcal{S},\,\tfrac{X^\alpha}{W^\alpha}\in\mathcal{B}\right\}\leq \inf\left\{\pi(W)^{\alpha}\,;\,W\in\mathcal{S},\,\tfrac{X^\alpha}{W^\alpha}\in\mathcal{B}\right\}\\
        &= \inf\left\{\pi(W)^{\alpha}\,;\,W\in\mathcal{S},\tfrac{X}{W}\in\mathcal{B}^{\frac{1}{\alpha}}\right\}\leq \inf\left\{\pi(W)^{\alpha}\,;\,W\in\mathcal{S},\tfrac{X}{W}\in\mathcal{B}\right\}\\
        &= \eta_{\mf R}(X)^{\alpha}.
    \end{align*}
\end{proof}

\begin{remark}\label{sec:discussion_assumptions_logstarshaped}
    The assumption of $\mathcal C$ being log-star-shaped is satisfied if $\mathcal C=L^p_{++}$ with $p\in[1,\infty]\cup\{0\}$.
    
    In Section~\ref{sec:marrm_logReturns}, we focus on relative acceptance sets of the form $\mathcal B=\{X\in\mathcal{C}\,;\,\rho(\log(X))\le 0\}$ with $\rho$ being a star-shaped normalized monetary risk measure. These are log-star-shaped. Specifying $\rho$ as Value-at-Risk or Expected Shortfall delivers the relative acceptance sets used, e.g., in~\cite{mcneil_2000}.
    
    The assumption 
    of $\mathcal S$ being log-star-shaped is satisfied if $\mathcal{S}$ models all possible positive payoffs in a complete continuous-time financial market model. It is also satisfied in the situation of Section~\ref{sec:marrm_continuous_time}, where we use a Black-Scholes model with trading opportunities restricted to constant portfolio processes. 
    Additionally, $\mathcal S$ being log-star-shaped also holds for the Wishart volatility market in~\cite{baeuerle_2013}, which admits the famous Heston model as special case. 
    
    Finally, if the pricing functional is given by a (worst-case) expectation with respect to a (family of) pricing density, then the property of $\pi$ holds due to Jensen's inequality.
\end{remark}

As argued in~\cite[Section 3]{Laeven1}, natural properties of RRMs rewarding diversification in a context of continuously rebalanced portfolios are logconvexity and quasi-logconvexity. Continuous portfolio rebalancing means the {\em proportion} of wealth invested in each asset remains constant, while total wealth may vary over time.
This requires reallocating wealth across available assets over time, in contrast to buy-and-hold strategies, which are more straightforward in conjunction with monetary risk measures. The following proposition collects sufficient conditions to obtain (quasi-)logconvexity for MARRMs. 

\begin{proposition}\label{prop:quasiLogconvex}
    Assume that each set $\mathcal D\in\{\mathcal{C},\mathcal{B},\mathcal{S}\}$ is logconvex, i.e., 
    $$X,Y\in\mathcal D\text{ and }\alpha\in(0,1)\quad\implies\quad X^\alpha Y^{1-\alpha}\in\mathcal D.$$
    \begin{enumerate}[(i)] 
        \item Suppose that $\pi$ is quasi-logconvex, i.e., 
    $$\pi(X^\alpha Y^{1-\alpha})\le\max\{\pi(X),\pi(Y)\},\quad X,Y\in\mathcal C,\,\alpha\in(0,1).$$
    Then, $\eta_{\mf R}$ is also quasi-logconvex.
        \item Suppose that $\pi$ is logconvex, i.e., 
    $$\pi(X^\alpha Y^{1-\alpha})\le\pi(X)^{\alpha}\pi(Y)^{1-\alpha},\quad X,Y\in\mathcal C,\,\alpha\in(0,1).$$
    Then, $\eta_{\mf R}$ is also logconvex.
    \end{enumerate}
\end{proposition}

\begin{proof}
    The proof is analogous to the corresponding one of Proposition~\ref{prop:log-star-shaped}.
\end{proof}

\subsection{Finiteness and relative acceptability arbitrage}\label{sec:finitness}

Now, we discuss situations in which the MARRM does not attain $\infty$. 
In particular, this is the case in the studies~\cite{Return,Laeven1}.

\begin{lemma}\label{lem:finiteness}
Let $\mf R$ be a return risk measurement regime. Then $\eta_{\mf R}$ is finite-valued if and only if 
$$\mathcal C\subset\mathcal S\cdot\mathcal B.$$
If $\mathcal{C}= L^{\infty}_{++}$, then the latter condition is satisfied if $\mathcal S$ contains all positive multiples of an element in $\interior(L^\infty_{+}) = \{X\in L^{\infty}\,;\,\exists\,\eps>0: X\geq \eps\text{ a.s.}\}$ and $1\in\mathcal{B}$.
\end{lemma}
\begin{proof}Clearly, $\eta_{\mf R}(X)<\infty$ if and only if one can find $Z\in\mathcal S$ and $Y\in\mathcal B$ such that $X=Z\cdot Y$. Now assume that $\mathcal C=L^\infty_{++}$ and assume that $X\in\mathcal C$ is arbitary. For suitable $U\in \interior(L^\infty_{+})$ and $\lambda>0$, we have that $X\leq \lambda U$. By monotonicity of $\mathcal{B}$ and $1\in\mathcal{B}$, $\tfrac{X}{\lambda U}\in\mathcal{B}$. As $\lambda U\in\mathcal S$ by assumption, this means $\eta_{\mf R}(X)\le \pi(\lambda U)<\infty$. 
\end{proof}

The condition $\mathcal{C}\subset \mathcal{S}\cdot\mathcal{B}$ means that any loss in $\mathcal{C}$ can be made relatively acceptable by buying an appropriate market portfolio in $\mathcal{S}$.

If $\eta_{\mf R}(X)=0$, then there exists a sequence $\{Z_n\}\subset \mathcal{S}$ satisfying $\lim_{n\to\infty}\pi(Z_n)=0$ such that $\frac{X}{Z_n}\in \mathcal{B}$ holds for all $n\in\N$.
This means that market wealth relative to which the position $X$ is acceptable is available in the market at arbitrarily small initial investment. 
Hedging at arbitrarily small cost is known in the context of classical risk measures, as ``acceptable arbitrage'' opportunity, see e.g.,~\cite{artzner_risk_measures_2009}. This suggests that we call an $X$ with $\eta_{\mf R}(X)=0$ a ``\textit{relative acceptability arbitrage}'' opportunity. 
The next goal is to find conditions to avoid relative acceptability arbitrage opportunities. 
We start this discussion with the situation of positions in $L^{\infty}$. 
\begin{lemma}
    Let $\mathcal{S}=(0,\infty)$ and $\mathcal{C}=L^{\infty}_{++}$. For all $X\in\interior(L^{\infty}_+)$, it holds that $\eta_{\mf R}(X)>0$.
\end{lemma}

\begin{proof}
    By $\mathcal{B}$ being a nonempty proper subset of $\mathcal{K}$, there exists $Y\in\mathcal{K}\backslash\mathcal{B}$. By $X\in\interior(L^{\infty}_+)$, we know that there exists $\eps>0$ such that $X\geq \eps$ almost surely. Since $Y$ is bounded from above, we can find $0<\lambda^*<\eps$ such that for all $0<\lambda < \lambda^*$ it holds that $\frac{X}{\lambda}\geq \frac{\eps}{\lambda}\geq Y\notin\mathcal{B}$. This shows that $\eta_{\mf R}(X)\geq \lambda^*>0$.
\end{proof}

The next example shows that $X\notin\interior(L^{\infty}_+)$ can lead to $\eta_{\mf R}(X)=0$. 

\begin{example}\label{exam:marrmIsZero}
    Assume that the underlying probability space is atomless and consider the following  relative acceptance set:
    \begin{align*}
        \mathcal{B} = \{Y\in L^{\infty}_{++}\,;\, P(Y\leq 1)>0\}.
    \end{align*}
    It is easy to check that $\mathcal B$ is monotone, that $1\in\mathcal{B}$, and that $2\notin\mathcal{B}$. 
    Now, if $X$ is distributed uniformly over $(0,1)$ then for all $m\in(0,1)$ we have $P\left(\frac{X}{m}\leq 1\right)>0$, and hence $\frac{X}{m}\in\mathcal{B}$. 
    For $\mathcal{S} = (0,\infty)$ this implies that $\eta_{\mf R}(X)=0$. 
\end{example}

The previous example shows that even in the case of $\mathcal{C}=L^{\infty}_{++}$ we need additional assumptions on $\mathcal{S}$, $\pi$ and $\mathcal{B}$ to obtain $\eta_{\mathcal{R}}>0$. 
We suggest a combination of two assumptions which are sufficient to guarantee positivity of $\eta_{\mf R}$. 
The first, Assumption~\ref{assump:nflvr}, is motivated by the \textit{no free lunch with vanishing risk (NFLVR)} condition and concerns the financial market given by $\pi$ and $\mathcal{S}$. The second, Assumption~\ref{assump:relativeAcceptanceSet}, is our conclusion from  Example~\ref{exam:marrmIsZero} and concerns only the preferences of an agent given by $\mathcal{B}$.  

\begin{assumption}\label{assump:nflvr}
    For each sequence $\{Z_n\}\subset \mathcal{S}$ with $\pi(Z_n)\to 0$, $n\to\infty$, we have 
    $$P(0< Z_n)\rightarrow 0,\quad n\to\infty.$$
\end{assumption}

\begin{assumption}\label{assump:relativeAcceptanceSet}
    For each sequence $\{K_n\}\subset\mathcal{K}$ for which there exists $Y\in\mathcal{K}\backslash\mathcal{B}$ such that  $P(K_n> Y)\rightarrow 1$, $n\to\infty$, it follows that $K_n\notin \mathcal{B}$ for all but finitely many $n$.
\end{assumption}

\begin{remark}
    Assumption~\ref{assump:nflvr} on $\pi$ and $\mathcal{S}$ is satisfied if the underlying market model satisfies the NFLVR condition, which is essential for continuous-time financial market models, see~\cite{delbaen1994,delbaen1998,kreps_1981}. It guarantees that for every sequence of admissible payoffs $\{W_n\}$ at zero initial costs, it holds $\lVert\min\{W_n,0\}\rVert_{L^{\infty}}\rightarrow 0$ and $\lim\limits_{n\rightarrow\infty}P(W_n>0)=0$. This then implies $P(0< Z_n)\rightarrow 0$ in the situation of Assumption~\ref{assump:nflvr}. 

    Assumption~\ref{assump:relativeAcceptanceSet} says that loss fractions which are larger than an unacceptable loss fraction with probability close to one are unacceptable themselves. This prevents situations as in Example~\ref{exam:marrmIsZero}. 
\end{remark}

Next, under Assumption~\ref{assump:nflvr} we obtain strict positivity of the MARRM, if we exclude situations as in Example~\ref{exam:marrmIsZero} by Assumption~\ref{assump:relativeAcceptanceSet}.
\begin{theorem}\label{thm:noArbitrage_MARRM}
    If both Assumptions~\ref{assump:nflvr} and~\ref{assump:relativeAcceptanceSet} hold, then $\eta_{\mf R}>0$.
\end{theorem}

\begin{proof}
    Towards a contradiction, assume that there exists $X\in\mathcal{C}$ such that $\eta_{\mf R}(X) = 0$. Then, there exists a sequence $\{Z_n\}\subset \mathcal{S}$ with $\pi(Z_n)\rightarrow 0$ and $\frac{X}{Z_n}\in\mathcal{B}$ for all $n$. 

    By the non-triviality of $\mathcal{B}$, there exists $Y\in\mathcal{K}\backslash\mathcal{B}$. Further, by Assumption~\ref{assump:nflvr}, we obtain 
    \begin{align*}
        P\big(\tfrac{X}{Z_n}>Y\big) = 1 - P\big(\tfrac{X}{Z_n}\leq Y\big) = 1 - P\big(\tfrac{X}{Y}\leq Z_n\big)\rightarrow 1.
    \end{align*}
    Hence, by Assumption~\ref{assump:relativeAcceptanceSet}, we obtain that $\frac{X}{Z_n}\notin \mathcal{B}$ for all but finitely many $n$, a contradiction. 
\end{proof}

\section{Representation of MARRMs through multi-asset risk measures}\label{sec:marrm_logReturns}

We have already recalled from \cite{Return} that there is a close intrinsic connection between return and monetary risk measures.
For MARRMs, the counterpart thereof is a natural connection to multi-asset risk measures, which is the topic of the present section.

For our analysis we have to focus on relative acceptance sets defined via return risk measures as in~\cite{Return}. 
Return risk measures are positively homogeneous by definition. 
This requires  to make the following assumption throughout the present section:

\begin{assumption}
    $\mathcal{K}$ is a cone.
\end{assumption}

Now we recall the definition of a RRM from~\cite{Return}, but adjust the domain of a RRM.

\begin{definition}\label{def:RRM}
    If $L^{\infty}_{++}\subset\mathcal{K}$, then we call a positively homogeneous 
    and nondecreasing map $\tilde{\rho}\colon\mathcal{K}\rightarrow(0,\infty)$ with $\tilde{\rho}(1) = 1$ a  return risk measure (RRM). The relative acceptance set for a RRM $\tilde{\rho}$ is defined by
    $\mathcal{B}_{\tilde{\rho}}:=\{K\in\mathcal{K}\,;\,\tilde{\rho}(K)\leq 1\}.$
\end{definition}

\begin{remark}
\begin{enumerate}[(a)]
    \item Again, we assume that the domain of a RRM is a set of relative losses, i.e., dimensionless quantities.
    \item 
    We can write any RRM $\tilde{\rho}$ as an MARRM, by using $\mathcal{K}=L^{\infty}_{++}$, $\mathcal{B}=\{X\in\mathcal{K}\,;\,\tilde{\rho}(X)\leq 1\}$, $\mathcal{S} = (0,\infty)$ and $\pi(x)=x$ for all $x\in\mathcal{S}$. In particular, it holds that $\mathcal{C}=\mathcal{K}$.
    \end{enumerate}
\end{remark}
	
Now, we develop the desired connection between MARRMs and so-called multi-asset risk measures, thereby generalising the intimate connection between RRMs and monetary risk measures explored in~\cite{Return}.
We start with a MARRM and use it to derive a known risk measure formulation. To do so, we introduce sets containing logarithmic transformations of the random variables used in the definition of a MARRM. These sets are the following:
\begin{enumerate}[(1)]
    \item Logarithmic model set: $\mathcal{C}_{\log}:=\left\{\log(X)\,;\,X\in\mathcal{C}\right\}$;
    \item Set of log-securities: $\mathcal{S}_{\log} := \{\log(Z)\,;\, Z\in\mathcal{S}\}$;
    \item Set of logarithmic relative losses: $\mathcal{K}_{\log}:=\{\log(K)\,;\,K\in\mathcal{K}\}$.
\end{enumerate}

\begin{remark}
    Note that compatibility is equivalent to the condition that the Minkowski sum $\mathcal{C}_{\log}+\mathcal{S}_{\log}$ is a subset of $\mathcal{K}_{\log}$.
\end{remark}

From now on, for the rest of this section,
we use the following standing assumption:

\begin{assumption}\label{assump:group}
    Pointwise multiplication is a binary operation on $\mathcal{C}$, i.e., for all $X,Y\in\mathcal{C}$ it holds that $X Y\in\mathcal{C}$.
\end{assumption}
    
\begin{remark}
     Equivalent to Assumption~\ref{assump:group} is that the addition $+$ on $\mathcal{C}_{\log}$ is a binary relation, i.e., for all $X_{\log},Y_{\log}\in\mathcal{C}_{\log}$ it holds that $X_{\log} + Y_{\log}\in\mathcal{C}_{\log}$.
     Assumption~\ref{assump:group} ensures that we can analyze products of random variables as input for the MARRM. In doing so, a natural property for MARRMs is being submultiplicative, see Lemma~\ref{lem:connectionMARRMandMARM} and Remark~\ref{rem:here}. 
\end{remark}

Now, we recall the well-known definition of a MARM, see e.g.,~\cite{frittelli_2006,scandolo_2004}.

\begin{definition}
    Assume sets $\mathcal{C}^*,\mathcal{S}^*,\mathcal{K}^*\subset L^{0}$, $\mathcal{A}\subset\mathcal{K}^*$ and a map $\pi^*\colon\mathcal{S}^*\rightarrow\mathbb{R}$ are given. Further, assume that $\mathcal{A}$ is a nonempty proper subset of $\mathcal{K}^{*}$ (non-triviality) such that for all $X\in\mathcal{A}$, $Y\in\mathcal{K}^*$ with $Y\leq X$ it holds that $Y\in\mathcal{A}$ (monotonicity). The map $\rho_{\mathcal{A},\mathcal{S}^*,\pi^*}\colon\mathcal{C}^*\rightarrow[-\infty,\infty]$ defined by 
    \[\rho_{\mathcal{A},\mathcal{S}^{\ast},\pi^{\ast}}(X):=\inf\left\{\pi^{\ast}(Z^{\ast})\,;\,Z^{\ast}\in\mathcal{S}^{\ast},X-Z^{\ast}\in\mathcal{A}\right\}\]
    is called a multi-asset risk measure (MARM).
\end{definition}

\begin{remark}
    A MARM gives us the minimal costs of a hedging payoff that has to be added to the loss $X$ to obtain an acceptable position, i.e.,~an element of the set $\mathcal{A}$. So, in contrast to a MARRM, we add a position to the loss instead of dividing by it.
\end{remark}

Our next goal is to represent a MARRM via a MARM in the case in which the relative acceptance set is defined via a RRM.

\begin{proposition}\label{prop:MARRMasMARM}
    Let  $\tilde{\rho}\colon\mathcal{K}\rightarrow(0,\infty)$ be a RRM and
    $\mf R=(\mathcal{B}_{\tilde{\rho}}, \mathcal{S}, \pi)$
    be a return risk measurement regime. 
    Set $\nu=\log\circ\tilde{\rho}\circ \exp$, $\mathcal{A}_{\nu}=\{X_{\log}\in\mathcal{K}_{\log}\,;\,\nu(X_{\log})\leq 0\}$, and $\pi_{\log}=\log\circ\pi\circ \exp$ defined on $\mathcal S_{\log}$.
    Then $\eta_{\mf R}$ can be represented via the MARM $\rho_{\mathcal{A}_{\nu},\mathcal{S}_{\log},\pi_{\log}}$ as 
    \[\eta_{\mf R} = \exp\circ \rho_{\mathcal{A}_{\nu},\mathcal{S}_{\log},\pi_{\log}}\circ\log.\]
\end{proposition}

\begin{proof}
    For $X\in\mathcal{C}$, we obtain
    \begin{align*}
        \eta_{\mf R}(X)&=\inf\left\{\pi(Z)\,;\,Z\in\mathcal{S},\,\tfrac{X}{Z}\in\mathcal{B}_{\tilde{\rho}}\right\}=\inf\left\{\pi(Z)\,;\,Z\in\mathcal{S},\,\tilde{\rho}\big(\tfrac{X}{Z}\big)\leq 1\right\}\\
        &=\exp\left(\inf\left\{(\log\circ\pi\circ\exp)(\log(Z))\,;\,Z\in\mathcal{S},\,(\log\circ \tilde{\rho}\circ\exp)\big(\log(X)-\log(Z)\big)\leq 0\right\}\right)\\
        &=\exp\left(\inf\left\{\pi_{\log}(L)\,;\,L\in\mathcal{S}_{\log},\,\nu(\log(X)-L)\leq 0\right\}\right)\\
        &=\exp\left(\inf\left\{\pi_{\log}(L)\,;\,L\in\mathcal{S}_{\log},\,\log(X)-L\in\mathcal{A}_{\nu}\right\}\right)\\
        &=\exp\left(\rho_{\mathcal{A}_{\nu},\mathcal{S}_{\log},\pi_{\log}}(\log(X))\right).
    \end{align*}
\end{proof}

\begin{remark}
    Note that $\tilde\rho(1)=1$ implies $0\in\mathcal{A}_{\nu}$ and that, additionally, $\mathcal{A}_{\nu}$ is monotone, i.e., $\mathcal{A}_{\nu}-(\mathcal{K}_{\log})_{+}\subset\mathcal{A}_{\nu}$ with $(\mathcal{K}_{\log})_{+}=\{X_{\log}\in\mathcal{K}_{\log}\,;\, X_{\log}>0\text{ a.s}\}$. In particular, $\eta_{\mf R}$ and $\rho_{\mathcal{A}_{\nu},\mathcal{S}_{\log},\pi_{\log}}$ are increasing.
\end{remark}

Proposition~\ref{prop:MARRMasMARM} reformulates the MARRM as the exponential of a MARM, applied to a log-transformation. 
This representation is in line with the considerations in \cite[Section~3]{Return}. Additional care is necessary though. It is common to call $\rho_{\mathcal{A}_{\nu},\mathcal{S}_{\log},\pi_{\log}}$ a MARM only if $\mathcal{C}_{\log}$ is a linear space, $\mathcal{S}_{\log}$ is a linear subspace of $\mathcal{C}_{\log}$ and $\pi_{\log}$ is a strictly positive functional.

In order to be precise, we work until the end of this section in the setting of Proposition~\ref{prop:MARRMasMARM}. 

To state properties of the MARRM we have to characterize properties of the model ingredients. We omit the direct proofs of the next two results.
\begin{lemma}\label{lem:geometricProperties}
    Assume that $(0,\infty)\subset\mathcal{C}$. Then the following holds:
    \begin{enumerate}[(i)]
        \item $\mathcal{C}$ is a cone if and only if for all $X_{\log}\in\mathcal{C}_{\log}$ and $m\in\mathbb{R}$ it holds that $X_{\log}+m\in\mathcal{C}_{\log}$.
        \item $\mathcal{S}$ is a cone if and only if for all $X_{\log}\in\mathcal{S}_{\log}$ and $m\in\mathbb{R}$ it holds that $X_{\log}+m\in\mathcal{S}_{\log}$.
        \item If $\mathcal{S}$ is a cone, then  $\pi$ positively homogeneous if and only if $\pi_{\log}$ is cash-additive.
        \item $\mathcal{C}_{\log}$ is convex if and only if $\mathcal C$ is log-star-shaped.  
    \end{enumerate}
\end{lemma}

We can use this setting to generalize Lemma 2 (b), (d)--(f) in~\cite{Return}. 
\begin{lemma}\label{lem:connectionMARRMandMARM}
    Using the notation from Proposition~\ref{prop:MARRMasMARM},
    we obtain the following results:
    \begin{enumerate}[(i)]
        \item If $\mathcal{C}$ is a cone, then
        $\eta_{\mf R}$ is positively homogeneous if and only if 
        $\rho_{\mathcal{A}_{\nu},\mathcal{S}_{\log},\pi_{\log}}$ is cash-additive.
        \item If $\mathcal{C}$ is star-shaped, then $\eta_{\mf R}$ is star-shaped 
        if and only if 
        $$\forall\,X_{\log}\in\mathcal{C}_{\log}\,\forall\,r\in(-\infty,0):\quad \rho_{\mathcal{A}_{\nu},\mathcal{S}_{\log},\pi_{\log}}(X_{\log}+r) \leq \rho_{\mathcal{A}_{\nu},\mathcal{S}_{\log},\pi_{\log}}(X_{\log})+r,$$
        i.e., $\rho_{\mathcal{A}_{\nu},\mathcal{S}_{\log},\pi_{\log}}$ is cash-superadditive.
        \item$\rho_{\mathcal{A}_{\nu},\mathcal{S}_{\log},\pi_{\log}}$ is subadditive if and only if $\eta_{\mf R}$ is submultiplicative, i.e., $$\forall X,Y\in\mathcal{C}: \eta_{\mf R}(XY)\leq\eta_{\mf R}(X)\eta_{\mf R}(Y).$$
        \end{enumerate}
    For the next results, we assume that $\mathcal{C}_{\log}$ is convex.
    \begin{enumerate}[(i)]
    \addtocounter{enumi}{3}
        \item$\rho_{\mathcal{A}_{\nu},\mathcal{S}_{\log},\pi_{\log}}$ is convex if and only if $\eta_{\mf R}$ is logconvex. 
        \item $\rho_{\mathcal{A}_{\nu},\mathcal{S}_{\log},\pi_{\log}}$ quasi-convex if and only if $\eta_{\mf R}$ is quasi-logconvex.
        \item$\rho_{\mathcal{A}_{\nu},\mathcal{S}_{\log},\pi_{\log}}$ star-shaped if and only if $\eta_{\mf R}$ is log-star-shaped.         
        \item $\rho_{\mathcal{A}_{\nu},\mathcal{S}_{\log},\pi_{\log}}$ positively homogeneous if and only if 
        $$\forall\,X\in\mathcal{C}\,\forall\,\alpha>0:\quad\eta_{\mf R}(X^{\alpha})=\left(\eta_{\mf R}(X)\right)^{\alpha}.$$
    \end{enumerate}
\end{lemma}

\begin{remark}\label{rem:here} Point (a) in 
    Lemma~\ref{lem:connectionMARRMandMARM} establishes that the positive homogeneity of a MARRM is equivalent to cash-additivity of the associated MARM. 
    This generalizes the observation for RRMs and monetary risk measures in~\cite{Return}.

    Quasi-logconvexity in point (b) in Lemma~\ref{lem:connectionMARRMandMARM} describes the reward for diversification for continuously rebalanced portfolios. In this case, a payoff admits the following form: $Y = \exp(\omega_A r_A + \omega_B r_B)$, where $r_A$ and $r_B$ are the log-returns of two investment opportunities $A$ and $B$ and $\omega_A$ and $\omega_B$ are the portfolio weights, i.e.,~the percentage of capital invested in one of the investment opportunities. The latter implies that $\omega_A+\omega_B = 1$. So, if $\omega_A\in(0,1)$, then we have diversified our portfolio. Quasi-logconvexity of $\eta_{\mf R}$ then gives us that $\eta_{\mf R}(Y)\leq\max\{\eta_{\mf R}(\exp(r_A)),\eta_{\mf R}(\exp(r_B))\}$. Hence, quasi-convexity describes that the MARRM rewards diversification of relative amounts. Especially, in Section 3 in~\cite{Laeven1}, the authors conclude that quasi-convexity for classical risk measures is related to buy-and-hold strategies, whereas quasi-logconvexity is related to rebalanced portfolios. 
    
    A similar interpretation holds for submultiplicative MARRMs. For this, Assumption~\ref{assump:group} ensures that the product of two random variables $X,Y\in\mathcal{C}$ is still an element of $\mathcal{C}$. Note, the product of two losses corresponds to a squared amount of money. This is consistent with the upper bound, which is also a squared amount of money as a product of two MARRMs.
\end{remark}

\begin{remark}\label{rem:MARRM_from_MARM}
    In this section, we have rewritten a MARRM as the exponential of a MARM applied to log-returns. 
    To recap Equation (2)  in~\cite{Return} we can also go the other way around and start with a MARM and rewrite it as the logarithm of a MARRM. 
    To make this precise, we have to assume the following model ingredients of a MARM: 
    the logarithmic model set $\mathcal{C}_{\log}\subset L^{0}$, the logarithmic security set $\mathcal{S}_{\log}\subset L^{0}$, the pricing map $\pi_{\log}\colon\mathcal{S}_{\log}\rightarrow\mathbb{R}$, the set of logarithmic relative losses $\mathcal{K}_{\log}\subset L^0$ with $\mathcal{K}_{\log}+\mathbb{R}\subset\mathcal{K}_{\log}$. 
    Further suppose that $\nu\colon\mathcal{K}_{\log}\rightarrow[-\infty,\infty]$ is a cash-additive and nondecreasing map with $\nu(0)=0$. Then, we define the acceptance set as $\mathcal{A}_{\nu}=\{K_{\log}\in\mathcal{K}_{\log}\,;\,\nu(K_{\log})\leq 0\}$. Then, we obtain the ingredients of a MARRM by the following transformations:
    \begin{enumerate}[(1)]
        \item Model set: $\mathcal{C}=\left\{e^{X_{\log}}\,;\,X_{\log}\in\mathcal{C}_{\log}\right\}$;
        \item Security set: $\mathcal{S} = \left\{e^{Z_{\log}}\,;\, Z_{\log}\in\mathcal{S}_{\log}\right\}$;
        \item Pricing map: $\pi:\mathcal{S}\rightarrow(0,\infty),~Z\mapsto (\exp\circ\pi_{\log}\circ\log)(Z)$.
        \item Set of relative losses: $\mathcal{K}:=\left\{e^{K_{\log}}\,;\,K_{\log}\in\mathcal{K}_{\log}\right\}$;
        \item RRM: $\tilde{\rho} = \log\circ\nu\circ \exp$.
    \end{enumerate}
    Consider the return risk measurement regime $\mf R:=(\mathcal{B}_{\tilde{\rho}}, \mathcal{S}, \pi)$.
    Then, for $X_{\log}\in\mathcal{C}_{\log}$, by setting $X=\exp(X_{\log})$, we obtain $\rho_{\mathcal{A}_{\nu},\mathcal{S}_{\log},\pi_{\log}}(X_{\log}) = \log\left(\eta_{\mf R}(X)\right)$.
\end{remark}

Another remarkable feature of representing the MARRM via a MARM is that the domain of the MARM contains log-returns. 
This contrasts with the original intention of MARMs to measure the risk of monetary loss amounts, see e.g.,~\cite{artzner_risk_measures_2009, farkas_2015}. 
The interpretation remains unchanged though: The MARM determines the minimal log-price that is needed to reduce the log-return of our position via the log-return of a hedging position to an acceptable level. Note that, in our context, the MARM is merely an auxiliary function.

\section{Dual Representations of MARRM}\label{sec:dualRepresentation}

In this section, we obtain dual representations of MARRMs, by applying convex duality arguments to the corresponding MARMs. 
As a novel contribution, in Remarks~\ref{rem:diff_logconvex_logstarshaped} and~\ref{rem:diff_logconvex_quasilogconvex}, we discuss differences in the dual representations of logconvex/log-star-shaped/quasi-logconvex MARRMs based on the financial market given by $\pi$, $\mathcal{S}$, and the acceptability criterion given by $\mathcal{B}_{\tilde{\rho}}$. This analysis goes beyond the representations for RRMs in~\cite{Laeven1}.

\subsection{Assumption on $\mathcal{C}_{\log}$}

Our approach is similar to the one in~\cite{Aygun}. 
The first step is to choose $\mathcal{C}_{\log}$, while the set $\mathcal{C}$ is obtained as a consequence and not controlled directly.
This is also motivated by practitioners working with time-series models commonly using log-returns as input for the risk measure $\rho_{\mathcal{A}_{\nu},\mathcal{S}_{\log},\pi_{\log}}$. 
Hence, log-returns are the objects of interest, and it is therefore natural to start by choosing the set $\mathcal{C}_{\log}$. 

Our aim is to leverage the representation of a MARRM via a MARM and the existence of a dual representation of the latter to obtain a dual representation of the MARRM in question.
To do so and to perform convex duality arguments, we equip $\mathcal{C}_{\log}$ with a topology $\tau_{\log}$ and make the following standing assumption:
\begin{assumption}\label{assump:Clog}
    $(\mathcal{C}_{\log},\tau_{\log})$ is a locally convex topological vector space.
\end{assumption}

\begin{remark}
\begin{enumerate}[(a)]
     \item Although Assumption~\ref{assump:Clog} is not explicitly used in the upcoming results, it is crucial for applying known representations for convex, quasi-convex, and star-shaped functionals. For a short introduction on locally convex spaces we refer the reader to Aliprantis \& Border~\cite[Section 5.12]{aliprantis}.
    \item We obtain that $(\mathcal{C},\tau)$ with $\tau=\{\exp(\mathcal{A})\,;\,\mathcal{A}\in\tau_{\log}\}$ is a topological space and the exponential function gives us a homeomorphism between $(\mathcal{C}_{\log},\tau_{\log})$ and $(\mathcal{C},\tau)$.
    \end{enumerate}
\end{remark}

This assumption means that we choose a topology on $\mathcal{C}_{\log}$ and not on $\mathcal{C}$ directly. The following example contains the two most common choices for $\mathcal{C}_{\log}$ in this situation.

\begin{example}
    Set $\mathcal{C}_{\log}=L^p$ with $p\in[0,\infty]$ allowing for a locally convex topology on that space. We then obtain $\mathcal{C}=\{X\in L_{++}^{0}\,;\,\log(X)\in L^p\}$. Concrete examples are:
\begin{enumerate}[(1)]
    \item If $\mathcal{C}_{\log} = L^{\infty}$, then $\mathcal{C} = \interior(L_+^{\infty})$. This means that we only work with random variables that give us a sure minimum loss in the future.
    \item If $\mathcal{C}_{\log} = L^{1}$, then $\mathcal{C} = \{X\in L_{++}^{0}\,;\,\log(X)\in L^1\}$.
\end{enumerate}
In order to derive dual representations, we need a locally convex space though. 
Hence, the results in this section are relevant for the remaining examples, i.e.,~$\mathcal{C}_{\log} = L^p$ with $p\in[1,\infty]$, which are locally convex when equipped with the standard $L^p$-norm. 
Note, that the case of $\mathcal{C}_{\log} = L^p$, $p< 1$, typically does not satisfy Assumption~\ref{assump:Clog}.
\end{example}

\subsection{Dual representations}

For the next proof, we denote the topological dual space of a locally convex space $\mathcal{X}$ by $\mathcal{X}^{\prime}$. 
Further, a real-valued function $\varphi$ on a suitable set $\mathcal{C}\subset L^0_{++}$ is called log-linear, if $\varphi\left(X^{\alpha}Y^{\beta}\right) = \alpha\varphi(X)+\beta\varphi(Y)$ for all $X,Y\in\mathcal{C}$ and $\alpha,\beta\in\mathbb{R}$, assuming tacitly that such products lie in $\mathcal C$. 
For instance, if $\psi\colon \mathcal{C}_{\log} \to\R$ is linear, then $\psi\circ\log$ is a log-linear functional on $\mathcal C$. 
Conversely, in the upcoming proofs we set
$$\ph_{\log}(Y)=\ph(e^Y),\quad Y\in\mathcal C_{\log}.$$
By log-linearity of $\ph$, for all $\alpha\in\R$ and $Y,Y'\in\mathcal C_{\log}$, 
$$\ph_{\log}(Y+\alpha Y')=\ph(e^Y(e^{Y'})^\alpha)=\ph(e^Y)+\alpha\ph(e^{Y'})=\ph_{\log}(Y)+\alpha\ph_{\log}(Y').$$
Hence, each log-linear and $\tau$-continuous map $\varphi$ can be transformed by $\varphi\circ \exp$ into a linear and $\tau_{\log}$-continuous map.

For the remainder of this section, we use the notation from the previous section, i.e.,~$\tilde{\rho}\colon\mathcal{C}\rightarrow (0,\infty)$ is a RRM, $\nu=\log\circ\tilde{\rho}\circ \exp$, $\mathcal{A}_{\nu}=\{X_{\log}\in\mathcal{K}_{\log}\,;\,\nu(X_{\log})\leq 0\}$ and $\pi_{\log}=\log\circ\pi\circ \exp$ is defined on $\mathcal S_{\log}$.

We start with a dual representation for a logconvex MARRM.
    
\begin{proposition}\label{prop:dualLogconvexMARRRM} 
    Suppose that $\rho_{\mathcal{A}_{\nu},\mathcal{S}_{\log},\pi_{\log}}$ is proper, convex, and lower semicontinuous with respect to $\tau_{\log}$. Consider the associated return risk measurement regime $\mf R=(\mathcal B_{\tilde\rho},\mathcal S,\pi).$
    Then, for all $X\in\mathcal C$, the MARRM $\eta_{\mf R}$ can be represented as 
    \begin{align}\label{eq:rep_logc}
        \eta_{\mf R}(X) = \sup_{\varphi\in\mathcal{D}_{\geq 1}}\exp\left(\varphi(X)-\left(\sup_{Y\in\mathcal{B}_{\tilde{\rho}}}\varphi(Y)\right)-\left(\sup_{Z\in\mathcal{S}}\varphi(Z)-(\log\circ\pi)(Z)\right)\right),
    \end{align}
    where $\mathcal{D}_{\geq 1} = \{\varphi:\mathcal{C}\rightarrow\mathbb{R}\,;\, \varphi\text{ log-linear,  $\tau$-continuous}, \forall X\in\mathcal{C}\cap\mathcal{K}\text{ with }X\geq 1:\varphi(X)\geq 0\}.$
\end{proposition} 

\begin{proof}
    By the Fenchel-Moreau Theorem and following the reasoning of the proof of Proposition 3.9 in~\cite{frittelli_2006} we obtain for $X_{\log}\in\mathcal{C}_{\log}$ that 
    \begin{align*}
        \rho_{\mathcal{A}_{\nu},\mathcal{S}_{\log},\pi_{\log}}(X_{\log}) = \sup_{\psi\in(\mathcal{C}_{\log})^{\prime}}\left(\psi(X_{\log})-\left(\sup_{Y_{\log}\in\mathcal{A}_{\nu}}\psi(Y_{\log})\right)-\left(\sup_{Z_{\log}\in\mathcal{S}_{\log}}\psi(Z_{\log})-\pi_{\log}(Z_{\log})\right)\right).
    \end{align*}
    By the fact that $\mathcal{A}_{\nu}-(\mathcal{K}_{\log})_{+}\subset\mathcal{A}_{\nu}$, we can restrict the dual representation to functionals $\psi\in(\mathcal{C}_{\log})^{\prime}$ such that for all $X_{\log}\in(\mathcal{C}_{\log}\cap\mathcal{K}_{\log})_{+}$ it holds that $\psi(X_{\log})\geq 0$. By a reformulation in terms of the model ingredients of the MARRM, we obtain the desired representation.
\end{proof}

\begin{remark}
\begin{enumerate}[(a)]
\item
    Recall that the convexity of the MARM is equivalent to the logconvexity of the MARRM. 
    By Proposition~\ref{prop:quasiLogconvex}, logconvexity of the MARRM typically requires that the pricing map $\pi$ is logconvex. As already mentioned in Remark~\ref{sec:discussion_assumptions_logstarshaped}, this is for instance the case if the pricing map is defined for all $X\in\mathcal{C}$ via $\pi(X)=\sup_{Q\in\mathcal{Q}}\E_{Q}[X]$ for a set $\mathcal{Q}$ of equivalent martingale measures, due to H\"older's inequality.
\item The dual representation in Proposition~\ref{prop:dualLogconvexMARRRM} is of the same form as the one in~\cite[Proposition 31]{Aygun}. But, instead of defining a (multiplicative) Fenchel-Legendre conjugate for the MARRM $\eta_{\mf R}$, as it is done in~\cite[Definition 29]{Aygun}, we directly leverage the dual representation of the corresponding MARM. 
\item Note that the assumption of 
$\rho_{\mathcal{A}_{\nu},\mathcal{S}_{\log},\pi_{\log}}$ being proper implies $\eta_{\mf R}>0$, i.e., to apply of the Fenchel-Moreau Theorem, we require the assumption that  no relative acceptability arbitrage opportunities exist. This assumption is also used in the next theorem.
\end{enumerate}
\end{remark}

Now, we use two weaker assumptions than being logconvex. First, we assume that the MARRM is log-star-shaped as defined in the context of  Proposition~\ref{prop:log-star-shaped}. Second, we assume that the MARRM is quasi-logconvex; see Proposition~\ref{prop:quasiLogconvex}. 

We start with log-star-shapedness. To the best of our knowledge, it is the first time that a dual representation for this type of functionals is mentioned in the literature.

\begin{theorem}\label{thm:dualLogStarShapedMARRRM} 
    Let $\rho_{\mathcal{A}_{\nu},\mathcal{S}_{\log},\pi_{\log}}$ be proper, star-shaped, and normalized, i.e., $\rho_{\mathcal{A}_{\nu},\mathcal{S}_{\log},\pi_{\log}}(0) = 0$.
    Then the dual representation of the MARRM is given for $X\in\mathcal{C}$ by
    \begin{equation}\label{eq:rep_star}
        \eta_{\mf R}(X) = \min_{W\in\dom\left(\eta_{\mf R}\right)}\sup_{\varphi\in\mathcal{D}}\exp\Bigg(\varphi(X) -\sup_{\substack{Z\in\mathcal{S}\\ Y\in\mathcal{B}_{\tilde{\rho}}\\ W = YZ}}\max\left\{0,\varphi(Y)+\varphi(Z)-(\log\circ\pi)(Z)\right\}\Bigg),
    \end{equation}
    where $\mathcal{D} = \{\varphi:\mathcal{C}\rightarrow\mathbb{R}\,;\, \varphi\text{ log-linear,  $\tau$-continuous}\}=\{\ph_{\log}\circ\log\,;\,\ph_{\log}\in\mathcal (C_{\log})'\}$.
\end{theorem}

\begin{proof}
    By applying Proposition 8 in~\cite{Laeven2} we obtain for $X_{\log}\in\mathcal{C}_{\log}$ that
    \begin{align}\label{eq:proof_dualLogStarShapedMARRRM}
        \rho_{\mathcal{A}_{\nu},\mathcal{S}_{\log},\pi_{\log}}(X_{\log}) = \min_{W_{\log}\in\dom\left(	\rho_{\mathcal{A}_{\nu},\mathcal{S}_{\log},\pi_{\log}}\right)}\sup_{\psi\in(\mathcal{C}_{\log})^{\prime}}\left(\psi\left(X_{\log}\right)-f^{*}_{W_{\log}}(\psi)\right),
    \end{align}
    where $f^{*}_{W_{\log}}$ is the convex conjugate of the following function:
    \begin{align*}
        f_{W_{\log}}(X_{\log}) = \begin{cases}
            \alpha\rho_{\mathcal{A}_{\nu},\mathcal{S}_{\log},\pi_{\log}}(W_{\log}), &\exists\alpha\in[0,1]:X_{\log} = \alpha W_{\log},\\
            +\infty, &\text{otherwise}.
        \end{cases}
    \end{align*}
    For the convex conjugate of $f_{W_{\log}}$ we obtain
    \begin{align*}
        f^{*}_{W_{\log}}(\psi)&= \sup_{X_{\log}\in\mathcal{C}_{\log}}\left(\psi\left(X_{\log}\right)-f_{W_{\log}}(X_{\log})\right)= \sup_{\alpha\in[0,1]}\alpha\left(\psi\left(W_{\log}\right)-\rho_{\mathcal{A}_{\nu},\mathcal{S}_{\log},\pi_{\log}}(W_{\log})\right)\\
        &= \sup_{\alpha\in[0,1]}\sup_{\substack{Z_{\log}\in\mathcal{S}_{\log}\\ W_{\log}-Z_{\log}\in\mathcal{A}_{\nu}}}\alpha\left(\psi\left(W_{\log}\right)-\pi_{\log}(Z_{\log})\right)\\
        &= \sup_{\substack{Z_{\log}\in\mathcal{S}_{\log}\\ Y_{\log}\in\mathcal{A}_{\nu} \\ W_{\log} = Y_{\log}+Z_{\log}}}\max\{0,\psi\left(Y_{\log}\right)+\psi\left(Z_{\log}\right)-\pi_{\log}(Z_{\log})\}.
    \end{align*}

    Using this in combination with~\eqref{eq:proof_dualLogStarShapedMARRRM} and an analogous reformulation in terms of the model ingredients of the MARRM as it is done in the proof of Proposition~\ref{prop:dualLogconvexMARRRM} gives us the result.
\end{proof}

\begin{remark}\label{rem:diff_logconvex_logstarshaped}
    Let us compare the representations in Proposition~\ref{prop:dualLogconvexMARRRM} and Theorem~\ref{thm:dualLogStarShapedMARRRM}: 
    \begin{enumerate}[(a)]
    \item In the star-shaped case, we obtain an outer minimum, which stems from the fact that the star-shaped MARRM can be represented as minimum of convex functions.
    \item In \eqref{eq:rep_logc} for the logconvex case, the penalty term treats 
    the security set and the relative acceptance set separately. 
    In the star-shaped case, \eqref{eq:rep_star}, the outer minimum couples the two sets and links them with the condition $W=YZ$ appearing in the infimization. 
    \item Interpreting $\varphi$ as a model for the costs $\varphi(X)$ of a loss $X$, the additional summand in the representation is a penalty term which reduces the price of the loss. The penalty term is positive, if there exist an acceptable loss $Y$ and a market portfolio $Z$ that represents $W$ and the sum of the costs of $Y$ and the misspecification of the price of a market portfolio $Z$ by using $\varphi$ instead of the correct log-price $\log\circ\pi$ is positive. So, one can say, the penalty term evaluates, if the model $\varphi$ is plausible for a given loss $W$.    
    \end{enumerate}
\end{remark}

Now, we analyze the case of a quasi-logconvex MARRM. As in~\cite{Laeven1}, our result is based on the dual representation of quasi-convex functions in~\cite[Theorem 3.4]{volle_1998}. For a more compact version of this result, we refer to~\cite[Theorem 1.1]{frittelli_2011}. 

\begin{theorem}\label{thm:dualQuasiLogconvexMARRRM}
    Assume that $\rho_{\mathcal{A}_{\nu},\mathcal{S}_{\log},\pi_{\log}}$ is quasi-convex and lower semicontinuous. Then, the dual representation of the MARRM is given for $X\in\mathcal{C}$ by
    \begin{align}\label{eq:dual_logconvex}
        \eta_{\mf R}(X) =\sup_{\varphi\in\mathcal{D}}\inf\{\pi(Z)\,;\,Z\in\mathcal{S}, W\in\mathcal{B}_{\tilde{\rho}}, \varphi(W)\geq \varphi(X)-\varphi(Z)\},
    \end{align}
    with $\mathcal{D}$ as in Proposition~\ref{prop:dualLogconvexMARRRM}. 
\end{theorem}

\begin{proof}
    By applying Theorem 3.4 in~\cite{volle_1998} we obtain for $X_{\log}\in\mathcal{C}_{\log}$ that
    \begin{align}\label{eq:dual_volle}
        \rho_{\mathcal{A}_{\nu},\mathcal{S}_{\log},\pi_{\log}}(X_{\log}) = \sup_{\psi\in(\mathcal{C}_{\log})^{\prime}}\sup\Bigg\{r\in\mathbb{R}\,;\,\psi(X_{\log})> \sup_{\substack{Y_{\log}\in\mathcal{C}_{\log}\\ \rho_{\mathcal{A}_{\nu},\mathcal{S}_{\log},\pi_{\log}}(Y_{\log})<r}}\psi(Y_{\log})\Bigg\}.
    \end{align}
    As it is pointed out in Theorem 1.1 in~\cite{frittelli_2011} and can be shown that Equation~(\ref{eq:dual_volle}) is equal to
    \begin{align*}
        \rho_{\mathcal{A}_{\nu},\mathcal{S}_{\log},\pi_{\log}}(X_{\log}) = \sup_{\psi\in(\mathcal{C}_{\log})^{\prime}}\inf_{\substack{Y_{\log}\in\mathcal{C}_{\log}\\ \psi(Y_{\log})\geq \psi(X_{\log})}}\rho_{\mathcal{A}_{\nu},\mathcal{S}_{\log},\pi_{\log}}(Y_{\log}).
    \end{align*}
    Using the definition of the MARM, we obtain:
    \begin{align*}
        \rho_{\mathcal{A}_{\nu},\mathcal{S}_{\log},\pi_{\log}}(X_{\log}) &= \sup_{\psi\in(\mathcal{C}_{\log})^{\prime}}\inf_{\substack{Y_{\log}\in\mathcal{C}_{\log}\\ \psi(Y_{\log})\geq \psi(X_{\log})}}\inf_{Z_{\log}\in\mathcal{S}_{\log}}\Big(\pi_{\log}(Z_{\log})+\delta_{\mathcal{A}_{\nu}}(Y_{\log}-Z_{\log})\Big)\\
        &= \sup_{\psi\in(\mathcal{C}_{\log})^{\prime}}\inf_{Z_{\log}\in\mathcal{S}_{\log}}\inf_{\substack{W_{\log}\in\mathcal{A}_{\nu}\\ \psi(W_{\log})\geq \psi(X_{\log})-\psi(Z_{\log})}}\pi_{\log}(Z_{\log}).
    \end{align*}
    A reformulation in terms of the ingredients of the MARRM gives us~(\ref{eq:dual_logconvex}).
\end{proof}

\begin{remark}\label{rem:diff_logconvex_quasilogconvex}
    For $\varphi\in\mathcal{D}$ set $\sigma_{\mathcal{B}_{\tilde{\rho}}}(\varphi)=\sup_{Y\in\mathcal{B}_{\tilde{\rho}}}\varphi(Y)$. If the supremum is attained, then we can write 
    \begin{align*}
        \inf\{\pi(Z)\,;\,Z\in\mathcal{S}, W\in\mathcal{B}_{\tilde{\rho}},\varphi(W)\geq \varphi(X)-\varphi(Z)\}=\inf\{\pi(Z)\,;\,\sigma_{\mathcal B_{\tilde\rho}}(\varphi)\geq \varphi(X)-\varphi(Z)\}.
    \end{align*}
    We use this to compare Proposition~\ref{prop:dualLogconvexMARRRM} and  Theorem~\ref{thm:dualQuasiLogconvexMARRRM}: The representation~\eqref{eq:dual_logconvex} leads to the model $\varphi$ for which the minimum hedging costs $\pi(Z)$ are maximized subject to the existence of an acceptable loss with larger logarithmic costs (under $\varphi$) than the ones for $\frac{X}{Z}$. In particular, the objective is only the price $\pi(Z)$ of a market portfolio $Z$. 
        
    If the MARRM satisfies the stronger condition of log-convexity, then representation~\eqref{eq:rep_logc} holds. Here, for a model $\varphi$, we do not restrict the set of hedging payoffs by the inequality $0\geq \varphi(X)-\varphi(Z)-\sigma_{\mathcal{B}_{\tilde{\rho}}}(\varphi)$. Instead, $\varphi(X)-\varphi(Z)-\sigma_{\mathcal{B}_{\tilde{\rho}}}(\varphi)$ appears in the objective directly (as a penalty term). This means, a model $\varphi$ for which the logarithmic costs $\varphi(X)-\varphi(Z)$ are larger than the ones for an acceptable loss (under $\varphi$) leads to a penalization of $(\log\circ\pi)(Z)$.
\end{remark}

In summary, log-star-shapedness prevents separation of the market components $\mathcal{S}$, $\pi$, and the relative acceptance set $\mathcal{B}_{\tilde{\rho}}$ in the dual representation~\eqref{eq:rep_star}, and leads to a penalty for possible models $\varphi$ without rewards, due to the additional maximization with zero in~\eqref{eq:rep_star}.

For the dual representation~\eqref{eq:dual_logconvex} of quasi-logconvex MARRMs we obtain that only hedging costs are used as objective function and terms which occur as summands in the dual representation~\eqref{eq:rep_logc} of logconvex MARRMs only occur in a constraint.

\subsection{Duality without Assumption~\ref{assump:Clog}}

In this section, we discuss the impact of Assumption~\ref{assump:Clog}. Typical choices for $\mathcal{C}$ are different to the ones by working under Assumption~\ref{assump:Clog}. For instance, think of a RRM $\tilde\rho\colon L^\infty_{++}\to(0,\infty)$ as defined in Definition~\ref{def:RRM}. As recalled above, $\rho\colon \mathcal{C_{\log}}\to\R$ with $\mathcal{C}_{\log}=\{Y\in L^0\,;\, \exists\,c\in\mathbb{R}:Y\leq c \text{ a.s.}\}$ defined by $\rho(Y)=(\log\circ\tilde\rho\circ \exp)(Y)$. But,  $\mathcal{C}_{\log}$ is only a convex cone and not a linear space. Hence, Assumption~\ref{assump:Clog} does not hold. By the fact that $L^{\infty}\subset \mathcal{C}_{\log}$, the representations from the previous section only holds for the restricted function $\tilde{\rho}|_{\interior(L^{\infty}_{+})}$. 

The previous example shows that a dual representation for a MARM on $L^\infty$ gives us a representation for the corresponding MARRM with domain $L^{\infty}_{++}$ restricted to $\interior(L^{\infty}_{+})$. Analogously, a dual representation for a MARM on $L^p$ with $p\in[1,\infty)$ gives us a representation for the corresponding MARRM with domain $L^p_{++}$ restricted to the set $\{X\in L_{++}^{p}\,;\,\log(X)\in L^p\}$.

So, it is an open question how to obtain a dual representation if we start with $\mathcal{C}$ instead of $\mathcal{C}_{\log}$. 
If the MARRM is not only logconvex, but also convex, one idea is to apply convex duality directly to the MARRM and avoid the usage of the connection with MARMs. The main difficulty is that in the situation of $\mathcal{C}\subset L^1$, $\mathcal{C}$ is not a closed convex cone in a suitable $L^p$-norm topology. This means, even if the MARRM $\eta_{\mf R}$ is lower semicontinuous on $\mathcal{C}$, then it is not clear under which assumptions we can extend $\eta_{\mf R}$ to a lower semicontinuous function on the norm-closure of $\mathcal{C}$. 

We only demonstrate an example in which the latter is possible and for completeness, we also state the dual representation stemming from Proposition~\ref{prop:dualLogconvexMARRRM}.

\begin{example}
    Assume the RRM is given by the $L^1$-norm $\tilde{\rho}\colon L^1_{++}\ni X\mapsto \E[X]$. We can extend this functional to all of $L^1$ by
    \begin{align*}
        f\colon L^{1}\rightarrow [0,\infty],\quad X\mapsto\begin{cases}
            \E[X], &X\in L^{1}_{+},\\
            \infty, &X\notin L^{1}_{+}.
        \end{cases}
    \end{align*}
    In this case, $f$ is proper, convex, and lower semicontinuous. Hence, by the Fenchel-Moreau Theorem it holds that the Fenchel-Legendre biconjugate $f^{\ast\ast}$ of  $f$ is equal to $f$ itself. This delivers for all $X\in L^{1}_{++}$ the dual representation
    \begin{align}\label{eq:normDual}
        \tilde{\rho}(X) = f(X) = \sup\big\{\E[YX]\,;\,Y\in L^{\infty},\left\lVert Y\right\lVert_{\infty}\leq 1\big\}=\E[X].
    \end{align}
    Alternatively, we can apply Proposition~\ref{prop:dualLogconvexMARRRM} to obtain for all $X\in \{W\in L^{1}_{++}\,;\, \log(W)\in L^1\}$ that
    \begin{align}\label{eq:normDual_fromOurProp}
        \tilde{\rho}(X) = \sup\big\{\exp\big(\E[Y\log(X)]-H(Y|P)\big)\,;\,Y\in L_{+}^{\infty},E[Y]=1\big\},
    \end{align}
    where $H(Y|P) =\E[Y\log(Y)]$ is the relative entropy of $Y$ with respect to $P$. Thus, the results from the previous subsection only gives us a dual representation for elements of a specific subset. Further,~\eqref{eq:normDual_fromOurProp} is a representation based on the behaviour of the log-return of $X$. Hence, the representation admits a geometrical behaviour, due to the exponential function in this formula. Note, the relative entropy occurs due to the fact that $\log\circ\tilde{\rho}\circ \exp$ is the entropic risk measure. In contrast,~\eqref{eq:normDual} is obtained by directly applying the Fenchel-Moreau Theorem to $\tilde{\rho}$. Hence, the log-return of $X$ does not play a role here. 
\end{example}

\begin{remark}
    There are two main difficulties in applying the Fenchel-Moreau Theorem to a MARRM directly: (1) In general, it is not clear under which assumptions we can extend the MARRM defined on $\mathcal{C}$ to a {\em closed} convex cone such that the extended function is still lower semicontinuous. (2) Convexity of a MARRM is less common than log-convexity. 

    Regarding these difficulties, we leave the following points for future research: (1) What are conditions such that a MARRM on $\mathcal{C}$ can be extended to a proper, convex, lower semicontinuous function on a closed convex cone containing $\mathcal{C}$? (2) For non-convex but star-shaped MARRMs, is it possible to transfer the ideas behind dual representations for star-shaped monetary risk measures (see~\cite{castagnoli_2022,Laeven2}) to the case of MARRMs? The latter would result in a ``minmax'' representation as in Theorem~\ref{thm:dualLogStarShapedMARRRM}, in contrast to the simple supremum in \eqref{eq:normDual_fromOurProp}. 
\end{remark}

\section{{Empirical study}}\label{sec:marrm_continuous_time}

In this section, we calculate MARRMs with respect to a multivariate Black-Scholes market model and different relative acceptance sets. In this framework, we also compare the MARRM to the RRM and the MARM counterparts. 
The analysis is based on a data set which contains losses from a private auto insurance. The security payoffs are obtained from a Black-Scholes model.  
The structure of the case study is as follows: We begin by comparing MARRMs and RRMs. Here, we test the influence of different acceptability levels and covariance matrices of the underlying stocks (Figures~\ref{fig:var_arar_rrm} and~\ref{fig:var_arar_rrm_changes_parameters}). Then, we also illustrate differences and similarities between MARRMs and MARMs (Figure~\ref{fig:var_arar}).
In addition, we connect our case study to various aspects of risk sharing discussed in Section~\ref{sec:riskSharing}. 
We present the optimal hedging strategies obtained from the MARRM approach (Figure~\ref{fig:var_arar_portfolio_processes}) and compare the optimal risk sharing strategies dictated by MARM and MARRM in Figure~\ref{fig:risk_sharing_marrm_marm}.
For an example based on stock exchange data, we refer to Section A.3 of the online appendix.

First, we introduce the ingredients of the underlying stochastic market model. 
The multivariate Black-Scholes model consists of a riskless bank account $\{B_t\}_{t\in[0,T]}$ and a $d$-dimensional stochastic process $\{S_t\}_{t\in[0,T]}$ modeling the dynamics of risky stocks. For $d$-dimensional standard Brownian motion $W_t = (W_t^1,\dots,W_t^d)^{\intercal}$,  interest rate $r\in\mathbb{R}$, drift vector $b=(b^1,\dots,b^d)^{\intercal}$ and covariance matrix $\Sigma = (\sigma_{ij})_{i,j\in\{1,\dots,d\}}$, the dynamics are as follows:
\begin{align*}
    &\begin{cases}\diff B_t = B_t\, r \diff{t},\\
    B_0 = 1,\end{cases}\quad \text{and}\quad
    \begin{cases}
        \diff S_t^i = S_t^i(b_i \diff{t} + \sum_{j=1}^{d}\sigma_{ij}\diff{W_t^j}),\\
        S^{i}_0=1,
    \end{cases}\quad i\in\{1,\dots,d\}.
\end{align*}

For tractability, we use constant portfolio processes $\pi = (\pi^1,\dots,\pi^d)^{\intercal}\in\mathbb{R}^d$, where $\pi^i$ is the fraction of initial wealth invested in stock $i$. The latter corresponds to a constant trading strategy and is an assumption, e.g.,~in~\cite{Emmer}. 
Two reasons support this specification: 
First, restricting the set of admissible strategies to constant ones makes our problem tractable. Second, constant portfolio strategies are optimal for standard portfolio optimization problems in continuous time,  e.g., the solution of Merton's portfolio problem (see \cite{merton_1971}). For such a portfolio process and initial capital $x_0>0$, we obtain the following wealth process at time $t\in[0,T]$:
\begin{align}\label{eq:wealth_process}
    X^{x_0,\pi}_t = x_0\exp\left(\left(\pi^{\intercal}(b-r\mathbf{1}) + r -\tfrac{\left\lVert\pi^{\intercal}\Sigma\right\rVert^2}{2}\right)t + \pi^{\intercal}\Sigma W_t\right).
\end{align}

The parameter values are chosen as follows:
\begin{align*}
    &d = 2,\quad T =1,\quad r = 0.01,\quad b = (0.04,0.08)^{\intercal},\quad \Sigma = \begin{pmatrix}
        0.15 & -0.1\\
        -0.1 & 0.25
    \end{pmatrix}.
\end{align*}

The security sets for MARRM and MARM consists of all payoffs $X_T^{x_0,\pi}$ as in~(\ref{eq:wealth_process}). 

\subsection{US private auto claims}	

As a first example, we use the \texttt{usprivautoclaim} data from the \texttt{R}-package \texttt{CASdatasets}. 
This data set contains private motor insurance losses of a US insurance company. On the left-hand side of Figure~\ref{fig:insurance_data} we present the histogram of the losses. 
Additionally, we fit a lognormal and an exponential distribution to the data. 
The quantile-quantile(QQ)-plot on the right-hand side shows that the exponential distribution does not lead to a good fit. In contrast, the lognormal distribution seems to be a reasonable choice. 
This is line with the literature, see, e.g., Wüthrich~\cite[Section 3.2.3]{Wuthrich}. 
Furthermore, the lognormal distribution can be used to approximate the total claim amount for small portfolio sizes, see again~\cite[Section 4.1.2]{Wuthrich}. 

\begin{figure}
    \begin{subfigure}[b]{0.48\textwidth}
        \includegraphics[width=\linewidth]{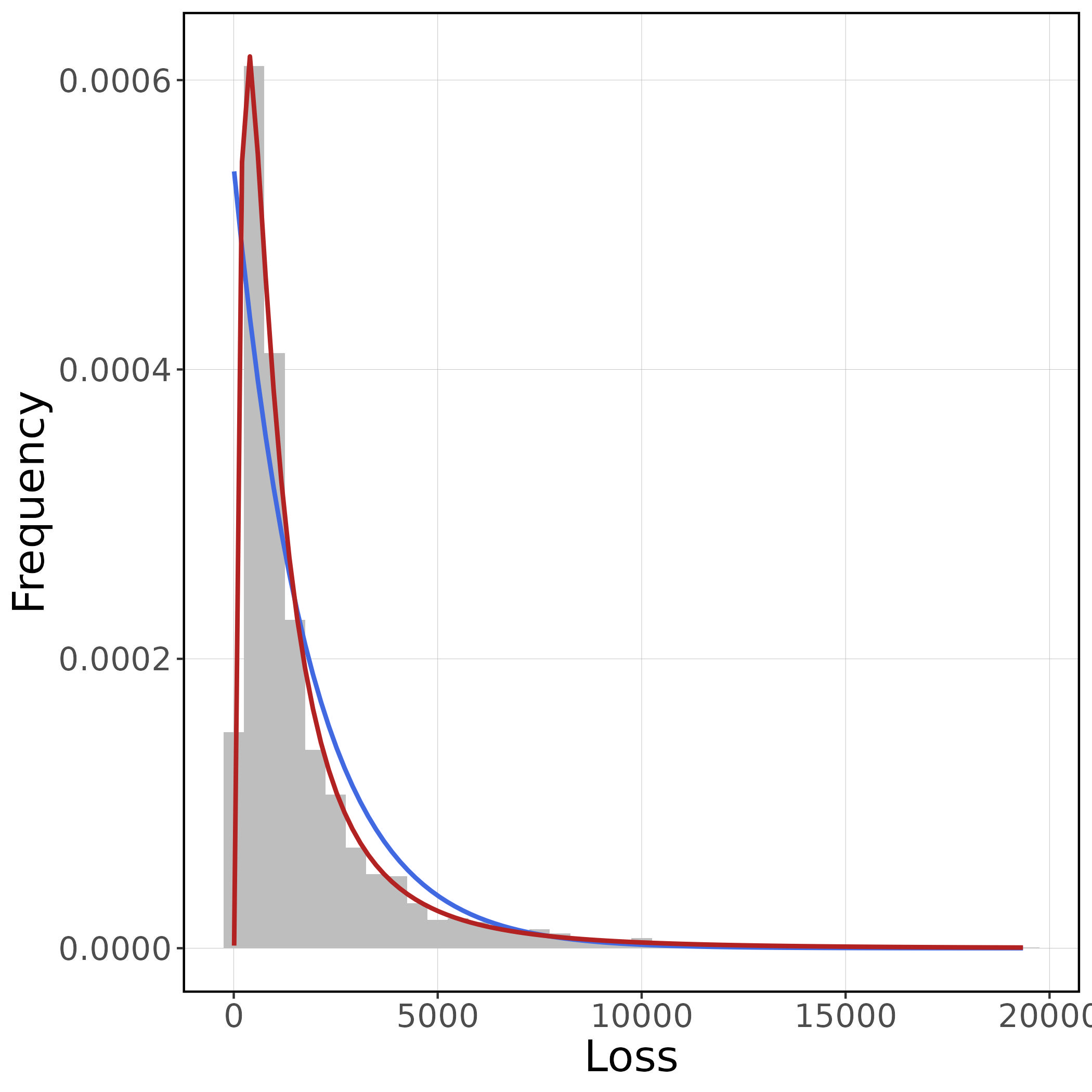}
    \end{subfigure}
    \hfill
    \begin{subfigure}[b]{0.48\textwidth}
        \includegraphics[width=\linewidth]{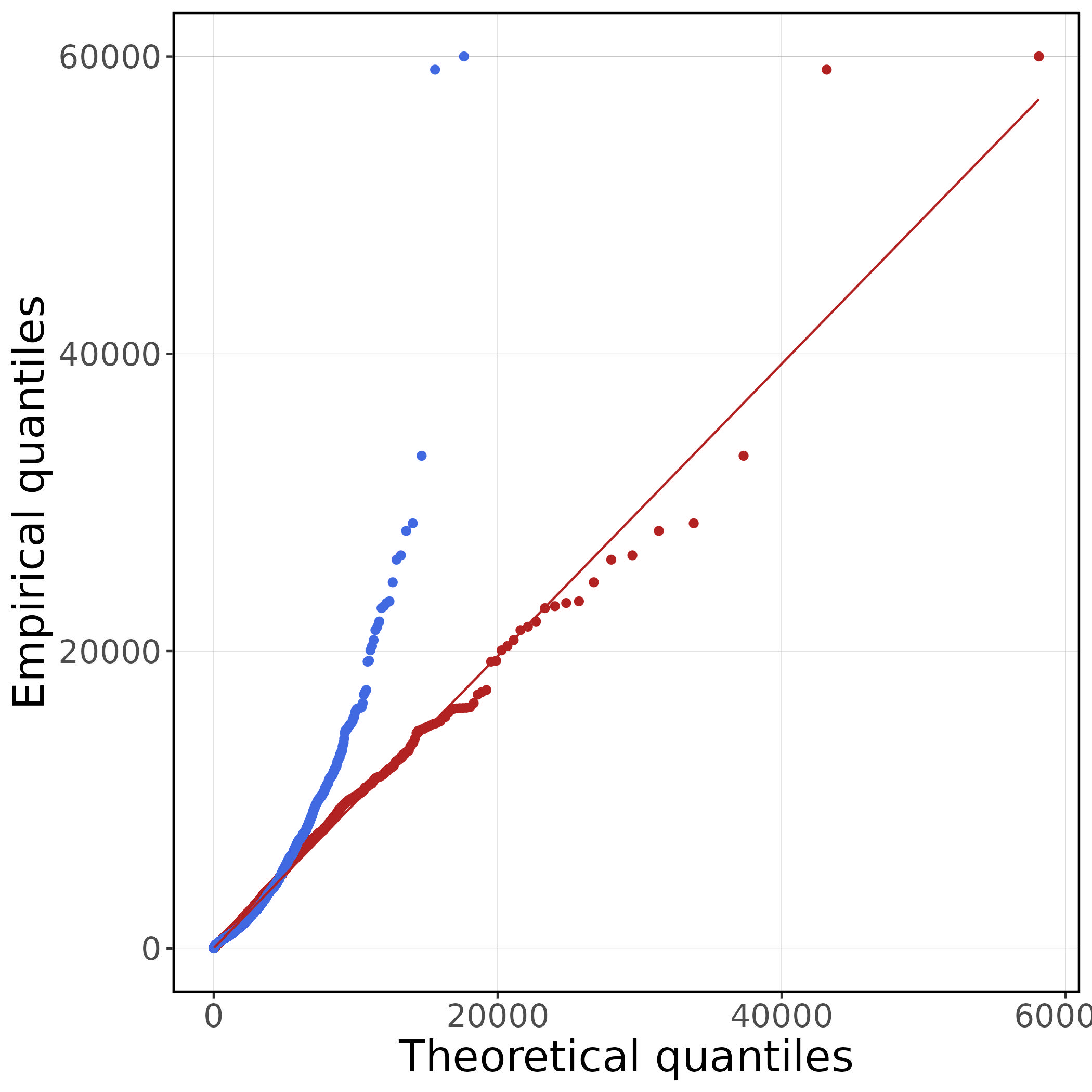}
    \end{subfigure}
    \captionsetup{font=footnotesize}
    \caption{\footnotesize\textit{LHS:} Histogram of losses below $20\,000$ from the \texttt{usprivautoclaim} data. The red, respectively blue, curve is the density of a lognormal, respectively exponential, distribution, fitted via method of moments. \textit{RHS:} QQ-plots of the empirical quantiles against the theoretical ones regarding the fitted lognormal (red) and exponential (blue) distribution.}
    \label{fig:insurance_data}
\end{figure}

So, from now on, we assume a risky position that is lognormally distributed, $X\sim\text{LN}(\mu,\sigma^2)$, where we specifically use the calibrated parameters of the lognormal distribution from Figure~\ref{fig:insurance_data}, $\mu = 6.9686$ and $\sigma = 1.0545$. In addition, we assume that $X$ is uncorrelated to the stocks. 
The latter is a meaningful assumption, because losses from the insurance business, such as those from car crashes, are independent of fluctuations on a stock exchange.
Under these specifications, we always end up with lognormally distributed relative losses:
\begin{align}\label{eq:distributionLognormalFraction}
    \frac{X}{X^{x_0,\pi}_T}\sim \text{LN}\big(\mu-\log(x_0)-\big(\pi^{\intercal}(b-r\mathbf{1}) + r -\tfrac{\left\lVert\pi^{\intercal}\Sigma\right\rVert^2}{2}\big)T,\sigma^2+\left\lVert\pi^{\intercal}\Sigma\right\rVert^2 T\big).
\end{align}

\subsection{Comparison with RRM} 

We begin by comparing MARRMs with their RRM counterparts. 
To the best of our knowledge, even in the case of MARMs, there are only few references that illustrate the influence of multiple eligible assets in concrete market setups. For instance, \cite{desmettre2020} compare monetary risk measures and MARMs in a one-period Black-Scholes setup, i.e.,~no intermediate trading is allowed. \cite{laudage2022} compare the two types of risk measures in a balance-sheet model with distributions inspired by the CAPM market model. Finally, \cite{frittelli_2006} compare MARMs for different security sets consisting of stochastic processes instead of random variables.

We use the following acceptability criteria for a strictly positive random variable $X$ and a level $\lambda\in(0,1)$:
\begin{enumerate}[(1)]
    \item Value-at-Risk (VaR) criterion:\quad $\exp\left(q^{-}_{\log(X)}(\lambda)\right)\leq 1$;
    \item Average-Return-at-Risk (ARaR) criterion:\quad$\exp\left(\frac{1}{1-\lambda}\int_{\lambda}^{1}q^{-}_{\log(X)}(u){\rm d}u\right) \leq 1$. 
\end{enumerate}

The ARaR is the RRM based on the Expected Shortfall as underlying monetary risk measure, see~\cite[Example 7]{Laeven1}. In particular, acceptability is captured by the single parameter $\lambda\in(0,1)$. 

Note, that the resulting relative acceptance sets satisfy Assumption~\ref{assump:relativeAcceptanceSet}. Together with the fact that the classical Black-Scholes model does not admit a free lunch with vanishing risk (Assumption~\ref{assump:nflvr}), we obtain by Theorem~\ref{thm:noArbitrage_MARRM} that the MARRMs do not lead to relative acceptability arbitrage opportunities.

In Figure~\ref{fig:var_arar_rrm}, we compare the MARRM to the RRM at varying levels of the acceptability parameter $\lambda$.  
In the left-hand figure, MARRM and RRM values are depicted as functions of $\lambda$.
In the right-hand figure, for VaR- and ARaR-based acceptability criteria, MARRM and RRM are compared in {\em relative} terms by looking at the fraction
\begin{equation}\label{eq:reldeviation}\tfrac{\text{RRM}- \text{MARRM}}{\text{MARRM}}\end{equation}
representing the relative deviation of RRM from MARRM.
We see that even for large levels, the deviation is above $6\%$. 
Hence, the introduction of multiple eligible assets indeed leads to a significant reduction of the minimal hedging costs.
Furthermore, the relative deviation in the VaR-case is larger than in the ARaR-case.
The gap is smaller for larger levels, which indicates that for such levels the optimal hedging strategies for the MARRMs in the VaR- and ARaR-case are similar.

\begin{figure}
    \begin{subfigure}[b]{0.48\textwidth}
        \includegraphics[width=\linewidth]{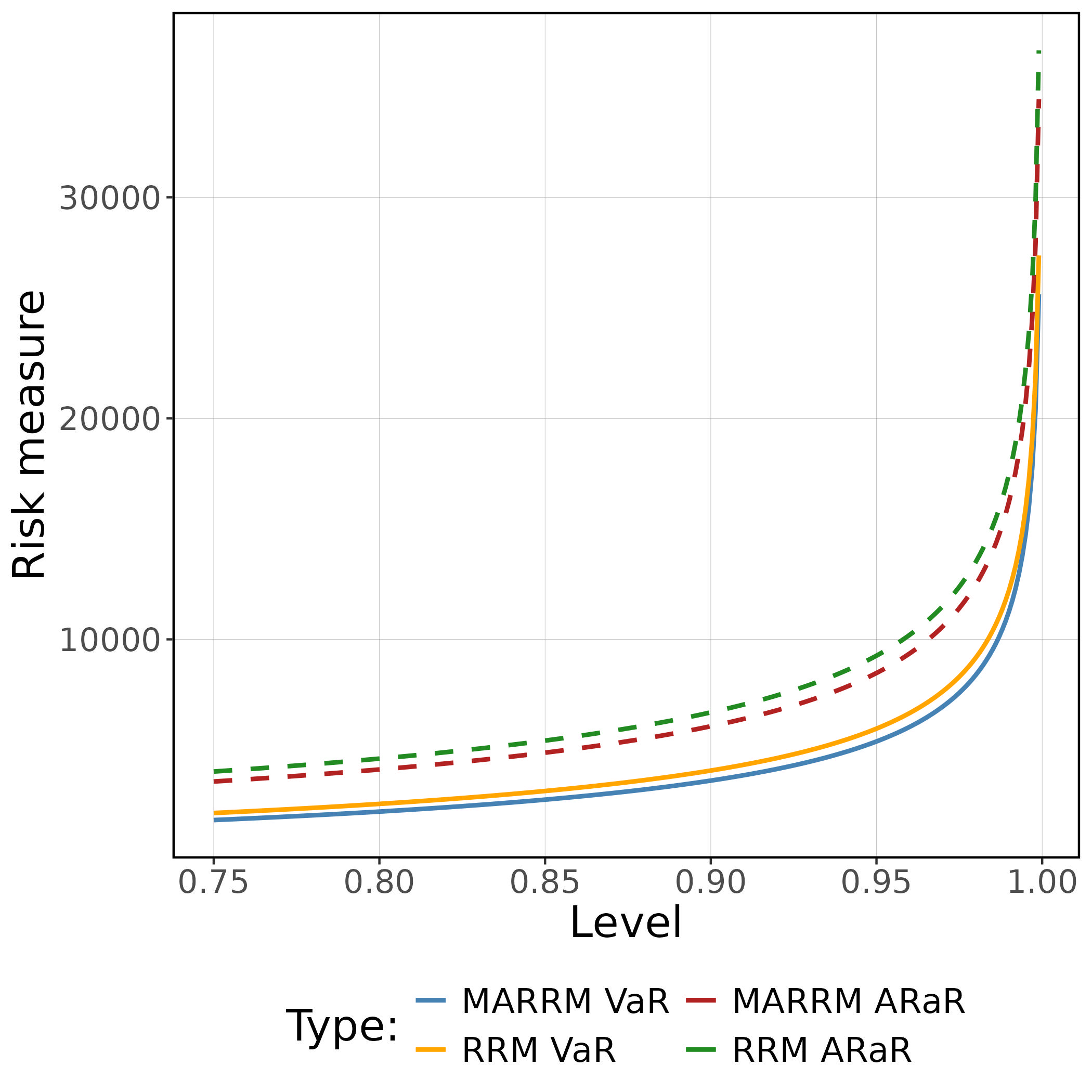}
    \end{subfigure}
    \hfill
    \begin{subfigure}[b]{0.48\textwidth}
        \includegraphics[width=\linewidth]{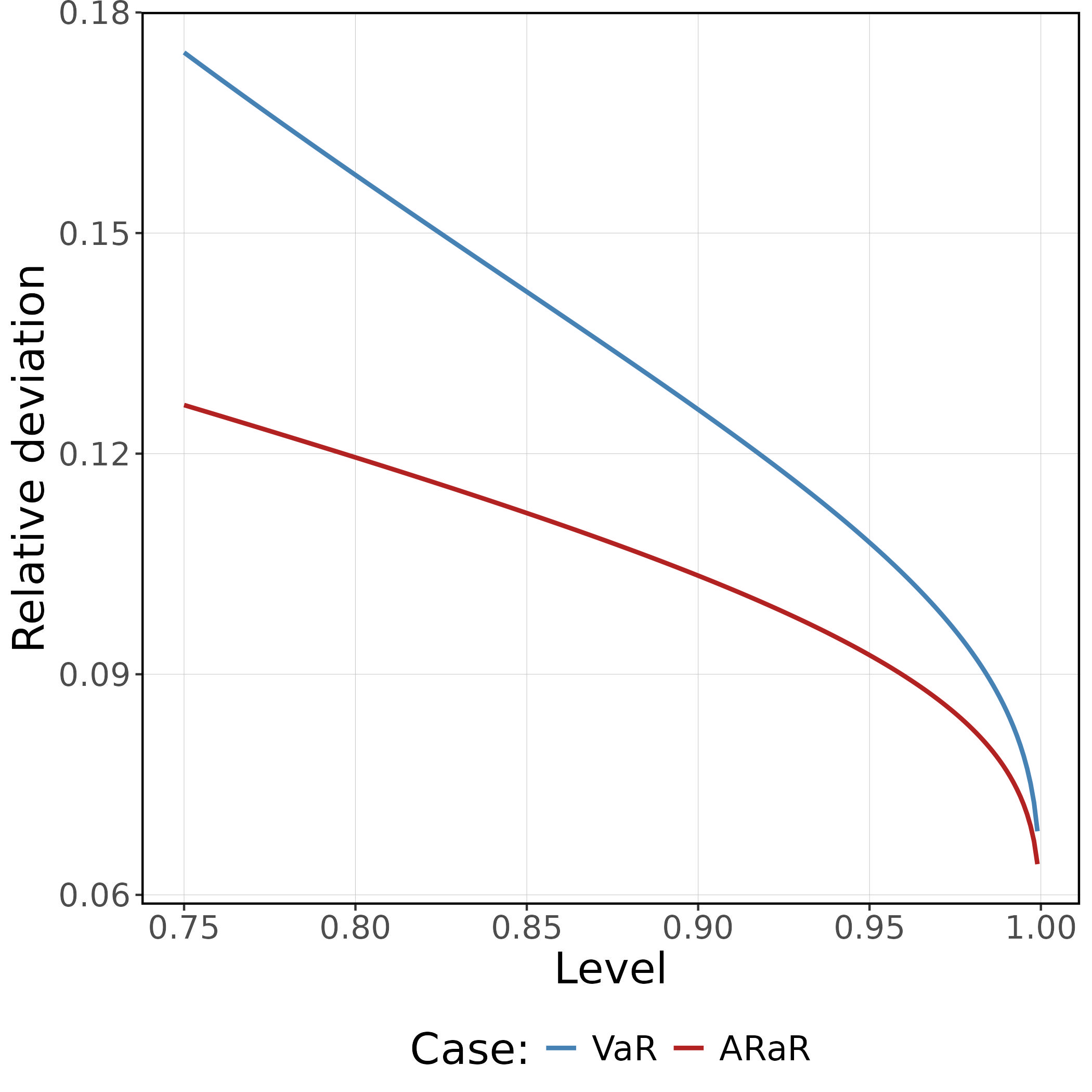}
    \end{subfigure}
    \captionsetup{font=footnotesize}
    \caption{\footnotesize\textit{LHS:}
    MARRM and RRM for different levels $\lambda$. \textit{RHS:} Relative deviation between MARRM and RRM for different levels $\lambda$ based on the same acceptance sets.}
    \label{fig:var_arar_rrm}
\end{figure}

In Figure~\ref{fig:var_arar_rrm_changes_parameters}, we test the sensitivity to the chosen parameters. To this effect, we fix $\lambda = 0.95$ and vary the volatility of the second stock $\sigma_{22}$ (left-hand side) and the correlation $\sigma_{12} = \sigma_{21}$ of the two stocks (right-hand side). The RRM remains unaffected by these changes, as it does not rely on stocks as hedging instruments. 
Further, we see that the MARRM goes to zero at values for which the matrix $\Sigma$ becomes singular. In this case, one can find a portfolio process for which the term $\left\lVert\pi^{\intercal}\Sigma\right\rVert^2 T$ in~(\ref{eq:distributionLognormalFraction}) is zero. 
By scaling this portfolio process, one can thus decrease the initial capital such that it is arbitrarily close to zero.   

For large values of $\sigma_{22}$ and $\sigma_{12}$ approaching zero, the MARRM and RRM values converge. 
In such cases, optimal hedging suggests avoiding investment in stocks.

\begin{figure}
    \begin{subfigure}[b]{0.48\textwidth}
        \includegraphics[width=\linewidth]{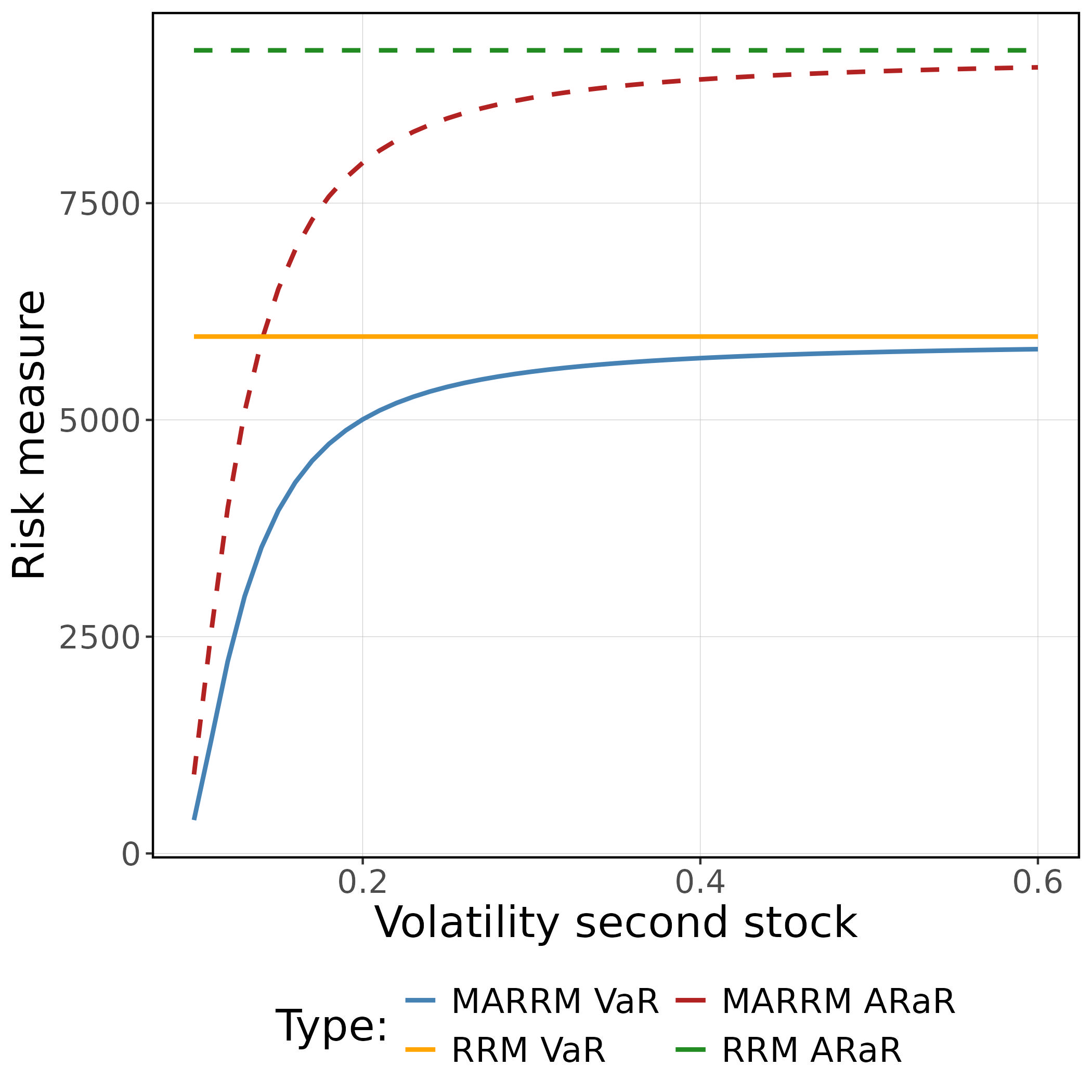}
    \end{subfigure}
    \hfill
    \begin{subfigure}[b]{0.48\textwidth}
        \includegraphics[width=\linewidth]{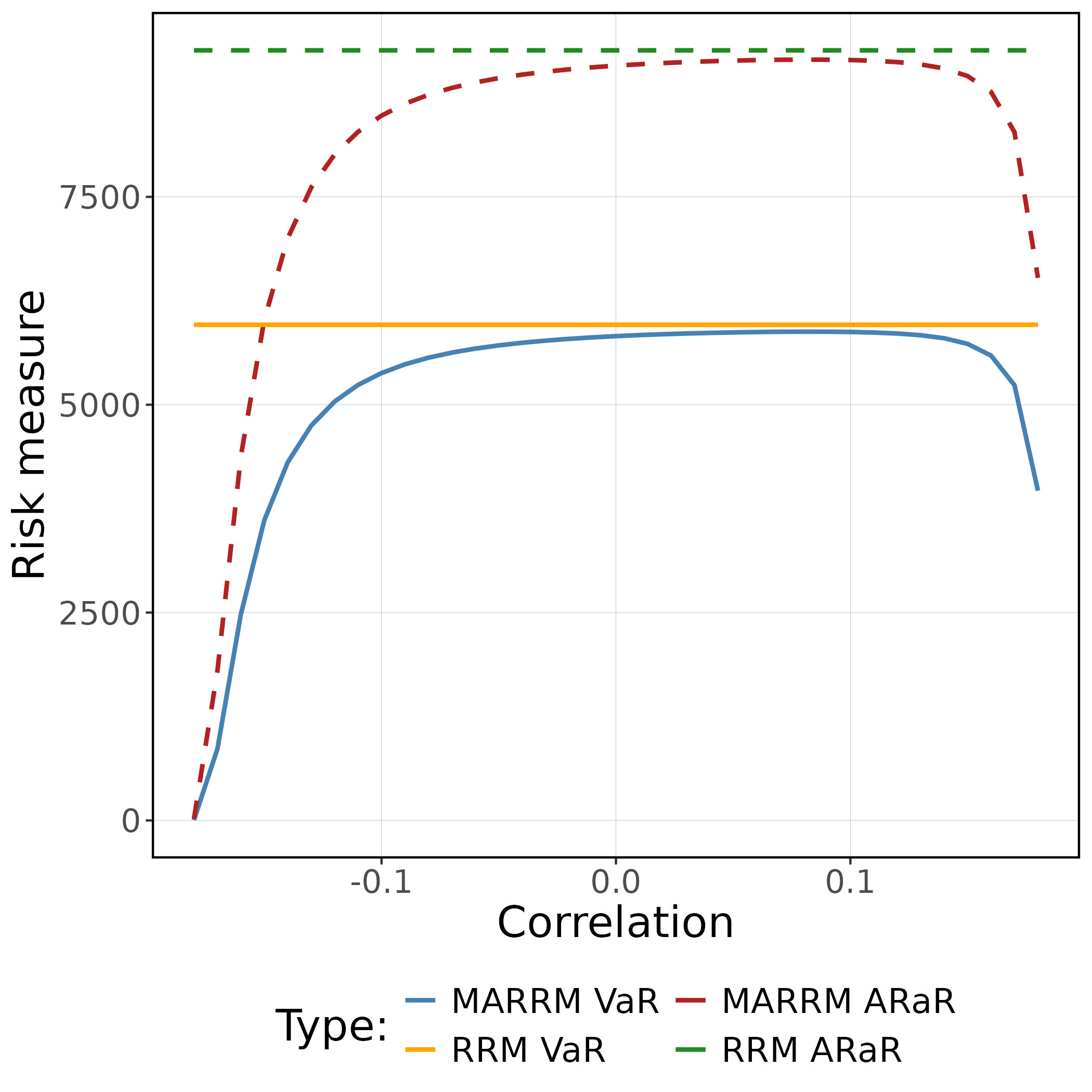}
    \end{subfigure}
    \captionsetup{font=footnotesize}
    \caption{\footnotesize \textit{LHS:} MARRM and RRM as a function of $\sigma_{22}$, the volatility of the second stock. \textit{RHS:} MARRM and RRM as a function of $\sigma_{12}$, the correlation between the stocks.} 
    \label{fig:var_arar_rrm_changes_parameters}
\end{figure} 

\subsection{Comparison with MARM}\label{sec:study_comparison_marm}

Now we compare the MARRM to a related MARM. When defining the latter, the following acceptability criteria for a random variable $X$ and a level $\lambda\in(0,1)$ will be used. 
\begin{enumerate}[(1)]
    \item Value-at-Risk (VaR) criterion for MARM:\quad $q^{-}_{X}(\lambda)\leq 0$;
    \item Expected Shortfall (ES) criterion for MARM:\quad$\frac{1}{1-\lambda}\int_{\lambda}^{1}q^{-}_{X}(u){\rm d}u \leq 0$.
\end{enumerate}

The MARRM for the ARaR and the MARM for the ES are generally not identical.\footnote{As an example, consider a position $X \sim \text{LN}(1.5, 0.16)$ and a security payoff $Z \sim \text{LN}(2.29, 0.04)$, and compare both risk measures numerically at level $\lambda = 0.9$.} 
However, as shown in the next result, they agree in the case of the VaR.

\begin{lemma}
    For $X,Z\in L^0_{++}$ and each $\lambda \in(0,1)$, $q_{\log\left(\frac{X}{Z}\right)}^{-}(\lambda)\leq 0$ holds if and only if $q_{X-Z}^{-}(\lambda)\leq 0$.
\end{lemma}

\begin{proof}
For any $Y\in L^0$ and $\lambda\in(0,1)$, $q^{-}_Y(\lambda)\le 0$ is equivalent to $P(Y\le 0)\ge \lambda$. It remains to observe that
$P\left(\log\left(\tfrac{X}{Z}\right)\leq 0\right)=P\left(\tfrac{X}{Z}\leq 1\right)=P\left(X\leq Z\right)=P(X-Z\leq 0).$
\end{proof}

On the left-hand side in Figure~\ref{fig:var_arar}, we plot the MARRM and MARM for varying levels of the parameter $\lambda$. 
Note that the blue and yellow curves overlap. The only curves that are visibly different are the green and the red ones. 
On the right-hand side, we plot the relative deviation between the two, i.e.,~$\frac{\text{MARRM}-\text{MARM}}{\text{MARM}}$, as a function of $\lambda$. 

The values in the VaR case differ only minimally, and the difference is solely due to different computational approaches. For MARM, we have to approximate the distribution of $X-Z$ via a Monte-Carlo simulation, while we know the exact value for the MARRM by~(\ref{eq:distributionLognormalFraction}).
Hence, the difference is fully explained by the Monte-Carlo approximation and the numerical solver used to obtain the minimal hedging price.\footnote{~As solver we use the function \texttt{nloptr} in \texttt{R} with the \texttt{NLOPT\_LN\_COBYLA}-algorithm.}

The previous point shows that a practitioner can use the MARRM to calculate the MARM in the VaR-case, which has the advantage that a computationally costly Monte-Carlo simulation can be avoided.

The curves in the ES/ARaR-case are not identical. So, in contrast to the VaR-case, a practitioner should be careful in using the MARRM as approximation of the MARM. 
In Section A.2 of the online appendix, we provide additional examples of acceptability criteria where MARRM and MARM differ significantly, with relative differences exceeding $30\%$.
    
\begin{figure}
    \begin{subfigure}[b]{0.48\textwidth}
        \includegraphics[width=\linewidth]{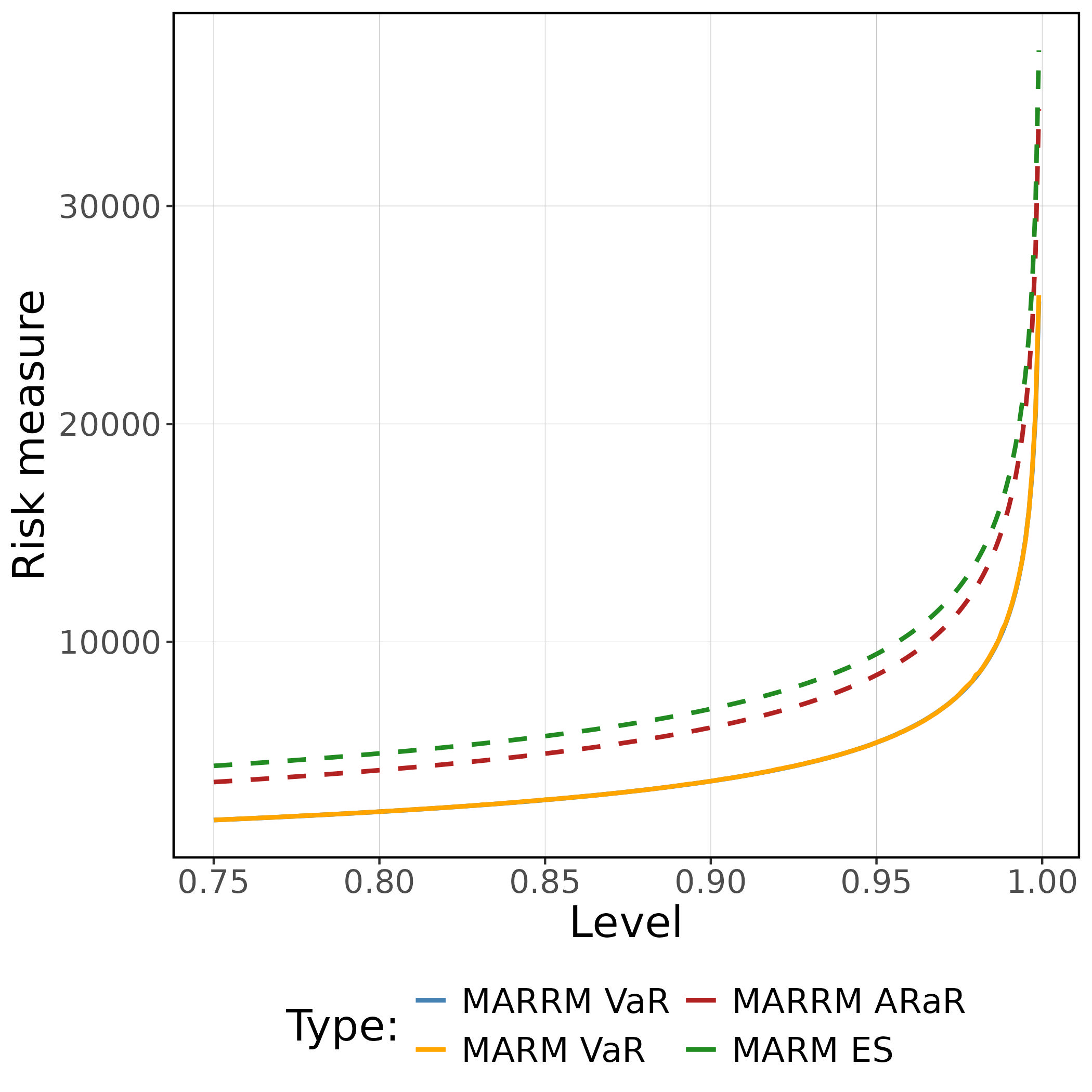}
    \end{subfigure}
    \hfill
    \begin{subfigure}[b]{0.48\textwidth}
        \includegraphics[width=\linewidth]{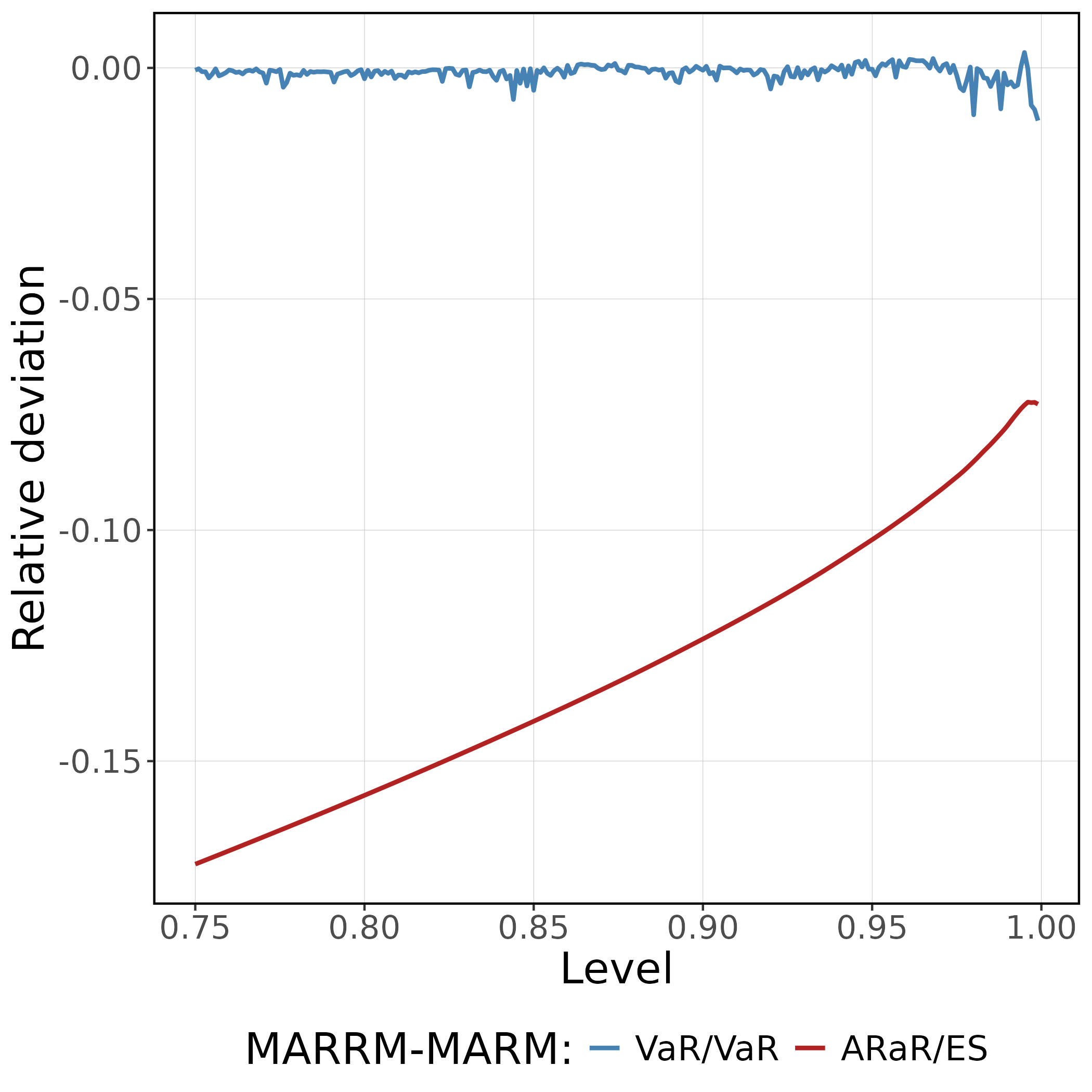}
    \end{subfigure}
    \captionsetup{font=footnotesize}
    \caption{\footnotesize \textit{LHS:} MARM and MARRM risk measures for different levels. 
    \textit{RHS:} Relative deviation between MARM and MARRM for the same acceptance sets.}
    \label{fig:var_arar}
\end{figure}

In Figure~\ref{fig:var_arar_portfolio_processes}, we return to the topic of risk sharing and the optimal portfolio process for the MARRM depending on the level $\lambda$ for VaR and ARaR.
For all levels, the optimal solution admits a short position in the bank account. This short position becomes smaller, if the level increases. We give a detailed explanation of these effects for the VaR-case in Section A.1 of the online appendix and show that the large short positions are mainly driven by the log-standard deviation $\sigma=1.0545$ of $X$. A situation in which smaller values of $\sigma$ occur is the one of logarithmic times series data of financial indices or stocks. In this case, it is possible that no short positions in the bank account occur. For a concrete example, we also refer to the online appendix.
We can interpret this effect as follows: a larger $\sigma$ results in a heavier tail of the lognormal distribution, which can be more effectively offset by other lognormal distributions than by a constant investment.

\begin{figure}
    \begin{subfigure}[b]{0.4\textwidth}
        \includegraphics[width=\linewidth]{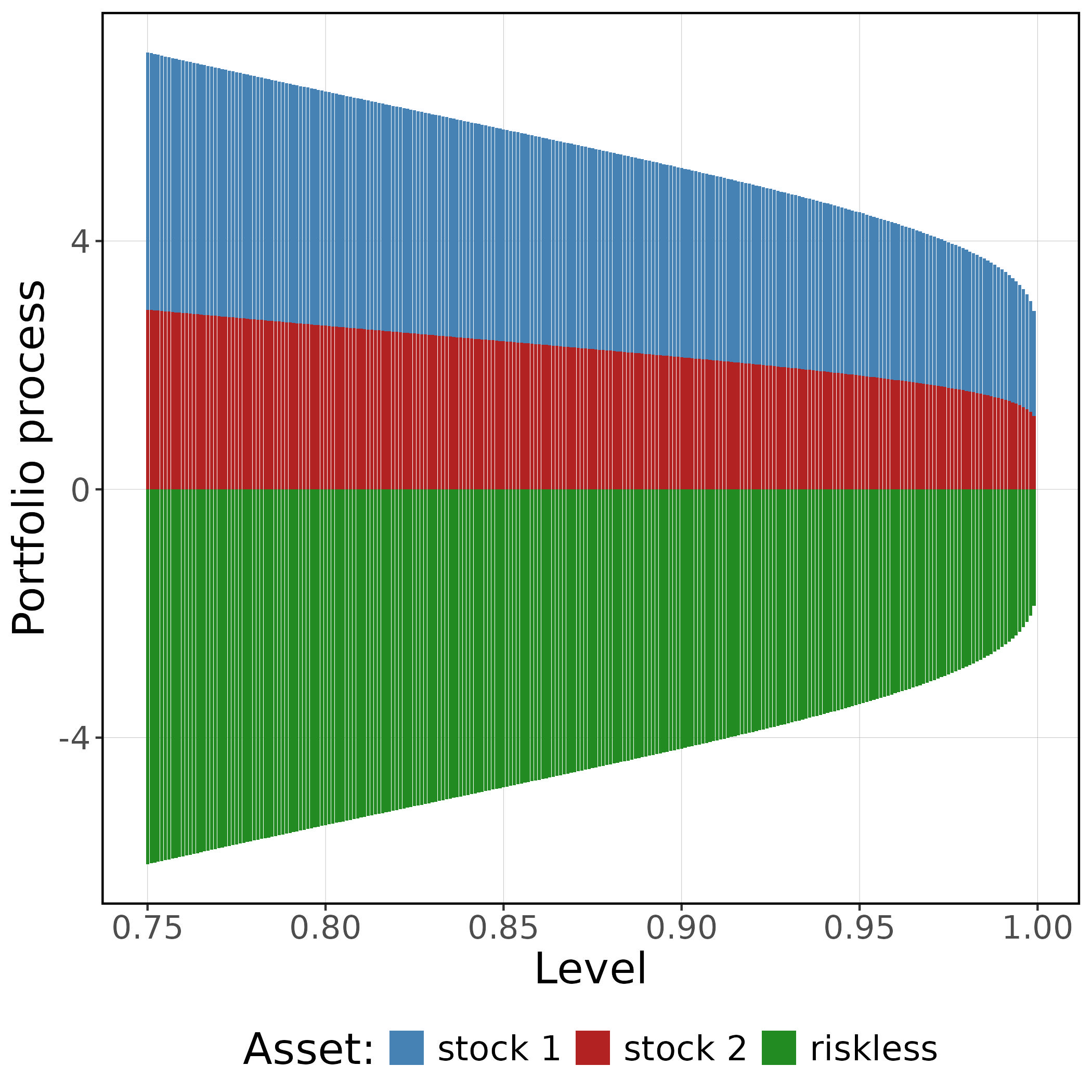}
    \end{subfigure}
    \hspace{1cm}
    \begin{subfigure}[b]{0.4\textwidth}
        \includegraphics[width=\linewidth]{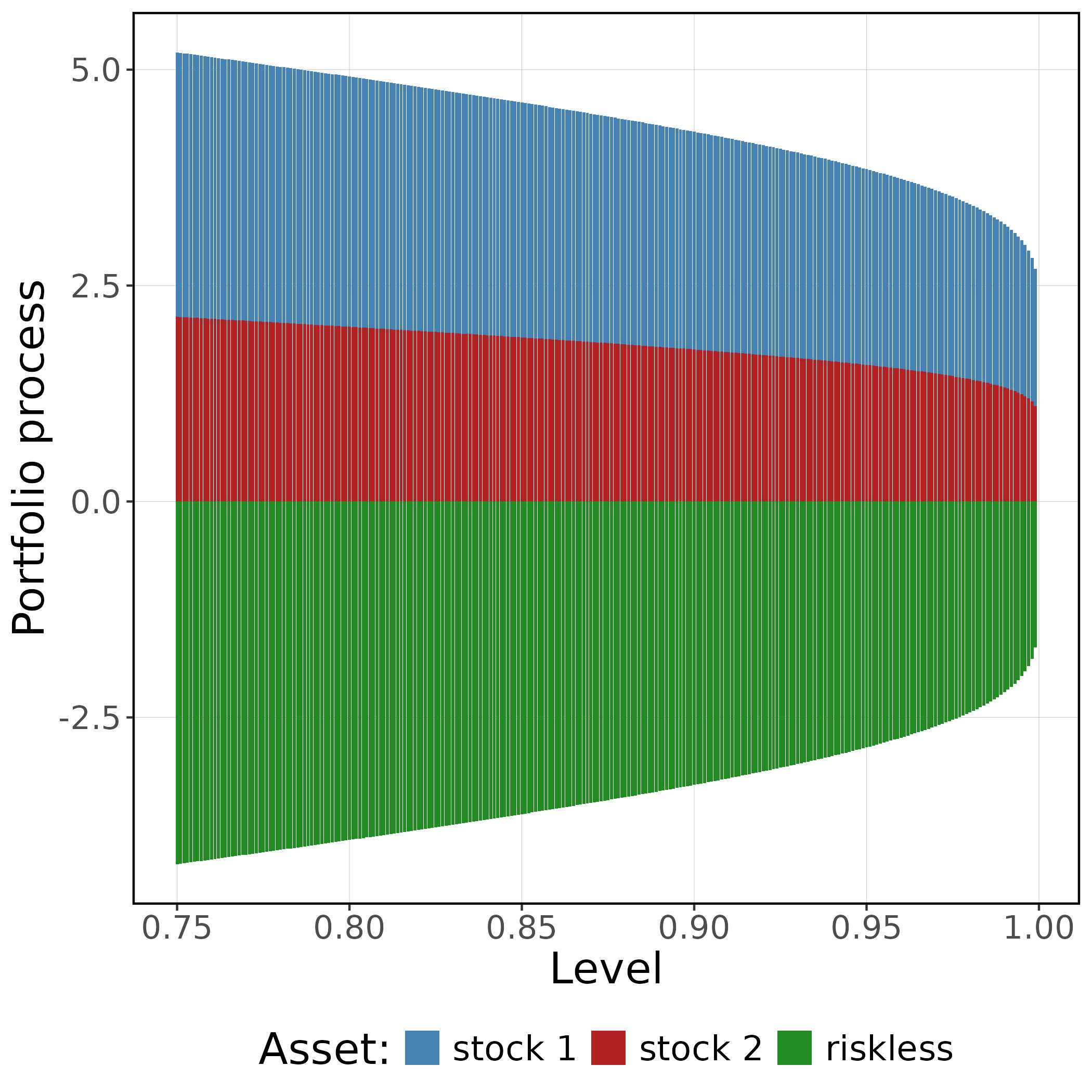}
    \end{subfigure}
    \captionsetup{font=footnotesize}
    \caption{\footnotesize \textit{LHS:} Optimal portfolio process for MARRM with VaR acceptance set. \textit{RHS:} Optimal portfolio process for MARRM with ARaR acceptance set.}
    \label{fig:var_arar_portfolio_processes}
\end{figure}

Finally, we would like to illustrate the risk sharing suggested by the concepts of MARRM and MARM. As mentioned in Section~\ref{sec:riskSharing}, the MARRM subdivides the loss $X$ multiplicatively into a security payoff $X^{x_0,\pi}_T$ and a leverage factor $X/X^{x_0,\pi}_T$ which lies in the relative acceptance set. Instead, the MARM leads to an additive decomposition of the loss $X$ into a security payoff $X^{x_0,\pi}_T$ and an acceptable loss $X-X^{x_0,\pi}_T$. In Figure~\ref{fig:risk_sharing_marrm_marm}, we illustrate on the left-hand side the densities of the loss $X$ and the security payoffs $X^{x_0,\pi}_T$ for the MARRM and MARM in the ARaR/ES-case at level $\lambda = 0.95$ and $\lambda = 0.99$. On the right-hand side, we present the densities of the leverage factor for the MARRM. Further, we add the leverage factor densities by using the optimal security payoff calculated via the MARM approach. 

The security payoff densities suggest that larger values are more likely under the MARM, explaining why its mode exceeds that of the MARRM density. In multiplicative risk sharing, the leverage factor of the MARM is more likely to produce values near zero, whereas the MARRM places more weight in the tail of the leverage factor.

\begin{figure}
    \begin{subfigure}[b]{0.48\textwidth}
        \includegraphics[width=\linewidth]{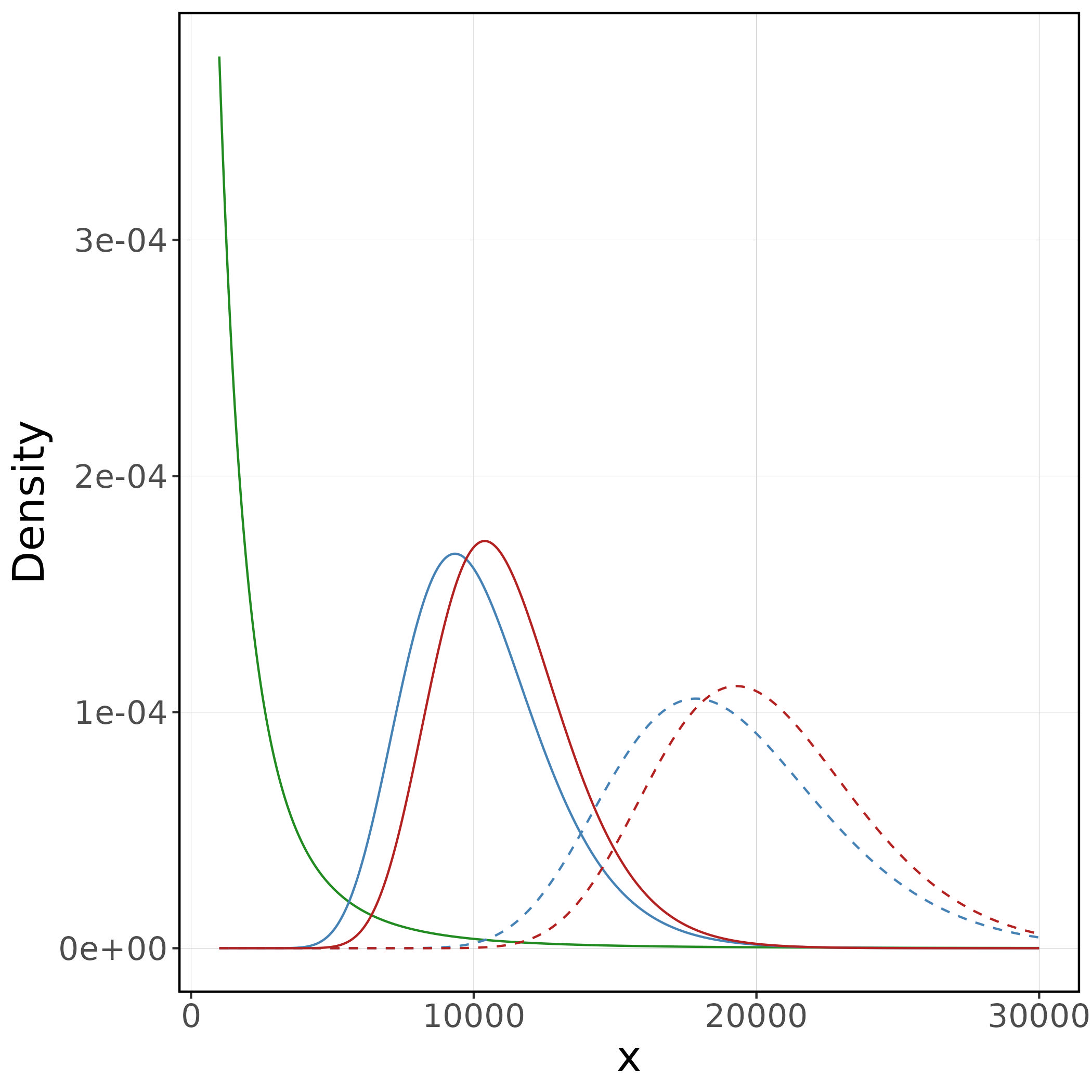}
    \end{subfigure}
    \hfill
    \begin{subfigure}[b]{0.48\textwidth}
        \includegraphics[width=\linewidth]{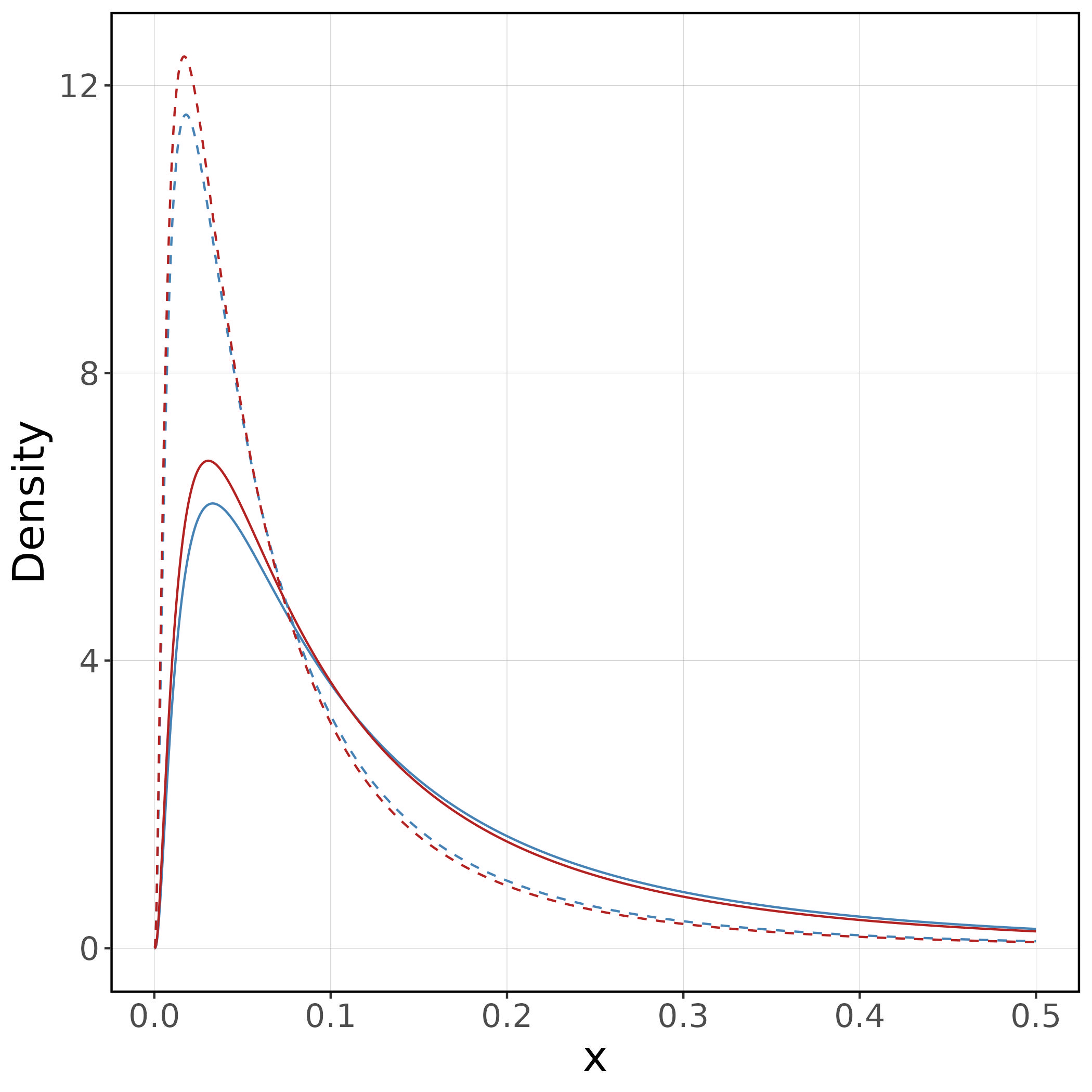}
    \end{subfigure}
    \captionsetup{font=footnotesize}
    \caption{\footnotesize In both figures, dashed lines correspond to $\lambda = 0.99$, solid lines to $\lambda = 0.95$. \textit{LHS:} Densities of loss $X$ (green) and the optimal security payoffs for MARRM (blue) and MARM (red). 
    \textit{RHS:} Densities of the leverage factor $X/X^{x_0,\pi}_T$ based on the optimal security payoffs in case of the MARRM and MARM.}
    \label{fig:risk_sharing_marrm_marm}
\end{figure}

\section{Conclusion}

We introduce the concept of multi-asset return risk measures (MARRMs) as an extension of the recently introduced return risk measures (RRMs). The philosophy behind a MARRM is the relative comparison of a loss random variable and a random variable representing the payoff of a hedging portfolio. The latter is created by trading in a financial market. This comparison is performed by a prespecified acceptability criterion. For instance, we can use the well-known concepts of VaR or Average-Return-at-Risk (ARaR) to obtain such a criterion. The latter is the RRM counterpart of the ES. The optimal hedging portfolio is then determined by minimizing its costs such that the acceptability criterion is satisfied. 

Theoretically, we prove that rewards for diversification in the classical sense (quasi-convexity of a functional) are only granted by a MARRM if it is convex. We also state an equivalent geometrical condition based on the acceptability criterion. As a second main finding, we identify assumptions on the acceptance set and the financial market which ensure that the MARRM provides a consistent risk evaluation, i.e., that it is strictly greater than zero. 
As a last theoretical contribution, we state several possibilities to represent MARRMs as the supremum over possible reference models, so-called dual representations.   

Finally, in an empirical study, we compare MARRMs to RRMs and so-called multi-asset risk measures (MARMs). The study is based on loss data from a US private auto insurance. 
In our example, we observe a deviation between MARRM and RRM based on VaR-/ARaR acceptability criteria which is larger than $6\%$ in all of the cases, even though the deviation decreases if the level increases. 
Further, MARRM and RRM differ significantly for a covariance matrix of the underlying stocks (used to hedge the loss) that is ``close to being singular''.

Furthermore, a MARRM allows us to avoid numerical approximations of the final payoffs in the VaR-case, if the payoffs are lognormally distributed. This is different from a MARM whose calculation is based on the sum of lognormally distributed random variables. 
The distribution of this sum needs to be approximated numerically. Hence, the MARRM admits a direct computation without relying on numerical procedures to calculate the value of the MARM.
Furthermore, the MARRM for the ARaR-case and the MARM for the ES-case diverge more significantly.
ly. Specifically, the relative differences exceed 5\% across all levels.
So, we need to be careful in using the MARRM as an approximation for the MARM in general. In addition, we illustrate the subdivision of the loss into a payoff and a leverage factor, which is a new economic concept of risk sharing. It splits the loss in a multiplicative way instead of an additive one. The latter is the philosophy behind MARMs.

The delicate problem of developing a direct convex duality approach for MARRMs remains open for future research. Furthermore, there are various options left to empirically test MARRMs, e.g.~using other stochastic market models and acceptability criteria. Finally, incorporating MARRMs into different portfolio optimization frameworks could be a promising direction of future research.

\section*{Competing interests}

The authors declare none.

\newpage

\appendix

{\centering\section*{ONLINE APPENDIX TO MANUSCRIPT\\ ``MULTI-ASSET RETURN RISK MEASURES''}}

\vspace{1cm}

\section{Additions to the empirical study}

Here, we provide additional information and extend the empirical study in the manuscript. In Section~\ref{sec:objectiveVaR},  we explain why, for large volatilities, the optimal payoff does not permit a short position in the bank account. 
In Section~\ref{sec:entropic}, we present another example, in which MARRM and MARM are significantly different. Finally, in Section~\ref{sec:empiricalStudy}, we test MARRMs for financial market data, namely the time series of the DAX. We find that the MARRM suggests larger stock positions in times of crisis. 

\subsection{VaR-case: objective of MARRM}\label{sec:objectiveVaR}

Here, we illustrate the objective function to calculate the MARRM for the VaR-case at level $\lambda = 0.95$ and the parameters used in the first case study. Recall that the VaR-based MARRM acceptability criterion is given by 
$$q^{-}_{\log\left(X/X_T^{x_0,\pi}\right)}(\lambda)\leq 0.$$  
The lognormal distribution of $\frac{X}{X_T^{x_0,\pi}}$ (see Equation (5.2) in the manuscript) then gives us
\begin{align}\label{eq:inequality_var}
    \log(x_0)\geq \mu-rT -\pi^{\intercal}(b-r\mathbf{1})T +\tfrac{\left\lVert\pi^{\intercal}\Sigma\right\rVert^2}{2}T+\sqrt{\sigma^2+\|\pi^{\intercal}\Sigma\|^2 T}\Phi^{-1}(\lambda) =: f(\pi),
\end{align}
where $\Phi^{-1}$ denotes the inverse of the cumulative distribution function of the standard normal distribution.
Hence, the logarithmic MARRM value is obtained by minimizing function $f$. This function is plotted in Figure~\ref{fig:var_objective}. 

\begin{figure}[h]
    \includegraphics[width=0.8\linewidth]{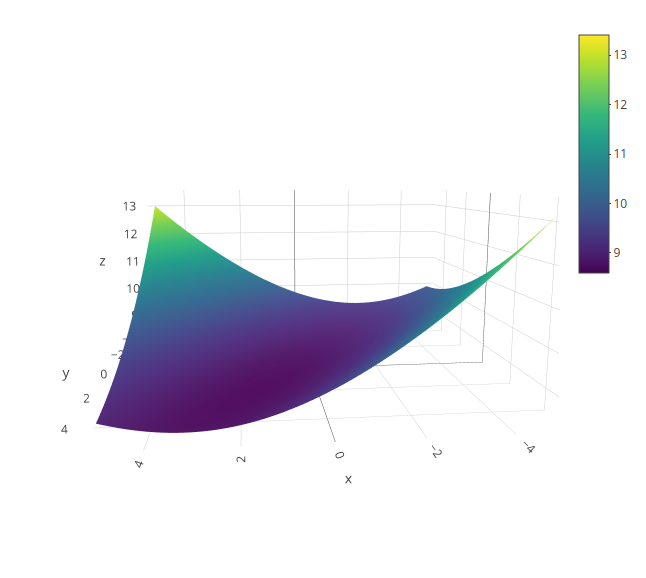}
    \captionsetup{font=footnotesize}
    \caption{\footnotesize Minimized value of this function gives us the logarithmic value of the MARRM with VaR acceptability criterion. The $x$-and $y$-axes are the possible values of $\pi_1$ and $\pi_2$.}
    \label{fig:var_objective}
\end{figure}

Analyzing the function $f$ in~(\ref{eq:inequality_var}) leads to the following observations: 
Since all entries of $b-r\mathbf{1}$ are positive, the minimum is not attained if all values of $\pi$ are negative. A mixture of positive and negative entries of $\pi$ gives us the largest values in Figure~\ref{fig:var_objective}, due to the negative correlation in $\Sigma$. Hence, the smallest values of $f$ are achieved by admitting only positive entries of $\pi$. Now, the large value of $\sigma$ (recall $\sigma>1$) means that the term $\sqrt{\sigma^2+\|\pi^{\intercal}\Sigma\|^2 T}$ is dominated by $\sigma^2$ if $\|\pi^{\intercal}\Sigma\|$ is small. Note that $\sqrt{\sigma^2+\|\pi^{\intercal}\Sigma\|^2 T}$ is minimized at zero and the slope of the function $x\mapsto \sqrt{\sigma^2+x^2 T}$ is smaller for larger values of $\sigma$. Now, the remaining term $\mu-rT -\pi^{\intercal}(b-r\mathbf{1})T +\tfrac{\left\lVert\pi^{\intercal}\Sigma\right\rVert^2}{2}T$ is minimized for at least one strictly positive entry of $\pi$. So, the decrease of the term
\begin{align*}
    \pi\mapsto \mu-rT -\pi^{\intercal}(b-r\mathbf{1})T +\tfrac{\left\lVert\pi^{\intercal}\Sigma\right\rVert^2}{2}T
\end{align*}
is less compensated by an increase of 
\begin{align}\label{eq:termObjectiveVaR}
    \pi\mapsto \sqrt{\sigma^2+\|\pi^{\intercal}\Sigma\|^2 T}\Phi^{-1}(\lambda),
\end{align}
if $\sigma$ is large. Hence, the minimum is achieved for large values of the portfolio process $\pi$. 

In the situation of the empirical study, Figure~\ref{fig:var_arar_portfolio_processes_no_short_position} shows that short positions in the bank account can be avoided, if the loss $X$ admits a smaller value for $\sigma$ than the one used in the manuscript.

\begin{figure}[h]
    \begin{subfigure}[b]{0.48\textwidth}
        \includegraphics[width=\linewidth]{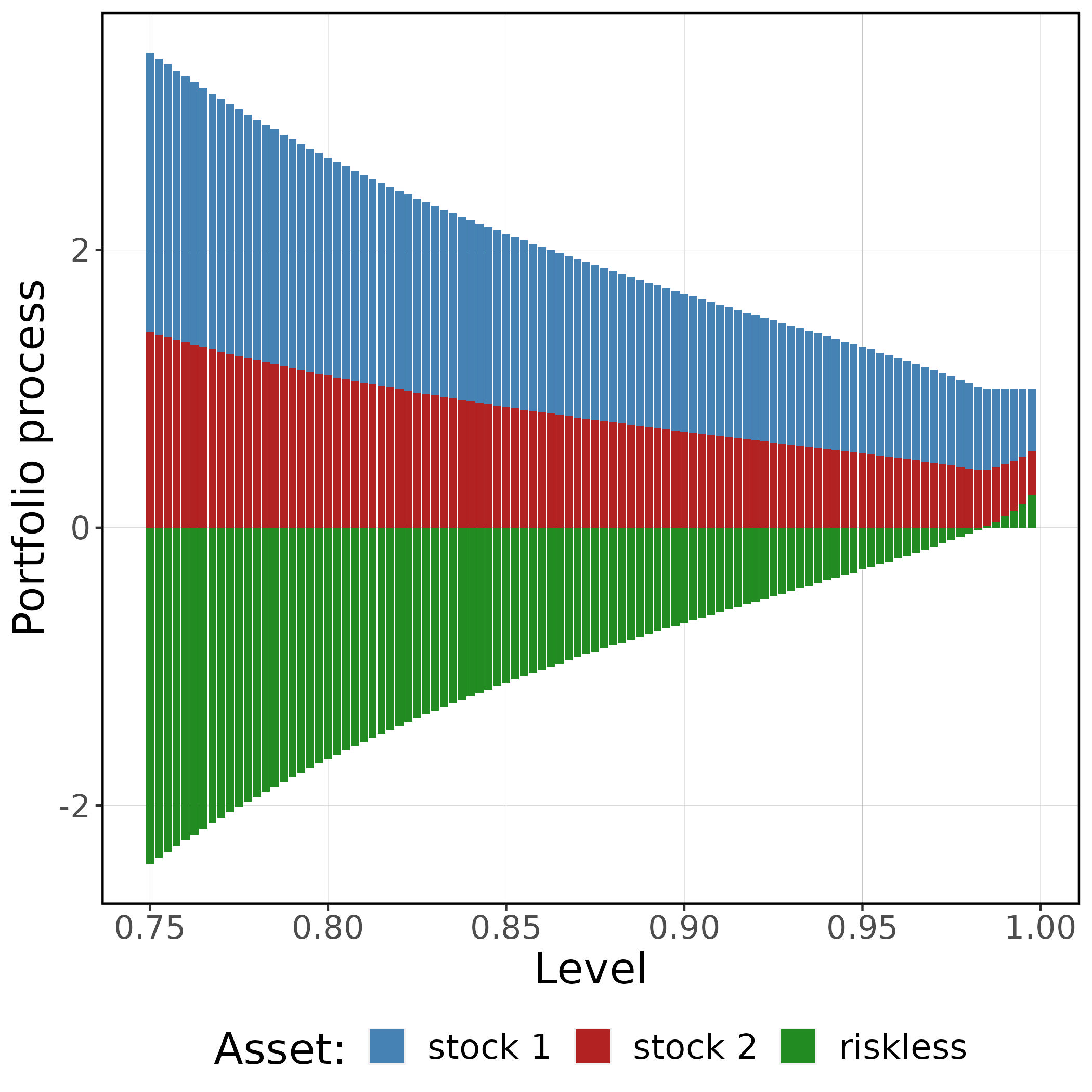}
    \end{subfigure}
    \hfill
    \begin{subfigure}[b]{0.48\textwidth}
        \includegraphics[width=\linewidth]{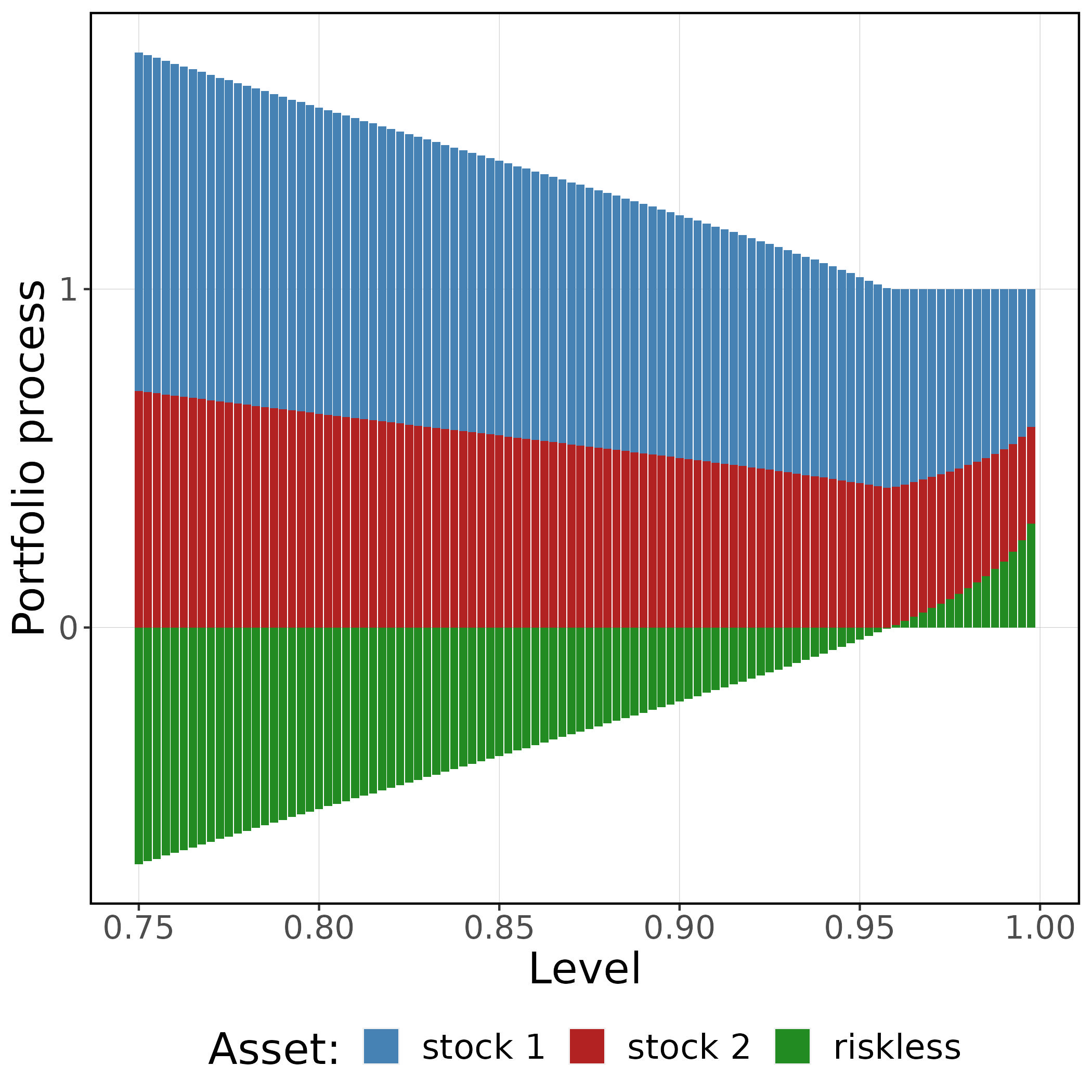}
    \end{subfigure}
    \captionsetup{font=footnotesize}
    \caption{\footnotesize Situation of $X\sim LN(\mu,\sigma^2)$ with $\mu = 1.5$ and $\sigma = 0.2$. \textit{LHS:} Optimal portfolio process for MARRM with VaR acceptance set. \textit{RHS:} Optimal portfolio process for MARRM with ARaR acceptance set.}
    \label{fig:var_arar_portfolio_processes_no_short_position}
\end{figure}

Finally, note that the factor $\Phi^{-1}(\lambda)$ increases if $\lambda$ does. Hence, the aforementioned influence of the term \label{eq:termObjectiveVaR} is amplified and the minimum of the function $f$ is shifted closer to the origin.  

\subsection{$L^\gamma$-norm and entropic acceptability criteria}\label{sec:entropic}

Here, we present another example, in which MARRM and MARM are significantly different. To do so, we use the classical Lebesgue norm to define an acceptability criterion for a  MARRMs. For a strictly positive loss $X$ and a market payoff $Z$ modeled as a positive random variable itself, the relative loss $X/Z$ is acceptable if and only if 
$$\left\lVert\tfrac{X}{Z}\right\rVert_{L^\gamma}=\E\left[\left|\tfrac{X}{Z}\right|^\gamma\right]^{1/\gamma}\leq 1.$$
Here, $\gamma>0$ is a risk aversion parameter we shall vary. 
This is especially important, since values like $\gamma\in\{1,2,3,4\}$ echo well-known statistics of the distribution of the relative loss random variable $\frac X Z$, namely mean, variance, skewness and kurtosis. 

Each $L^\gamma$-norm is a RRM. 
As mentioned in Example 3 in~\cite{Return}, the corresponding monetary risk measure is the entropic risk measure.
Hence, we now aim to compare the performance of the MARRM based on the $L^\gamma$-norm and the MARM based on the entropic risk measure. 
More precisely, the MARM deems a random variable $Y$ acceptable if $\log(\E[\exp(\gamma Y)])\leq 0$. 
The parameter $\gamma>0$ is the same as for the $L^\gamma$-norm. 
Speaking from the perspective of expected utility, where the entropic risk measure arises as certainty equivalent, $\gamma>0$ is the parameter of absolute risk aversion.

To ensure a finite expectation in case of the entropic risk measure, we use a light-tailed distribution for $X$, namely an exponential distribution with mean $\frac{1}{r}$ and $r>0$.
The acceptability criterion in case of the MARM for initial capital $x_0$ and a portfolio process $\pi$ is then given as $\log(\E[\exp(\gamma (X-X^{x_0,\pi}_T))])\leq 0$, which is equivalent to $\log\left(\tfrac{r}{r-\gamma}\right) + \log(\E[\exp(-\gamma X^{x_0,\pi}_T)])\leq 0.$
The summand $\log(\E[\exp(-\gamma X^{x_0,\pi}_T)])$ is obtained by Monte Carlo simulation. 
Further, we see that it is only meaningful for $r>\gamma$ and therefore, we set $r = 4$. This is a sufficiently large value to illustrate the effect of the approach for different values of the level $\gamma$.
On the left-hand side in Figure~\ref{fig:norm}, we plot  the MARRM and the MARM for different levels of $\gamma$.
The MARM function for values of $\gamma$ greater than $1.5$ looks strictly convex. 
The MARRM function looks linear. 
As for the ARaR/ES-case, we see on the right-hand side in Figure~\ref{fig:norm} that the difference between MARRM and MARM is significant. So, in this example it is important which risk measure an agent chooses, because the deviations for some levels are above $30\%$.

\begin{figure}
    \begin{subfigure}[b]{0.42\textwidth}
        \includegraphics[width=\linewidth]{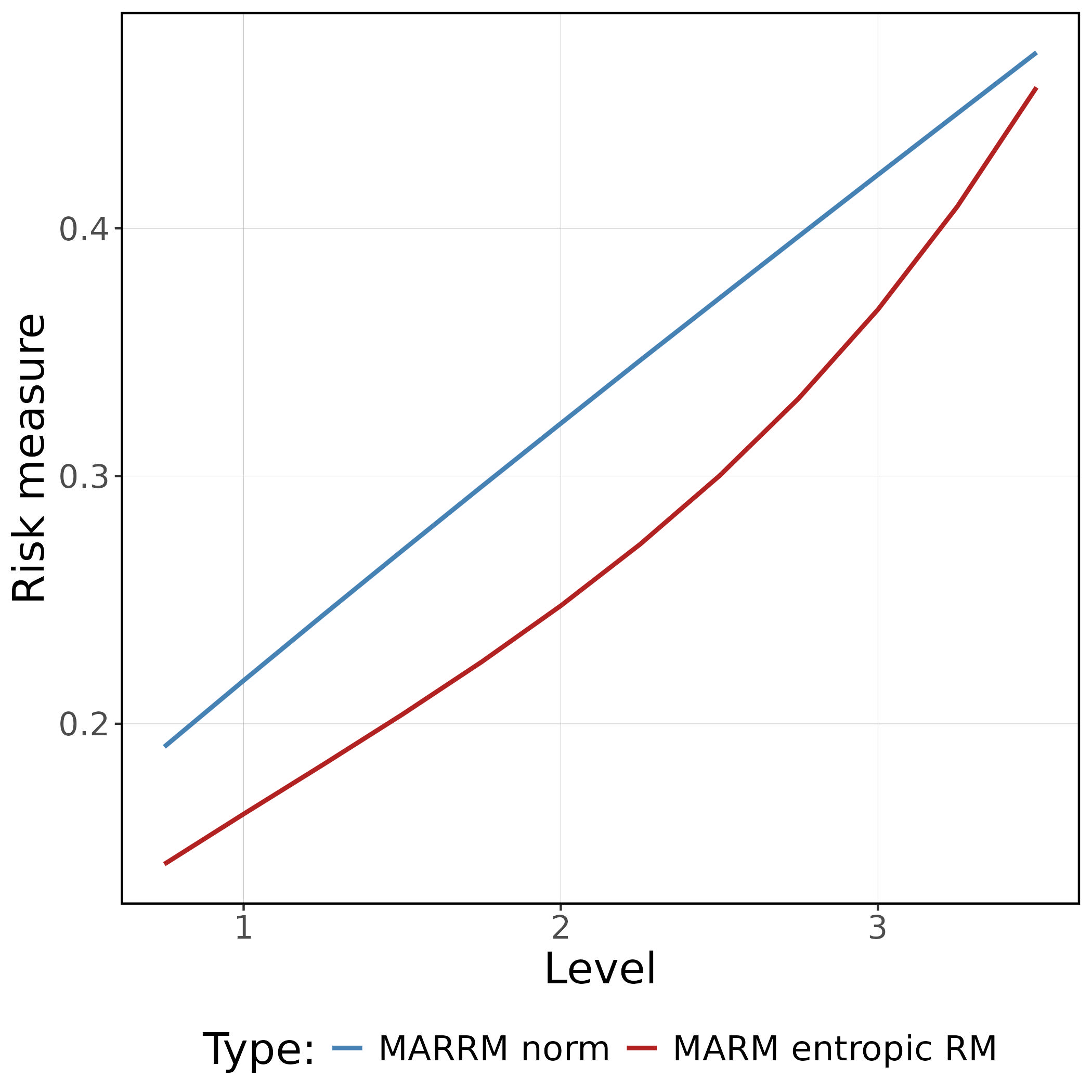}
    \end{subfigure}
    \hspace{1.1cm}
    \begin{subfigure}[b]{0.42\textwidth}
        \includegraphics[width=\linewidth]{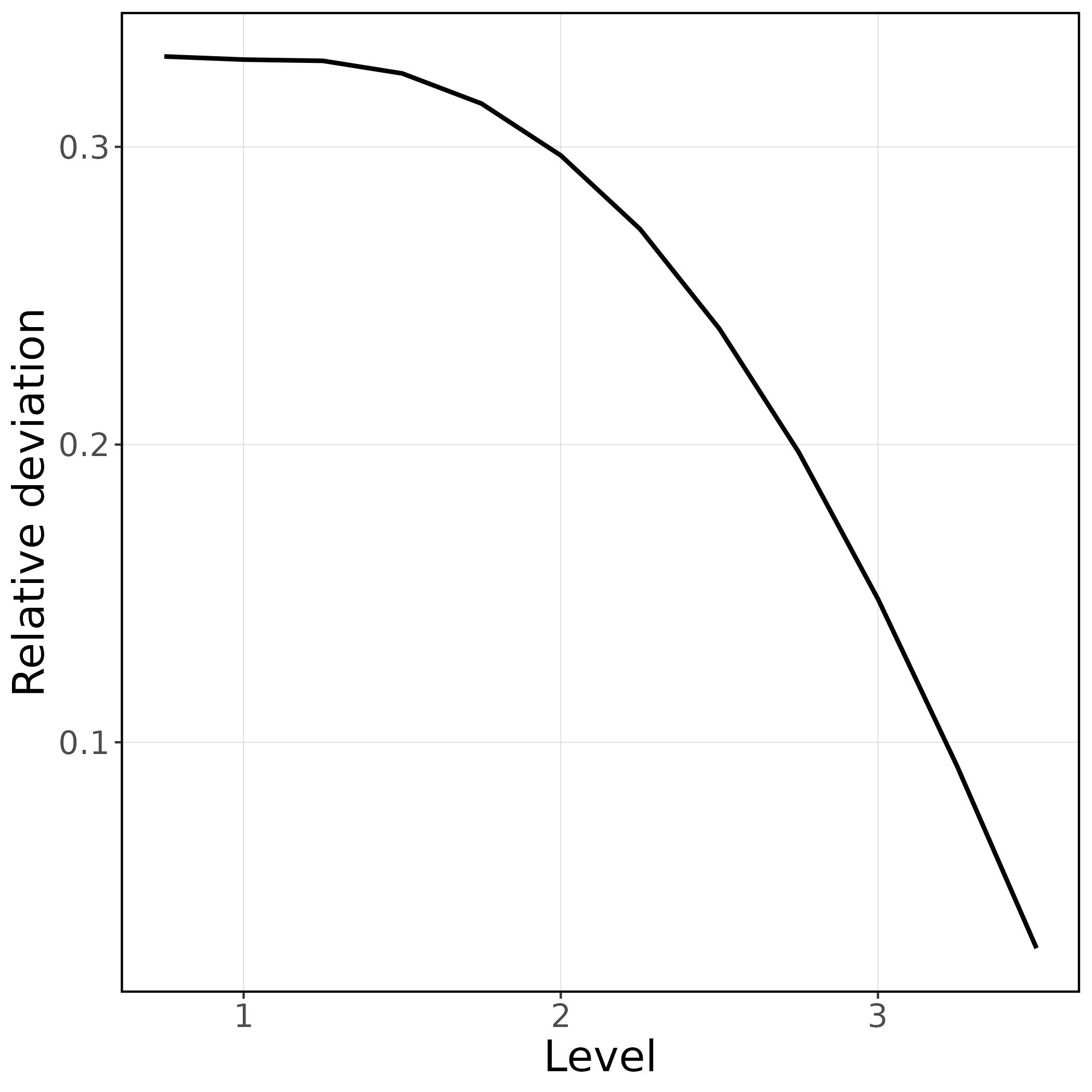}
    \end{subfigure}
    \captionsetup{font=footnotesize}
    \caption{\footnotesize \textit{LHS:} MARRM for $L^\gamma$-acceptability and MARM for acceptability based on the entropic risk measure. \textit{RHS:} Relative deviation between the two risk measures.}
    \label{fig:norm}
\end{figure}

For completeness, in Figure~\ref{fig:portfolio_norm}, we illustrate the underlying optimal portfolio processes. We see that the short position in the bank account seems to converge to zero if $\gamma\rightarrow r=4$.

\begin{figure}
    \begin{subfigure}[b]{0.42\textwidth}
        \includegraphics[width=\linewidth]{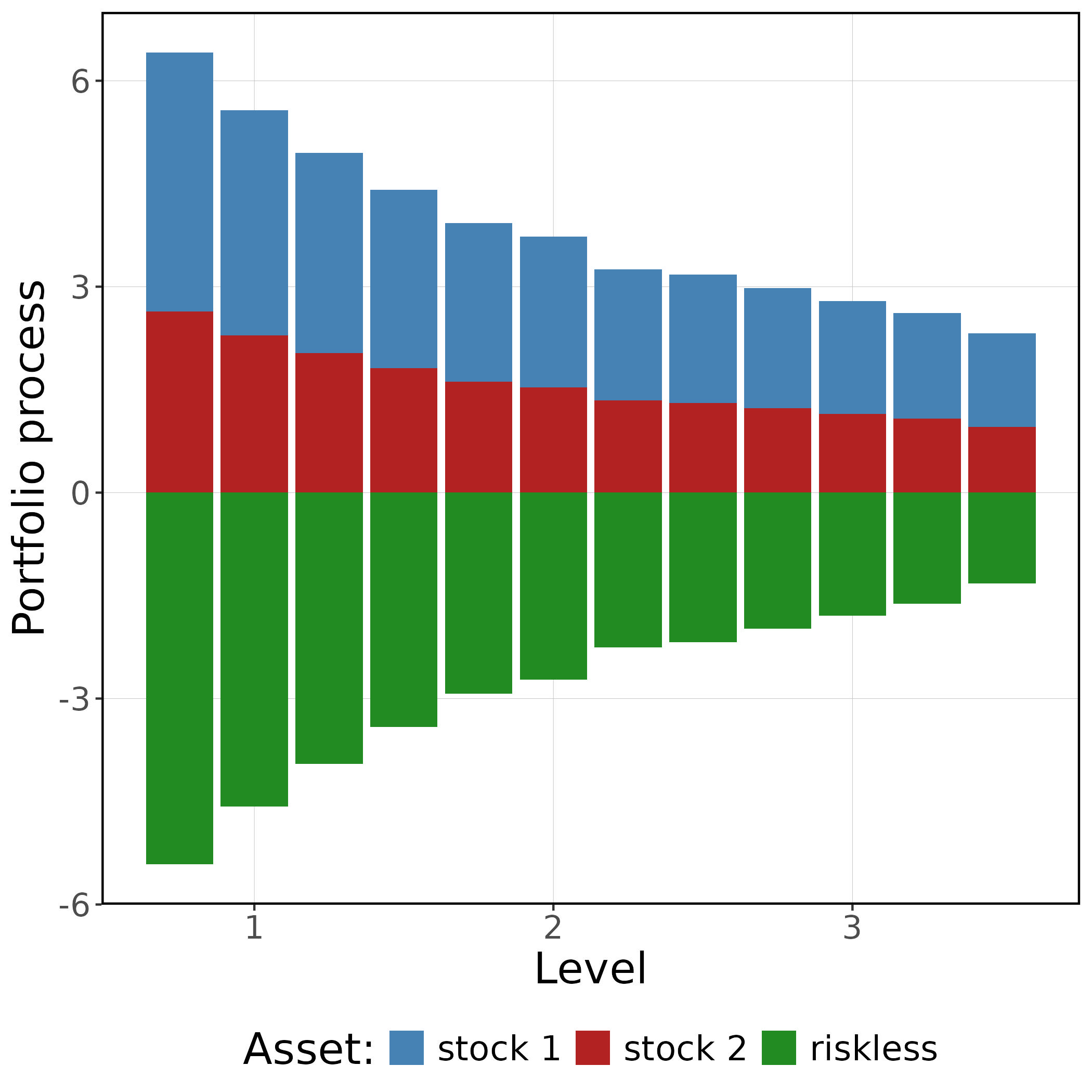}
        \captionsetup{font=footnotesize}\caption{\footnotesize MARRM for norm.}
    \end{subfigure}
    \hspace{1.1cm}
    \begin{subfigure}[b]{0.42\textwidth}
        \includegraphics[width=\linewidth]{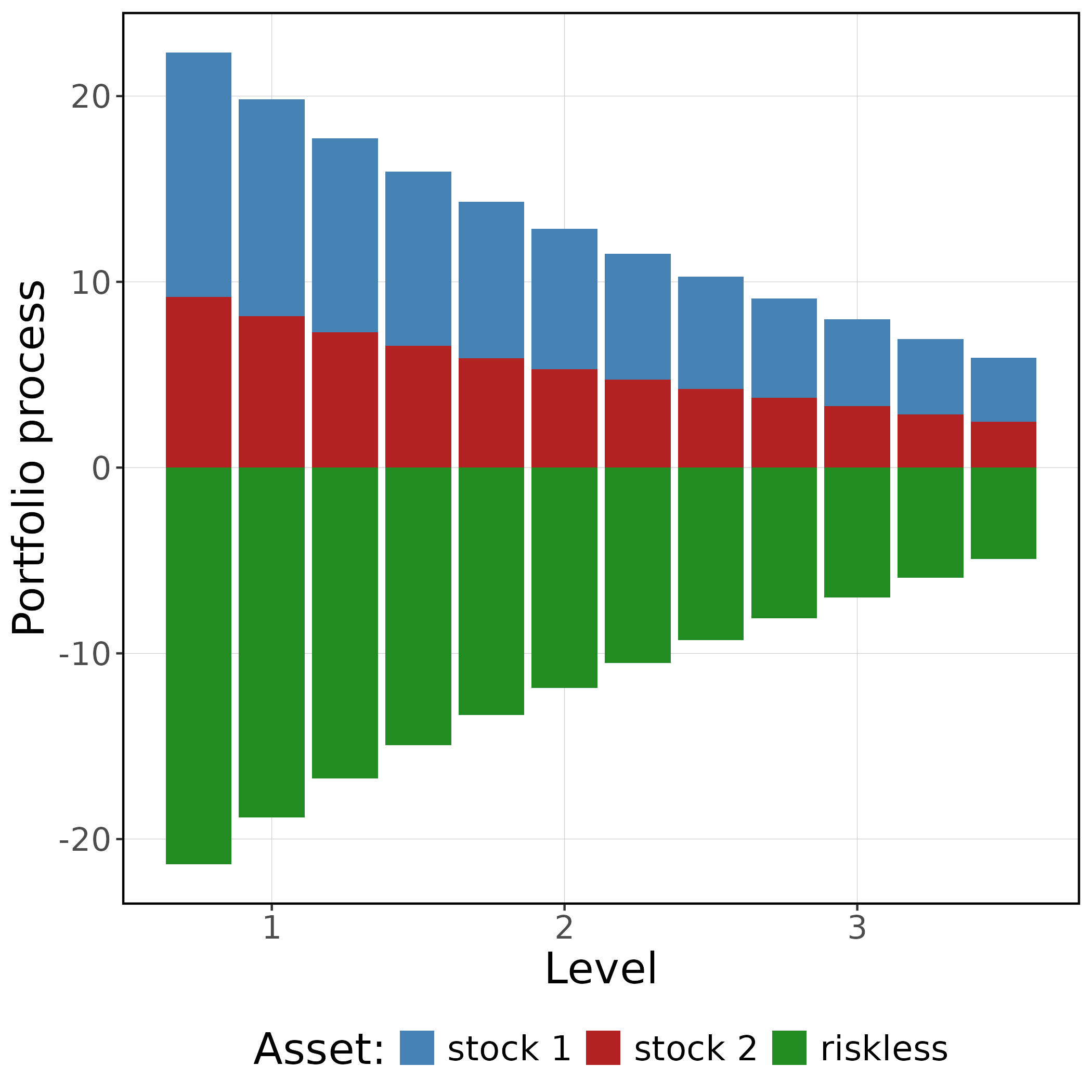}
        \captionsetup{font=footnotesize}\caption{\footnotesize MARM for entropic risk measure.}
    \end{subfigure}
    \captionsetup{font=footnotesize}
    \caption{\footnotesize Portfolio processes for different levels.}
    \label{fig:portfolio_norm}
\end{figure}

Concluding, we find a significant difference between MARRM and MARM in this example. Increasing the level shrinks the proportion invested in the bank account.

Comparing this with the empirical study in the manuscript, the VaR acceptability criteria suggest little difference between MARRM and MARM. However, the ARaR/ES and $L^{\gamma}$-norm/entropic criteria reveal that this conclusion is not generally valid, and the similarity between MARRM and MARM depends heavily on the chosen acceptability criteria.

\subsection{Financial market data}\label{sec:empiricalStudy}

In the empirical study, we used insurance losses which are part of the liabilities of the insurer in question. Now, we analyze the value of a MARRM for the assets of an insurer. Here, we focus on financial market models as it is done, for instance, in~\cite{chen_2019,chen_2018}. While these two contributions describe assets by a Black-Scholes type model, we shall adopt the approach of~\cite{mcneil_2000} to be able to calculate the MARRM dynamically over time. More precisely, we fit a classical time series model to real-world data (DAX) and calculate the MARRM for one-day-ahead market values.
	
The time series model that we are using is an AR(1)-GARCH(1,1) model for log-returns. The model description is as follows: We denote by $Y_t$ the index value at time $t$. Then, for the log-returns it holds that
\begin{align*}
\log\left(\tfrac{Y_t}{Y_{t-1}}\right) = \phi\log\left(\tfrac{Y_{t-1}}{Y_{t-2}}\right) + \sigma_t z_t\quad\text{and}\quad\sigma_t^2 = \alpha_0 + \alpha_1 (\sigma_{t-1}z_{t-1})^2 + \beta \sigma_{t-1}^2,
\end{align*} 
where $\phi, \alpha_0,\alpha_1,\beta\in \mathbb{R}$ are constant parameters and $\{z_t\}$ is a sequence of independent standard Gaussians. 
For further details we refer to~\cite{Ferenstein2004}. The parameters are calibrated via a pseudo-maximum-likelihood approach with the function \texttt{garchFit} from the \texttt{R}-package \texttt{fGarch}. As data for the calibration, we use the DAX index. 

We calibrate the AR(1)-GARCH(1,1) for two different time periods, illustrated by the gray and yellow boxes on the left-hand side in Figure~\ref{fig:logReturnsDAX}. The calculation of risk measures starts then always from the end of the corresponding time period used for calibration, see Figure~\ref{fig:forecast_DAX}. Both time periods consist of 1\,000 log-returns. The reason to choose two different time periods is the behavior of the log-returns after these intervals. In case of the first interval, the log-returns after this time period behave moderately. Instead, after the second interval, we have huge fluctuations of the log-returns due to the COVID crisis.

\begin{figure}[ht]
    \begin{subfigure}[b]{0.48\textwidth}
        \includegraphics[width=\linewidth]{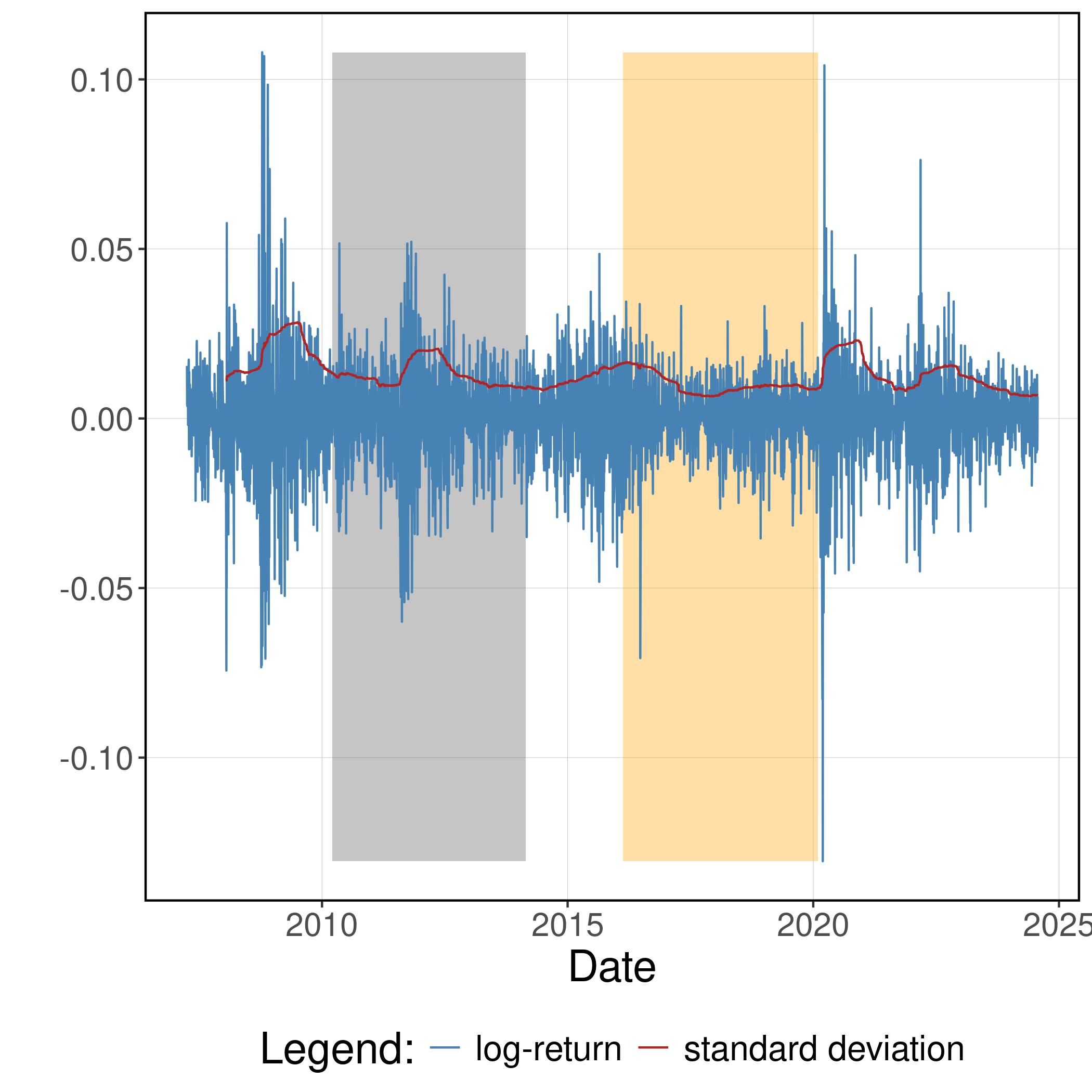}
        \captionsetup{font=footnotesize}
        \caption{\footnotesize DAX log-returns.}
    \end{subfigure}
    \hfill
    \begin{subfigure}[b]{0.48\textwidth}
        \includegraphics[width=\linewidth]{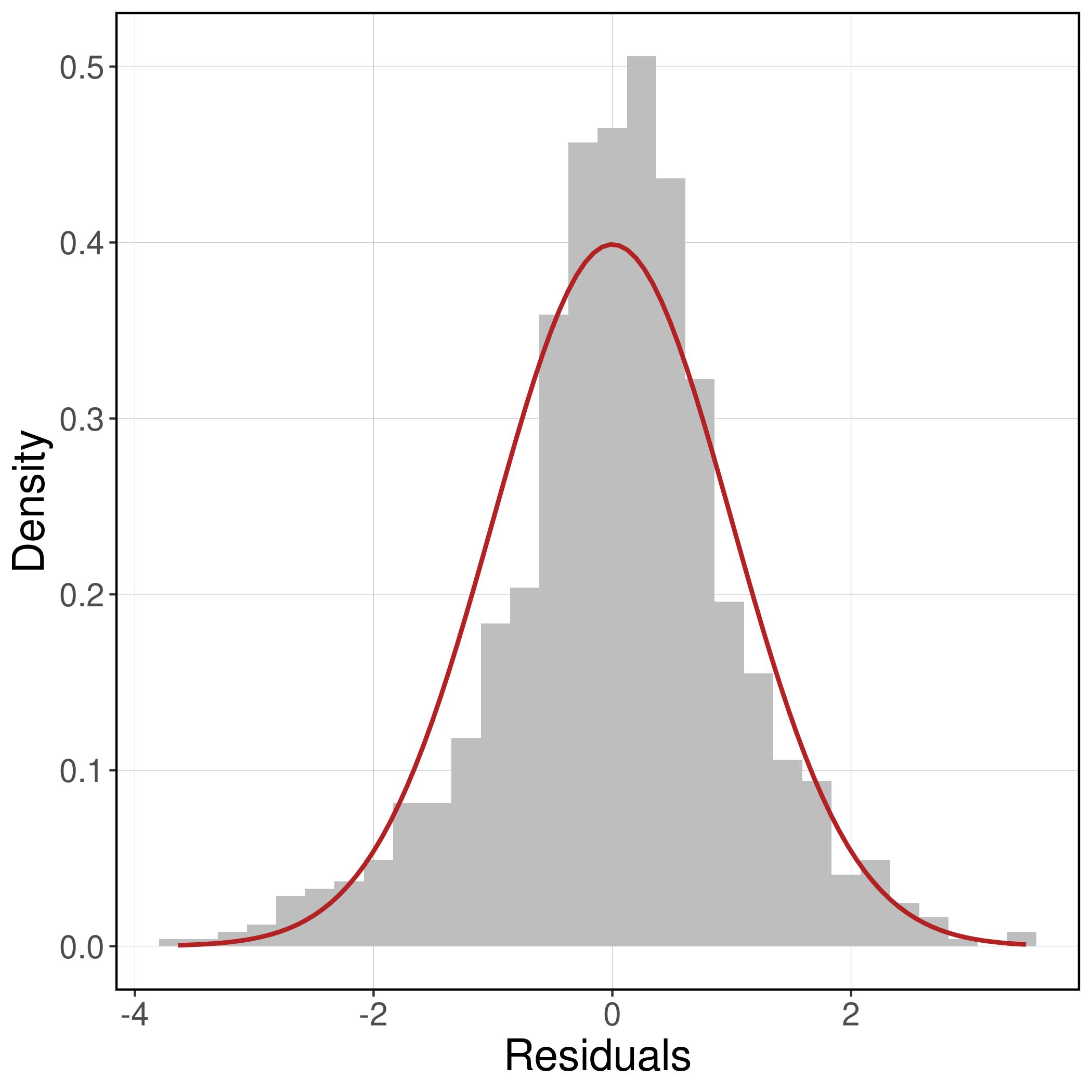}
        \captionsetup{font=footnotesize}
        \caption{\footnotesize Histogram of residuals.}
    \end{subfigure}
    \captionsetup{font=footnotesize}
    \caption{\footnotesize \textit{LHS:} Log-returns for DAX. The red line is the estimated (rolling) standard deviation based on 200 log-returns. The colored boxes are the log-returns used to calibrate the AR(1)-GARCH(1,1) model. \textit{RHS:} Histogram of residuals of the calibrated AR(1)-GARCH(1,1) with respect to log-returns from the gray box. The red line is the density function of a standard Gaussian.}
    \label{fig:logReturnsDAX}
\end{figure}

The calibrated parameters are given in Table~\ref{tab:calibrated_parametersDAX}. The corresponding residuals for the first time interval are illustrated on the right-hand side in Figure~\ref{fig:logReturnsDAX}.

\begin{table}[htb]
    
    \centering
    \begin{tabular}{l|r|r|r|r}
        {Time interval (YYYY/MM/DD)}   & {$\hat{\phi}$} & {$\hat{\alpha}_0$} & {$\hat{\alpha}_1$} & {$\hat{\beta}$}\\
        \hline
        2010/03/29 - 2014/03/28 & $0.0340$ & $2.8718\cdot 10^{-6}$ & $0.0746$ & $0.9067$\\
        2016/03/31 - 2020/03/26 & $-0.0003$ & $6.4917\cdot 10^{-6}$ & $0.1069$ & $0.8228$
    \end{tabular}	
    \captionsetup{font=footnotesize}
    \caption{\footnotesize Calibrated model parameters of the log-returns of the DAX.}
    \label{tab:calibrated_parametersDAX}
\end{table}

The calibrated model allows us to calculate a RRM for $Y_{t}$ based on the information at time $t-1$. 
Note that in this situation, we have that 
\begin{align*}
    Y_{t}|_{Y_{t-1},Y_{t-2}}\sim \text{LN}\left((1+\phi)\log(Y_{t-1})-\phi\log(Y_{t-2}), \sigma_t^2\right). 
\end{align*} 
Initial values for $\sigma_t$ are obtained from the \texttt{R}-function \texttt{volatility} applied to the output of \texttt{fGarch}. For the gray box this gives $0.0088$ and for the yellow box it is $0.0108$.

Our aim is to compare the RRM with the MARRM based on the Black-Scholes model already used in the empirical study, using the same parameter values as in there. 
Motivated by the results of the empirical case study, we omit the VaR case and focus solely on the case of the ARaR.  

We plot the following relative differences between RRM, MARRM and DAX on the left-hand side in Figure~\ref{fig:forecast_DAX}:
\begin{center}
$\frac{\text{RRM}-\text{index value}}{\text{index value}}$ (red),\quad$\frac{\text{MARRM} -\text{index value}}{\text{index value}}$ (blue),\quad$\frac{\text{RRM} - \text{MARRM}}{\text{MARRM}}$ (yellow).
\end{center}

\begin{figure}[ht]
    \begin{subfigure}[b]{0.48\textwidth}
        \includegraphics[width=\linewidth]{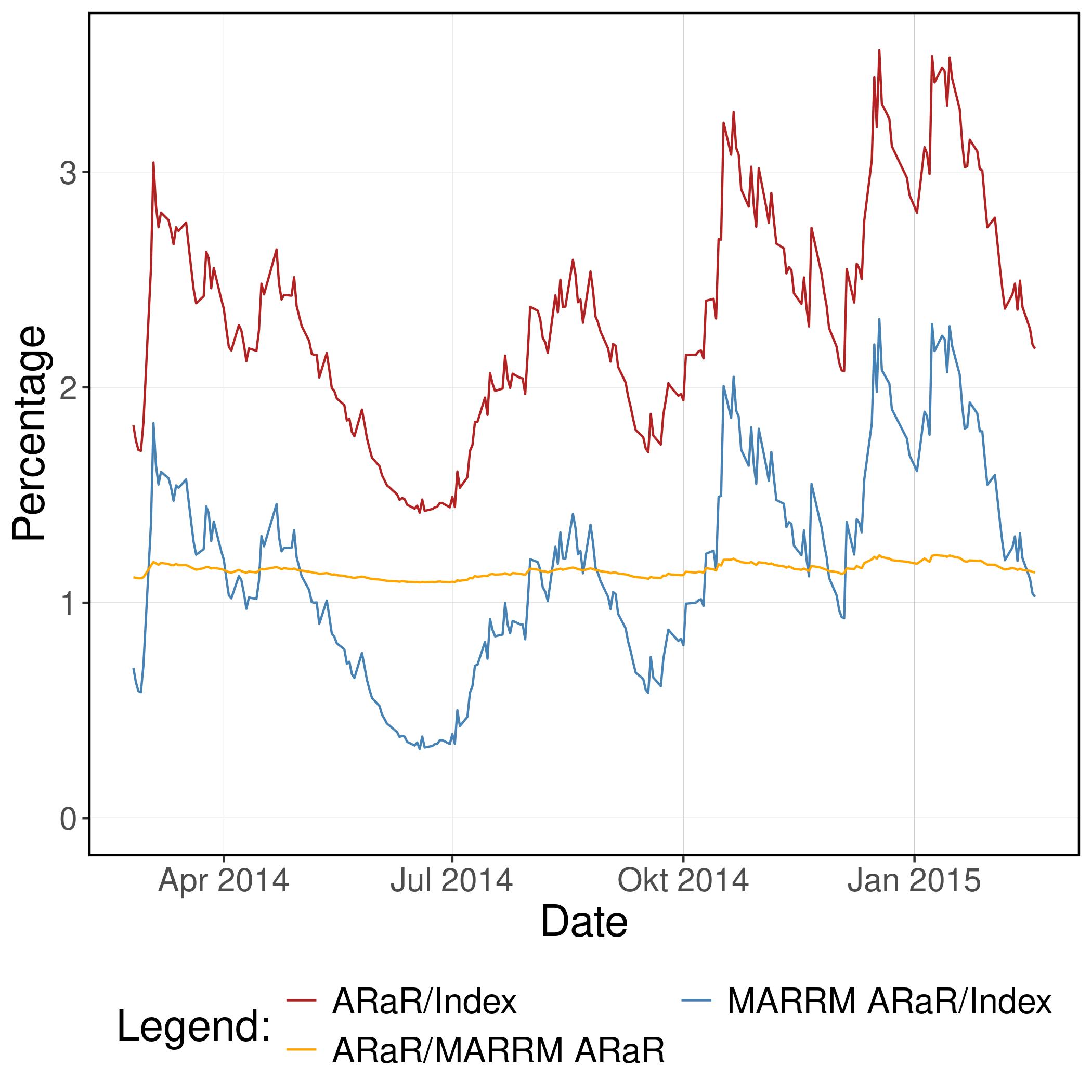}
    \end{subfigure}
    \hfill
    \begin{subfigure}[b]{0.48\textwidth}
        \includegraphics[width=\linewidth]{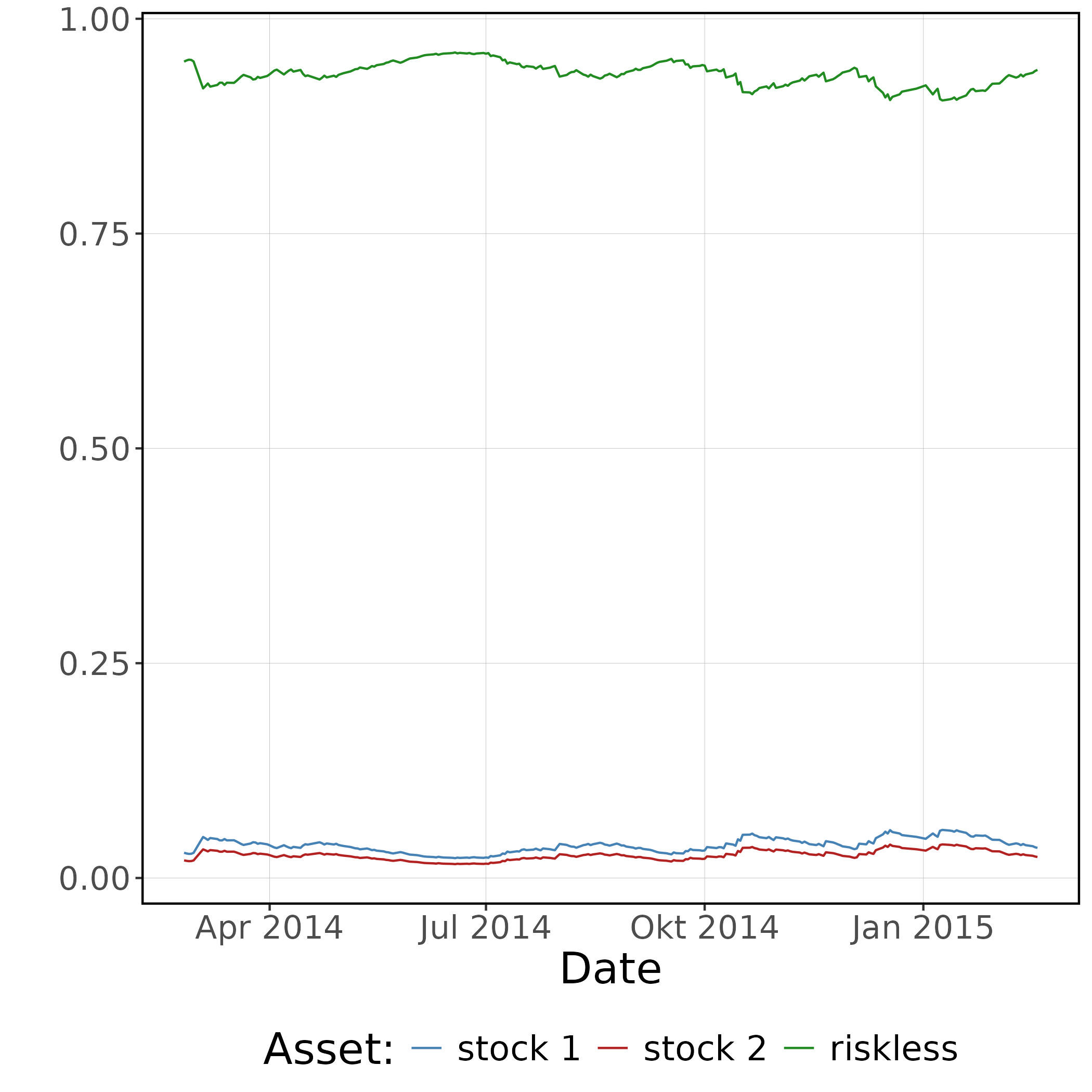}
    \end{subfigure}
    \begin{subfigure}[b]{0.48\textwidth}
        \includegraphics[width=\linewidth]{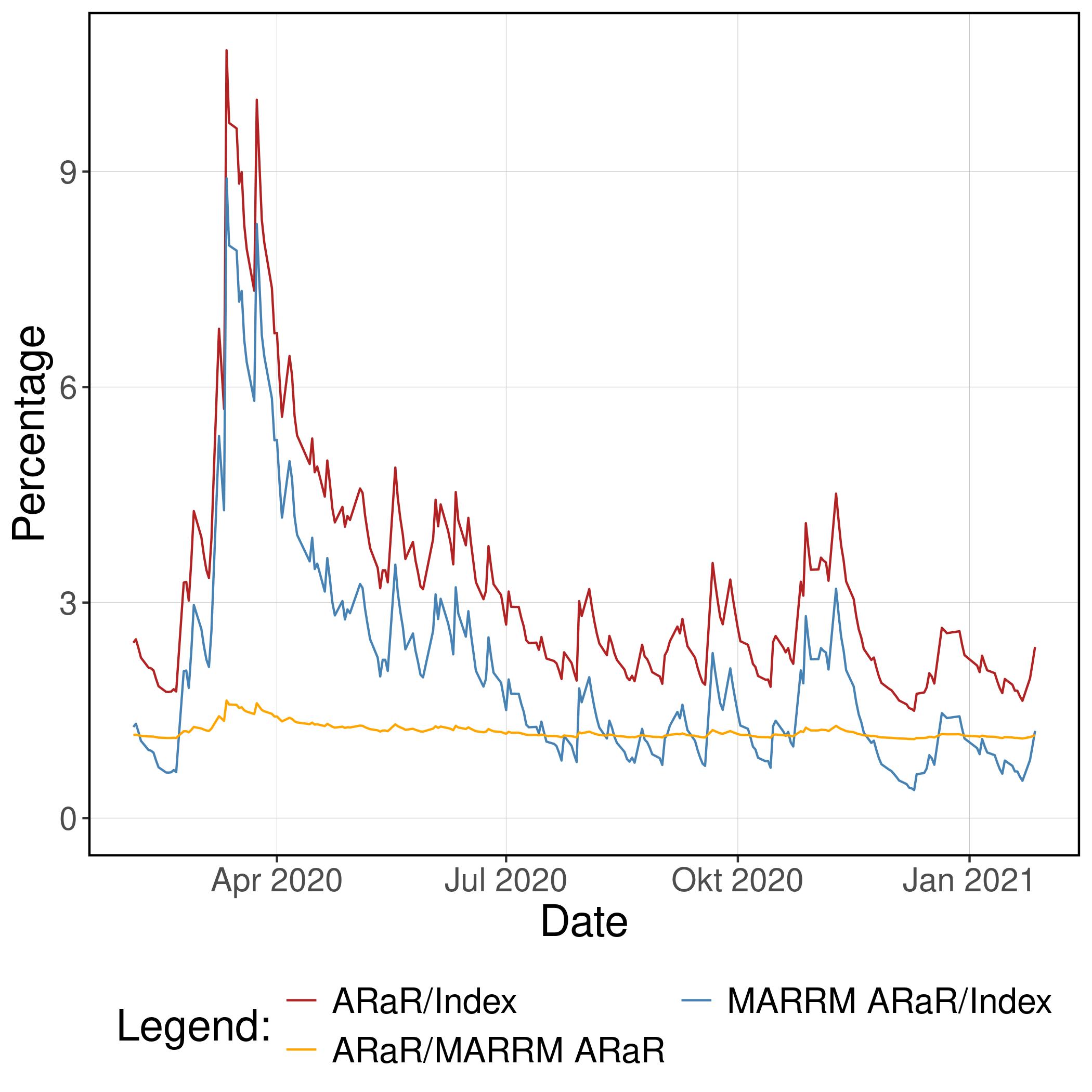}
    \end{subfigure}
    \hfill
    \begin{subfigure}[b]{0.48\textwidth}
        \includegraphics[width=\linewidth]{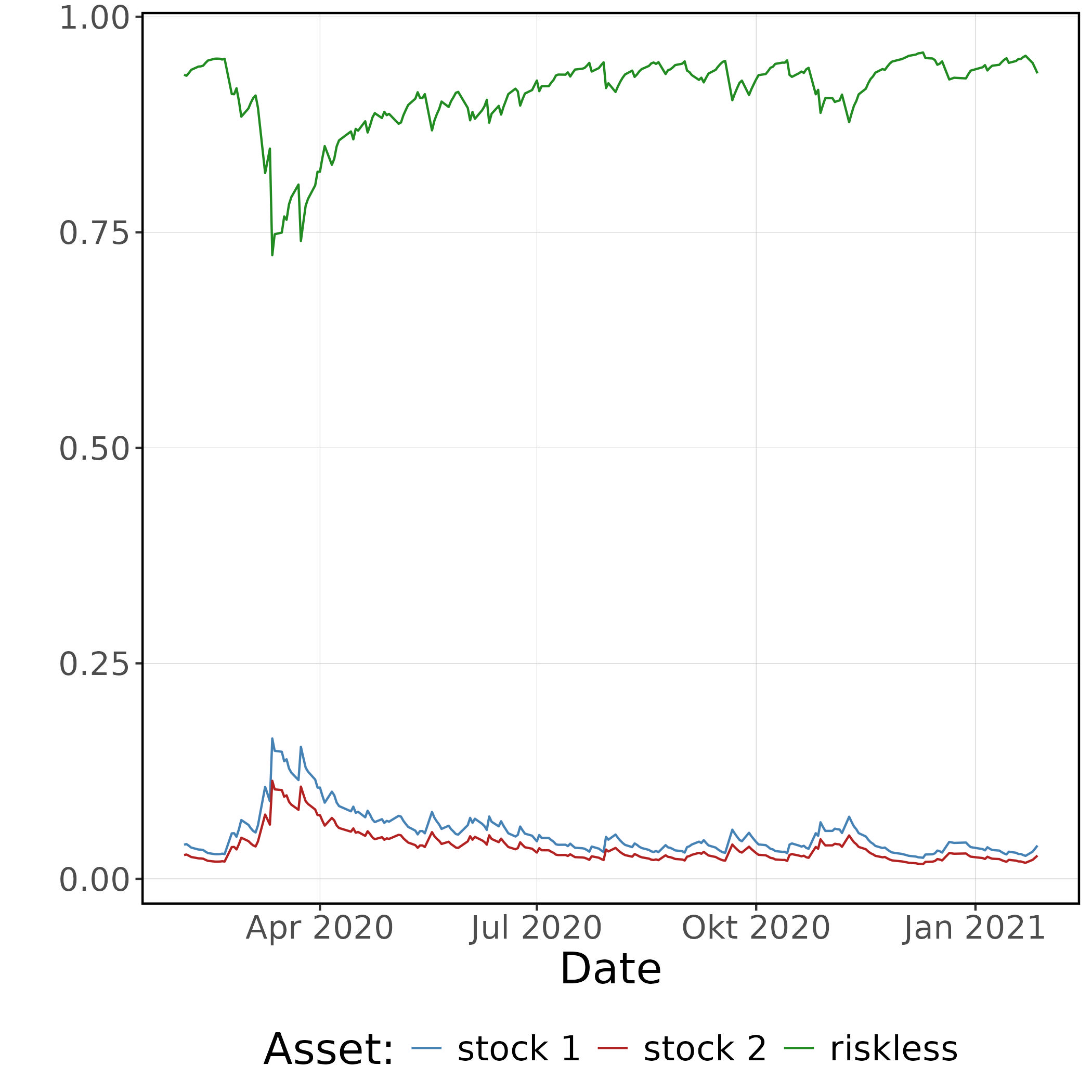}
    \end{subfigure}
    \captionsetup{font=footnotesize}
    \caption{\footnotesize\textit{LHS:} Relative differences of risk measures and index, as well as relative differences between RRM and MARRM for DAX. \textit{RHS:} Underlying portfolio processes.}
    \label{fig:forecast_DAX}
\end{figure}

We see that the relative differences between risk measures and index values (red and blue lines) are quite large at the beginning of the COVID crisis. In the aftermath, their behaviour is much more similar to the relative differences observed between April 2014 and April 2015.

At first glance, the relative difference between RRM and MARRM (yellow line) looks nearly constant. But, by a closer look, during the COVID crisis there is a peak of $1.64\%$ at April 12, 2020. In comparison to $1.12\%$ at February 21, 2020, this is $0.42\%$ larger, which is a significant increase by noting that we only perform daily forecasts.

This is also visible in the underlying portfolio process, i.e.,~the percentages of the capital invested in specific assets, see the right-hand side in Figure~\ref{fig:forecast_DAX}. During the COVID crisis, the MARRM reduces the investment in the riskless asset, which means that the hedging strategy is given by a more diversified portfolio between the riskless asset and the two risky stocks.

The main conclusion from Figure~\ref{fig:forecast_DAX} is that the relative difference between the RRM and MARRM is nearly constant in times without crisis. In times of crisis, the larger security space of the MARRM leads to a better diversified hedging portfolio, which in turn gives us a stronger relative reduction compared to the RRM than in times without crisis. 

\end{document}